\documentclass{article}

\usepackage{authblk}
\title{Engression: Extrapolation through the Lens of Distributional~Regression}
\author{Xinwei Shen and Nicolai Meinshausen
\\ \small 
\textit{Seminar f\"ur Statistik, ETH Z\"urich, Switzerland}
 \\ \vspace{0.1cm}\small 
 \texttt{\{xinwei.shen,meinshausen\}@stat.math.ethz.ch}
}
\date{}

\usepackage{amsmath, amssymb, mathrsfs, bm,amsthm} 
\usepackage{geometry, extarrows, xcolor, graphicx, enumitem, booktabs, subfigure}
\usepackage{makecell}
\usepackage{array}
 
 \usepackage{multirow}

\usepackage{natbib}
\usepackage[colorlinks,citecolor=blue]{hyperref}
\usepackage{array,float}
\usepackage{comment}
\usepackage[linesnumbered,ruled,vlined]{algorithm2e}
\SetKwInput{KwInput}{Input}
\SetKwInput{KwReturn}{Return}

\SetCommentSty{mycommfont}
\usepackage{fullpage}

\usepackage{graphbox}

\DeclareMathOperator*{\argmin}{argmin}

\def\ind{\mathbf{1}}
\def\bbE{\mathbb{E}}
\def\bbR{\mathbb{R}}
\def\bbP{\mathbb{P}}

\def\cB{\mathcal{B}}
\def\cD{\mathcal{D}}

\def\cF{\mathcal{F}}
\def\cG{\mathcal{G}}
\def\cH{\mathcal{H}}
\def\cL{\mathcal{L}}
\def\cM{\mathcal{M}}
\def\cN{\mathcal{N}}
\def\cU{\mathcal{U}}
\def\cO{\mathcal{O}}
\def\cP{\mathcal{P}}
\def\cX{\mathcal{X}}
\def\cY{\mathcal{Y}}
\def\cZ{\mathcal{Z}}
\def\rmE{\mathrm{E}}

\def\tx{\tilde{x}}
\def\kl{\mathrm{KL}}
\def\pto{\overset{p}{\to}}

\def\xm{x_{\max}}
\def\xmin{x_{\min}}

\def\etam{\eta_{\max}}

\def\crps{\mathrm{CRPS}}
\def\es{\mathrm{ES}}
\def\ed{\mathrm{ED}}
\def\gen{\textsl{g}}
\def\preanm{\cM_{\mathrm{pre}}}
\def\cramer{\mathrm{CD}}
\def\ptr{P_{\mathrm{tr}}}
\def\var{\mathrm{Var}}

\theoremstyle{plain}
\newtheorem{theorem}{Theorem}
\newtheorem{proposition}{Proposition}
\newtheorem{lemma}{Lemma}
\newtheorem{corollary}{Corollary}

\newtheorem{definition}{Definition}

\theoremstyle{remark}
\newtheorem{example}{Example}

\newcommand\revise[1]{{\color{black}{#1}}}

\begin{document}
\maketitle

\begin{abstract}
Distributional regression aims to estimate the full conditional distribution of a target variable, given covariates.  Popular methods include linear and tree-ensemble based quantile regression. We propose a neural network-based distributional regression methodology called `engression'. An engression model is generative in the sense that we can sample from the fitted conditional distribution and is also suitable for high-dimensional outcomes. Furthermore, we find that modelling the conditional distribution on training data can constrain the fitted function outside of the training support, which offers a new perspective to the challenging extrapolation problem in nonlinear regression. In particular, for `pre-additive noise' models, where noise is added to the covariates before applying a nonlinear transformation, we show that engression can successfully perform extrapolation under some assumptions such as monotonicity, whereas traditional regression approaches such as least-squares or quantile regression fall short under the same assumptions. Our empirical results, from both simulated and real data, validate the effectiveness of the engression method and indicate that the pre-additive noise model is typically suitable for many real-world scenarios. The software implementations of engression are available in both R and Python.
\end{abstract}

\section{Introduction}\label{sec:intro}
\revise{
We consider regression with a response variable $Y \in \mathbb{R}$, using a set of predictors or covariates $X \in \mathbb{R}^d$. For simplicity, we focus primarily on the case of a univariate target but our methodology generalises naturally to a multivariate response. A common regression task is then to estimate the conditional mean $ \bbE(Y|X=x)$. 
For many applications, however, we are interested in not just the conditional mean but rather the full conditional distribution of $ Y|X=x $. A possible approach for the distributional regression task is to model each quantile of the distribution separately using quantile regression \citep{koenker1978regression, koenker2005quantile}. Fitting a separate model for each quantile, however, is a computational burden and leads to issues such as quantile crossing. While the latter can be rectified  \citep{he1997quantile,chernozhukov2010quantile} and extensions to multivariate responses are non-trivial but possible
\citep{carlier2017vector}, an alternative approach is generative-type models with flexible model classes such as neural networks, having emerged dramatically recently in machine learning \citep{goodfellow2014}. Generative models characterise the distribution via a generating process instead of explicitly modelling multiple quantiles and allow to sample from the conditional distribution $Y|X=x$, which then enables sampling-based estimation for the conditional mean and quantiles. 

We propose a neural network-based distributional regression method called `engression', where the distributional fit is assessed with the energy score \citep{gneiting2007strictly}. Engression is a simple yet generic generative approach that can be applied to various regression tasks, including point prediction for the conditional mean and quantiles, constructing prediction intervals, and sampling from the conditional distribution. It is shown to be effective for many real data and suitable for modelling high-dimensional responses. Generally speaking, we envision engression as a flexible technique for statistical inference problems that involve (conditional) distribution estimation, as well as an interesting addition to the current nonlinear regression toolbox for practitioners. To facilitate this, we provide software of engression in R and Python packages \texttt{engression}; see Section~\ref{sec:software} for details.
}

\revise{
Equipped with this distributional regression tool, we then focus on a new perspective that engression offers to the extrapolation problem in nonlinear regression. In many statistical and machine learning applications, it is common to encounter data points that go beyond the support of the training data. Investigating the extrapolation properties of models is thus important in practice. 
}
Extrapolation for linear models is relatively straightforward, given that a linear function is generally uniquely defined by its values within the support of the training data. However, nonlinear models pose a much larger challenge. 
Tree ensemble methods such as Random Forest~\citep{breiman2001random} and boosted regression trees~\citep{friedman2001greedy,buhlmann2003boosting,buhlmann2007boosting} typically produce a constant prediction during extrapolation, as they invariably apply the nearest corresponding feature value from the training data when faced with a feature located outside the training support. 
Neural networks (NNs) on the other hand often \revise{exhibit} uncontrollable behaviour. For example, a recent study by \citet{dong2023first} revealed the existence of two-layer neural networks that align perfectly on one distribution, yet display starkly different behaviour on an expanded distribution that extends beyond the initial support. 

\begin{figure}
\begin{tabular}{@{}ccc@{}}
	\includegraphics[width=0.31\textwidth]{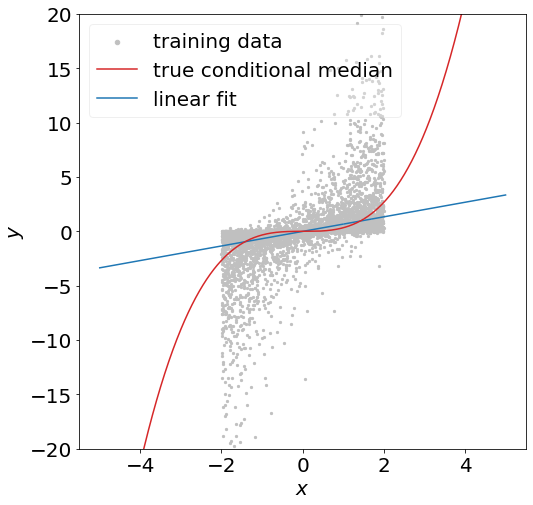} &
	\includegraphics[width=0.31\textwidth]{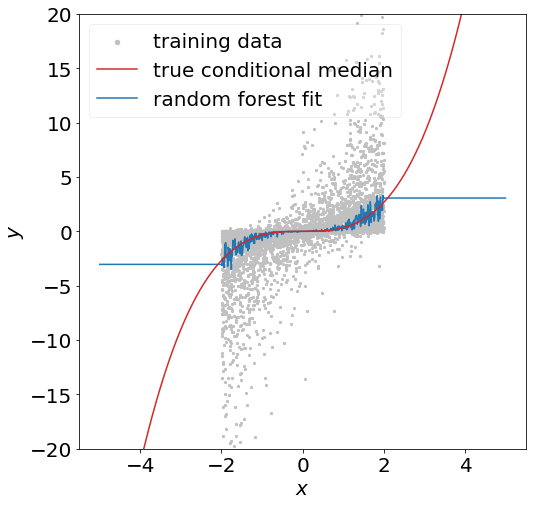} & 
	\includegraphics[width=0.32\textwidth]{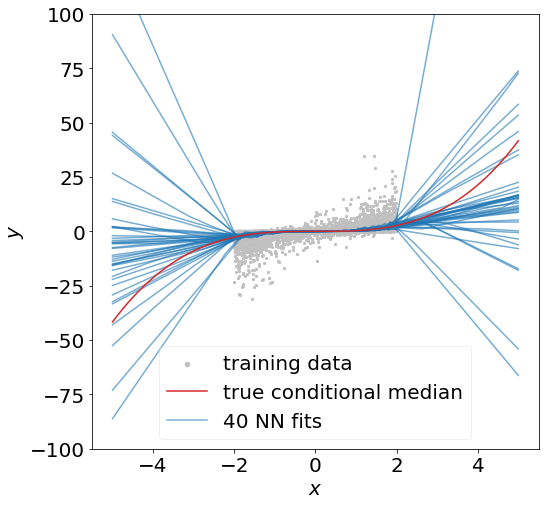} \\
	{(a) linear model} & {(b) tree-based model} & {(c) neural network}
\end{tabular}\vspace{-0.1in}
\caption{Examples of nonlinear extrapolation behaviour for different model classes. The true (conditional median) function is a cubic function. We use linear quantile regression, quantile regression forest, and quantile regression with a two-layer neural network and 100 hidden neurons to fit the training data supported on $[-2,2]$, and evaluate the models on a wider range of test data. }
\label{fig:illus_fail}
\end{figure}

\revise{Fig.~\ref{fig:illus_fail} illustrates the different behaviours of existing methods when extrapolating beyond the training support.}
Linear models tend to ignore the nonlinearity, offering a globally linear fit. Tree-based models default to providing constant predictions beyond the support boundaries. Neural networks fit nearly flawlessly the training data, yet display an arbitrary behaviour on data points located outside the support.
\revise{Note that the respective properties of tree-based models and neural networks for extrapolation may not necessarily be problematic but could be desirable depending on the application. For example, the behaviour of tree ensembles ensures that predicted values never go outside of the support of observed values of the target variable. The high variability of neural networks for extrapolation tasks could on the other hand be exploited to provide an uncertainty estimate if we look at an ensemble of networks fitted for subsamples of the data or with different initial weights~\citep{abe2022deep}. See also related approaches in Bayesian Machine Learning~\citep{snoek2012practical,osband2024epistemic}.  }

\revise{In this work, we study the extrapolation property regarding  whether the target of interest (either the conditional mean or conditional quantile functions) can be identified beyond the training support of the covariate variable. This question aligns with a substantial body of literature; see Section~\ref{sec:related_extrap} for a detailed review. We find that under suitable assumptions, distributional regression can benefit extrapolation, intuitively because the full conditional distribution (in contrast to just the conditional mean) may contain more information about the behaviour of the nonlinear functions of interest beyond the training support. Thus, modelling the conditional distribution exploits the extra information to constrain the fitted function outside of the bounded support.

In addition to a distributional estimation method, the modelling choice also plays a crucial role for successful extrapolation. For simplicity, here we illustrate the univariate case, while a more general setting with multivariate $X$ is described in Section~\ref{sec:anm}. Conventionally, a regression model assumes that noise $\eta$ is added after applying a nonlinear transformation $g$ to the covariates~$X$, i.e.
\begin{equation}\label{eq:postanm_uni}
	Y=g(X)+\eta.
\end{equation}
We refer to \eqref{eq:postanm_uni} as the post-additive noise model (post-ANM). 
Contrarily, in what we call the \emph{pre-additive noise model (pre-ANM)}, the noise is added directly to $X$ prior to any nonlinear transformation, i.e.
\begin{equation}\label{eq:preanm_uni}
	Y=g(X+\eta).
\end{equation}
We propose the pre-ANM as an alternative modelling choice to the often-used post-ANM, especially in the context of extrapolation. Intuitively, the pre-additive noise in the underlying data distribution can reveal some information about the true function outside of the support, which can then be exploited by engression to enable data-driven extrapolation. The advantages of pre-ANMs are formally justified. Our numerical results also indicate that pre-ANMs are typically suitable for many real data sets. 
}

\subsection{Structure}
The engression method is proposed in Section~\ref{sec:method}. 
Section~\ref{sec:theory_extra} presents the theoretical formulation and guarantees for extrapolation at the population level. We start by introducing the concept of extrapolation uncertainty as a measure of the degree of (dis)agreement between two models beyond the training support. We propose the term `extrapolability', which measures the capability of a model to extrapolate. Importantly, we consider extrapolability both functionally and distributionally. While previous studies have focused on functional extrapolability, we argue that distributional extrapolability is a more desirable notion. It can be achieved with much weaker assumptions about the function class. As an exercise, we demonstrate the extrapolability of the pre-ANM class of models, suggesting a possible advantage of pre-ANMs over post-ANMs for extrapolation tasks.

Furthermore, we derive a few simple theoretical results about engression. Specifically, we establish the extrapolability of engression under modest conditions of a strictly monotone function class, a feat that cannot be achieved by traditional methods such as least-squares regression and quantile regression. To the best of our knowledge, engression is the first methodology that can verifiably extrapolate beyond the support for nonlinear monotone functions (without the need to assume Lipschitz continuity). We also provide a quantification of the advantage of engression by examining the discrepancies between the extrapolation uncertainty of engression and that of benchmark methods, suggesting that engression consistently delivers an advantage. Interestingly, we demonstrate that engression can benefit from a high noise level: larger noise levels actually assist with extrapolation by employing engression. While much of our theoretical analysis is geared toward the univariate case, which already poses a significant challenge to extrapolation, our method can readily handle multivariate predictors and response variables.

Section~\ref{sec:finite_sample} is devoted to the finite-sample analysis of the engression estimator. We derive the finite-sample bounds for parameter estimation and prediction for the conditional mean or quantiles, especially outside the support, under specific settings of quadratic models. While traditional regression methods produce arbitrarily large errors outside the support, we show that engression attains consistency for both in-support and out-of-support data. Furthermore, when the pre-additive noise assumption is violated, the model yield by engression defaults to a linear fit that provides a useful baseline performance in comparison to the arbitrarily large extrapolation error from regression. 
In addition, we show the consistency for general monotone and Lipschitz functions. 

 More discussion on the relation to existing literature is in Section~\ref{sec:related}. In Section~\ref{sec:empirical}, we present comprehensive empirical studies using both simulated data and an assortment of real-world data sets. Our studies cover a range of regression tasks, including univariate and multivariate prediction, as well as prediction intervals. The promising results lend support to our theory and indicate the broad applicability and versatility of the engression method.

\subsection{Software}\label{sec:software}
The software implementations of engression are available in both R and Python packages \texttt{engression} with the source code at \url{https://github.com/xwshen51/engression}. \revise{They support general data types so that both $X$ and $Y$ can be univariate or multivariate, continuous or discrete, and support various regression tasks.} GPU acceleration is for now only available in the Python version, but CPU-based training, while slower in general, is still suitable for moderately large data sets of up to a few hundred variables and a few thousand observations on a standard single-core machine.

Below is already an example to demonstrate how the software can be used in R (with default options) on a simple linear model, where the test distribution is slightly shifted from the training distribution (although a linear least-squares or quantile fit would be clearly be a better choice for this linear model example).

\begin{footnotesize}
    \begin{verbatim}
> n = 1000; p=5                                      ## 1000 samples and 5 dimensional covariate
> beta = rnorm(p)                                    ## random regression coefficients 
> X = matrix(rnorm(n*p), ncol=p)                     ## design matrix
> Y = X %*% beta + rnorm(n)                          ## linear model 
> Xtest = matrix(1 + rnorm(n*p), ncol=p)             ## test data are slightly shifted
> Ytest = Xtest %*% beta + rnorm(n)  

> library(engression)                                ## load engression package
> engressionFit = engression(X, Y)                   ## fit an engression model
> yhatEngression = predict(engressionFit, Xtest)     ## predict on test data
> mean((yhatEngression - Ytest)^2)                   ## evaluate conditional mean prediction on test data
[1] 1.045293
> predict(engressionFit, Xtest, type="quantile",     ## quantile prediction with the same engression object
quantiles=c(0.1, 0.5, 0.9))
> predict(engressionFit, Xtest, type="sample",       ## sampling from the fitted conditional distribution
nsample=100) 
\end{verbatim}
\end{footnotesize}

\subsection{Notation}
Let $\ptr(x)$ and $\ptr(y|x)$ denote the marginal distribution of $X$ and the conditional distribution of $Y|X=x$ during training, respectively. Let $\mathcal{X} \subset \mathbb{R}^d$ represent the support of $\ptr(x)$. 
For a univariate function $f(x)$, let $\dot{f}(x)$ be its derivative and $\ddot{f}(x)$ be its second derivative. For $x\in\bbR$, let $(x)_+=\max\{x,0\}$. For a vector $x\in\bbR^d$, let $\|x\|$ be the Euclidean norm. For a random variable $X$ and $\alpha\in[0,1]$, let $Q_\alpha(X):=\inf\{x:\bbP(X\leq x)\ge\alpha\}$ be the $\alpha$-quantile of $X$, or simply denoted by $Q_\alpha^X$. 
For random variables $X$ and $X'$ that follow the same distribution, we denote by $X\overset{d}=X'$. 
For $L>0$, a function $f:\bbR^d\to\bbR$ is $L$-Lipschitz continuous if $|f(x)-f(x')| \le L\|x-x'\|$ for all $x,x'\in\bbR^d$. 
Throughout the paper, all distributions are assumed to be absolutely continuous with respect to the Lebesgue measure unless stated otherwise. 

\section{Engression}\label{sec:method}
In Section~\ref{sec:anm}, we discuss the pre- and post-additive noise models as different ways to model the nonlinear relationship between $X$ and $Y$. In Section~\ref{sec:dist_reg}, we discuss distributional and non-distributional regression. In Section~\ref{sec:engression}, we propose the engression method \revise{as a general technique for distributional regression. In Section~\ref{sec:extrap_recipe}, we introduce a recipe for extrapolation with engression for pre-ANMs.}

\subsection{Pre-Additive versus Post-Additive Noise Models}\label{sec:anm}
\revise{
We generalize the pre- and post-ANMs introduced earlier in \eqref{eq:postanm_uni} and \eqref{eq:preanm_uni} to multivariate $X$. Formally, a post-additive noise model is expressed as
\begin{equation}\label{eq:postanm}
	Y=g(WX)+\eta,
\end{equation}
where $W\in\bbR^{k\times d}$, $g:\bbR^k\to\bbR$ is a nonlinear function, and $\eta$ is a noise variable.} Taking $W$ to be the identity matrix, the model becomes $Y=g(X)+\eta$ which is the standard form in nonlinear regression. When $k=1$, the model is a single index model~\citep{mccullagh1983generalized,hardle1993optimal,brillinger2012generalized}. 

In contrast, a general pre-additive noise model is of the form
\begin{equation}\label{eq:preanm}
	Y=g(WX+\eta)+\beta^\top X,
\end{equation}
where $W\in\bbR^{k\times d}$, $g:\bbR^k\to\bbR$ is a nonlinear function, $\beta\in\bbR^d$, and the noise $\eta$ is a random vector unless we take $k=1$. Note that the linear term $\beta^\top X$ is merely for adding more flexibility to the model but not necessary for extrapolation, and it can be dropped by setting $\beta=0$. 

The difference in the position of the noise can impact the extrapolability of the model. 
Intuitively, the response variable generated according to a post-ANM~\eqref{eq:postanm} is a vertically perturbed version of the true function, whereas the response from a pre-ANM~\eqref{eq:preanm} is a horizontally perturbed version. As such, if the data are generated according to a post-ANM, the observations for the response variable are perturbed values of the true function evaluated at covariate values within the support. We hence generally have no data-driven information about the behaviour of the true function outside the support. In contrast, data generated from a pre-ANM contain  response values that reveal some information beyond the support.

In fact, we will establish in Section~\ref{sec:extrapolability_anm} that pre-ANMs are extrapolable under mild assumptions, specifically when $g$ is strictly monotone, whereas post-ANMs fail to be extrapolable unless one imposes stringent assumptions on $g$, such as belonging to a linear function class. Furthermore, in subsequent sections, we demonstrate that the pre-additive noise can be beneficial for extrapolation: larger noise levels actually assist with extrapolation, thus transforming what could be a curse into a blessing, given that we use distributional regression. 

\revise{We discuss the connections of pre-ANMs to the existing post-nonlinear models in Section~\ref{sec:related_model}.} Furthermore, assuming $g$ is invertible, the pre-ANM in \eqref{eq:preanm} without the linear term can be written as $g^{-1}(Y)=W X+\eta$, which appears to be a linear model after transforming $Y$. Invertibility of $g$  requires that $k$ is equal to the dimension of the response. In many applications, it is common to transform the response variable by a pre-defined function, such as logarithm, based on domain expertise. 
This indicates from a practical perspective that the pre-ANMs could potentially fit many real data well. However, manually choosing the suitable transformation becomes inplausible especially when we have multivariate variables $Y$ and $X$. Moreover, with an additional linear term, the pre-ANM can no longer be captured through such a transformation. Thus, our proposed method goes even beyond learning the suitable transformation of $Y$ in a data-driven way.


\subsection{Distributional regression versus non-distributional regression}\label{sec:dist_reg}
Regression is a well-established technique for estimating the dependence of the response variable $Y$ given covariates $X$. The most classical view of regression focuses on estimating the conditional mean $\bbE[Y|X=x]$, with the most prevalent technique being the least-squares regression, a.k.a.\ $L_2$ regression~\citep{legendre1806nouvelles}. Specifically, it searches among a function class $\cF$ for the best fit that minimises the $L_2$ risk,
\begin{equation}\label{eq:l2}
	\min_{f\in\cF}\bbE[(Y-f(X))^2].
\end{equation}
Beyond the conditional mean, other functionals of the conditional distribution have been considered, such as quantiles. Given $\alpha\in[0,1]$, quantile regression~\citep{koenker2005quantile} is defined as the solution to 
\begin{equation}\label{eq:qr}
	\min_{f\in\cF} \bbE[\rho_\alpha(Y-f(X))],
\end{equation}
where $\rho_\alpha(x):=x(\alpha-\ind_{(x<0)})$. In the special case with $\alpha=0.5$, quantile regression becomes the least absolute deviation regression, a.k.a.\ $L_1$ regression, 
\begin{equation}\label{eq:l1}
	\min_{f\in\cF}\bbE|Y-f(X)|,
\end{equation}
which is used for estimating the conditional median of $Y$ given $X=x$.

In all the aforementioned techniques, the conditional distribution is summarised in a univariate quantity of interests, which may often lead to the most efficient approach in terms of the in-support performance. However, when it comes to extrapolation outside the support, even with the same goal of point prediction, \revise{we will highlight formally in Section~\ref{sec:extrapolability_def} the gain from aiming for the entire distribution rather than just the quantity of interest}. Intuitively, the out-of-support information revealed by a pre-ANM requires a full capture of the noise distribution, which is essentially the conditional distribution of $Y|X=x$.

\revise{An important element of a distributional regression method is a quantitative metric 
to assess a distributional fit.} The \emph{energy score} introduced by \citet{gneiting2007strictly} is a popular scoring rule for evaluating multivariate distributional predictions. For a random vector $Z$ that follows a distribution $P$ and an observation $z$,  the energy score is defined as 
\begin{equation}\label{eq:es}
	\es(P,z)=\frac{1}{2}\bbE_{P}\|Z-Z'\| - \bbE_P\|Z-z\|,
\end{equation} 
where $Z$ and $Z'$ are two independent draws from $P$. Note that the higher the score, the better the distributional prediction. We will generally use the negative of the energy score to design a loss function and call it energy loss where appropriate.
The following lemma based on the results in \citet{szekely2003statistics} and \citet{Szkely2023TheEO} states that the energy score is a strictly proper scoring rule. 
\begin{lemma}\label{lem:es}
	For any distribution $P'$, we have $\bbE_{Z\sim P}[\es(P,Z)] \ge \bbE_{Z\sim P}[\es(P',Z)] $, where the equality holds if and only if $P$ and $P'$ are identical. 
\end{lemma}

When $Z$ is univariate whose cumulative distribution function (cdf) is denoted by $F$, the energy score reduces to the continuous ranked probability score (CRPS)~\citep{matheson1976scoring} defined as
\begin{equation*}
	\crps(F,z)=-\int_{-\infty}^\infty(F(z')-\mathbf{1}\{z'\ge z\})^2dz'.
\end{equation*}
The equality between CRPS and the energy score in \eqref{eq:es} in the univariate case was shown by \citet{BARINGHAUS2004190}, which is essentially a distributional extension of the (negative) absolute error. 

These scoring rules are associated with certain distance functions. In particular, the respective distance of the energy score for two distributions $P$ and $P'$ is the energy distance~\citep{szekely2003statistics} defined as
\begin{equation}\label{eq:energy_distance}
	\ed(P,P')= 2\bbE\|Z-Z'\| - \bbE\|Z-\tilde{Z}\| - \bbE\|Z'-\tilde{Z}'\|,
\end{equation}
where $Z$ and $\tilde{Z}$ are two independent draws from $P$, and $Z'$ and $\tilde{Z}'$ are independent draws from $P'$. 
The distance associated with CRPS for two cdf's $F$ and $F'$ is the Cram\'er distance~\citep{szekely2003statistics} defined as
\begin{equation*}
	\cramer(F,F')=\int_{-\infty}^\infty(F(z)-F'(z))^2dz.
\end{equation*}
Both the energy distance and the Cram\'er distance are special cases of the maximum mean discrepancy (MMD) distance~\citep{gretton2012kernel,sejdinovic2013equivalence} and had been incorporated in \citet{bellemare2017cramer} in the context of generative adversarial networks~\citep{goodfellow2014} \revise{and recently in \citet{chen2024generative} for post-processing ensemble weather forecasts. See Section~\ref{sec:related_distreg} for more literature review on distributional regression and generative models.}

\subsection{Engression methodology}\label{sec:engression}

\revise{
Consider a general model class  
\begin{equation}\label{eq:gen_model_class}
	\cM = \{\gen(x,\varepsilon)\}	
\end{equation}
where $\gen:(x,\varepsilon)\mapsto y$ belongs to a function class and $\varepsilon$ is a random vector with a pre-specified distribution independent from covariates $X$, e.g.\ with each component following i.i.d.\ a uniform or standard Gaussian distribution. 
Each model in \eqref{eq:gen_model_class} is a generative model that induces a conditional distribution of $Y|X=x$: for $\gen\in\cM$ and $x\in\bbR^d$, we denote by $P_{\gen}(y|x)$ the distribution of $\gen(x,\varepsilon)$ with the randomness coming from $\varepsilon$. Most generally, $\cM$ can contain both pre-ANM and post-ANM classes.
}

We define the solution $\tilde\gen$ to the population version of \emph{engression} as
\begin{equation}\label{eq:eng_pop}
	\tilde\gen\in\argmin_{\gen\in\revise{\cM}} \bbE_{X\sim \ptr(x)}[\cL(P_{\gen}(y|X);\ptr(y|X))],
\end{equation}
where $\cL(P;P_0)$ is a loss function for a distribution $P$ given a reference distribution $P_0$. We will define the finite-sample version of engression after specifying the choice for the loss function.  

Engression \revise{in principle} allows a rather flexible choice of the loss function, with the only requirement being the characterisation property: for all distributions $P$, it holds that $\cL(P;P_0)\ge\cL(P_0;P_0)$, where the equality holds if and only if $P=P_0$. It can be designed based on a statistical distance. Note that when the function classes $\cG$ and $\cH$ are nonparametric classes or neural networks, the loss function based on the KL divergence or Wasserstein distance in general does not have an analytic form to be easily optimised. For such loss functions, one typical solution is to adapt the generative adversarial networks~\citep{goodfellow2014}, which is more computationally involved. 

In comparison, the energy score \eqref{eq:es} or energy distance \eqref{eq:energy_distance} are appealing due to their computational simplicity: the estimation of the score or distance is \revise{explicitly computed} based on sampling. This works nicely with the generative nature of our model. Specifically, for any given $x$, one can draw a sample from the conditional distribution $P_{\gen}(y|x)$ by first sampling the random vector $\varepsilon\sim\mathrm{Unif}[0,1]$ and then applying the transformation $\gen$. Therefore, in this work, our primary choice for the loss function is the energy loss based on the negative energy score, defined as
\begin{equation*}
	\cL_e(P;P_0) := \bbE_{Y\sim P_0}[-\es(P,Y)].
\end{equation*}
Then the objective function in engression \eqref{eq:eng_pop} yields
\begin{equation}\label{eq:eng_obj_pop}
	\bbE[\cL_e(P_{\gen}(y|X);\ptr(y|X))] =  \bbE\Big[\|Y-\gen(X,\varepsilon)\| - \frac{1}{2}\|\gen(X,\varepsilon)-\gen(X,\varepsilon')\|\Big],
\end{equation}
where $(X,Y)\sim \ptr(x,y)$ and $\varepsilon$ and $\varepsilon'$ are independent draws from $\mathrm{Unif}[0,1]$ (without loss of generality). 
 
\revise{The first term in \eqref{eq:eng_obj_pop} measures the mean absolute (for univariate $Y$) residual of an engression sample. By minimising the first term alone, the inferential target would be the conditional median of $Y|X=x$, which boils down to the special case of a deterministic model $\gen(x,\varepsilon) = g(x)$. The second term measures the expected distance between samples of the generated distribution and encourages a positive variance of the fitted distribution (as it enters the loss with a negative sign). }

\revise{The following proposition provides a population guarantee for engression as a general distributional regression approach. It shows that when the model is correctly specified, the population engression learns the true conditional distribution for all $x$ within the support. Note that when the model class is a neural network class with a sufficiently large capacity, model specification is often not an issue in practice. Moreover, this guarantee also holds for a multivariate $Y$. See the proof in Appendix~\ref{pf:prop:dist_est_pop}.

\begin{proposition}\label{prop:dist_est_pop}
    Assume there exists $\gen\in\cM$ such that $\gen(x,\varepsilon)\sim \ptr(y|x)$ for all $x\in\cX$. Then the population engression solution $\tilde\gen$ defined in \eqref{eq:eng_pop} with the loss function being the energy loss \eqref{eq:eng_obj_pop} satisfies $\tilde\gen(x,\varepsilon)\sim \ptr(y|x)$ for all $x\in\cX$ almost everywhere.
\end{proposition}
}

Nevertheless, it is worth pointing out that the above result as well as all the theoretical results in Section~\ref{sec:theory_extra} hold for any loss functions that satisfy the above characterisation property. We investigate the empirical performance of engression with other loss functions in Appendix~\ref{app:losses} including the MMD distance and KL divergence, where we observe that different loss functions result in comparable performance. While it will be worthwhile to discuss the relative advantages from a theoretical point of view, we work primarily with the energy loss because it is simple to understand without the need to choose further hyperparameters, is computationally fast as it does not require adversarial training, and is empirically well behaved during optimisation. 

\subsubsection{Finite-sample engression}\label{sec:engression_empirical}
Consider an i.i.d.\ sample $\{(X_i,Y_i):i=1,\dots,n\}$ from $\ptr(x,y)$. For each observation, we sample the noise variable $\varepsilon$ independently from $\mathrm{Unif}[0,1]$ for $m$ times, resulting in an i.i.d.\ sample $\{\varepsilon_{i,j}:i=1,\dots,n,j=1,\dots,m\}$. Based on the finite samples, we define the empirical version of engression as 
\begin{equation}\label{eq:eng_emp}
	\hat\gen\in\argmin_{\gen\in\preanm}\hat{\cL}(\gen)
\end{equation}
with the empirical objective function given by
\small
\begin{equation*}
	\hat{\cL}(\gen):=\frac{1}{n}\sum_{i=1}^n\left[\frac{1}{m}\sum_{j=1}^m\big\|Y_i-\gen(X_i,\varepsilon_{i,j})\big\| - \frac{1}{2m(m-1)}\sum_{j=1}^m\sum_{j'=1}^m\big\|\gen(X_i,\varepsilon_{i,j})-\gen(X_i,\varepsilon_{i,j'})\big\|\right],
\end{equation*}
\normalsize\hspace{-6pt}
which is an unbiased estimator of the population objective function \eqref{eq:eng_obj_pop}.
We parameterise the pre-ANMs using neural networks. Then the above empirical objective function is differentiable \revise{almost everywhere} with respect to all the model parameters. Thus, we adopt (stochastic) gradient descent algorithms~\citep{cauchy1847methode,robbins1951stochastic} to solve the optimisation problem in \eqref{eq:eng_emp}. 

\subsubsection{Point prediction}
Engression, being a distributional regression method, estimates the conditional distribution of $Y|X=x$, thus also providing estimators for various characteristics of the distribution such as the mean and quantiles. Specifically, at the population level, the engression estimator for the conditional mean of $Y$ given $X=x$ is derived from $\tilde\mu(x):=\bbE_\varepsilon[\tilde\gen(x,\varepsilon)]$. The engression estimator for the conditional median is $\tilde{m}(x):=Q_{0.5}(\tilde\gen(x,\varepsilon))$, where the quantile is taken with respect to $\varepsilon$. For any $\alpha\in[0,1]$, the engression estimator for the conditional $\alpha$-quantile is $\tilde{q}_\alpha(x):=Q_{\alpha}(\tilde\gen(x,\varepsilon))$.

For a finite sample, based on the empirical engression solution $\hat\gen$, we construct the estimators by sampling. Specifically, for any $x$, we sample $\varepsilon_i,i=1,\dots,m$ for some $m\in\mathbb{N}_+$, and then obtain $\hat\gen(x,\varepsilon_i),i=1,\dots,m$. These form an i.i.d.\ sample from the conditional distribution $p_{\hat\gen}(y|x)$. All the quantities are then estimated using their empirical versions from this sample. For instance, the conditional mean of $Y$ given $X=x$ is estimated by $\frac{1}{m}\sum_{i=1}^m\hat\gen(x,\varepsilon_i)$.

\subsection{Engression for extrapolation: recipe and real data starter}\label{sec:extrap_recipe}

\revise{Our recipe for extrapolation consists of} two main ingredients: \vspace{-0.05in}
\begin{enumerate}[label=(\roman*)]
\setlength{\itemsep}{0pt}
\setlength{\parskip}{2pt}
	\item distributional regression;  
    \item a pre-additive noise model.
\end{enumerate}\vspace{-0.05in}

As we will demonstrate in the next section, both ingredients are necessary for successful extrapolation. 
Results in Section~\ref{sec:extrapolability_anm} indicate the failure of distributional regression in conjunction with post-ANMs, while results in Section~\ref{sec:local_extrap} show the failure of the combinations of non-distributional approaches with pre- or post-ANMs. 

\revise{Notably, engression possesses both ingredients. First, Proposition~\ref{prop:dist_est_pop} justifies engression as a distributional regression approach. Second, our general model class $\cM$ in \eqref{eq:gen_model_class} contains the pre-ANM class, for which we provide detailed explanations that align more closely with our practical implementations in Appendix~\ref{app:eng_preanm}. We later provide some evidence of successful extrapolation for engression. }

\begin{figure}
\centering
\begin{tabular}{cc}
	\includegraphics[align=c,width=0.4\textwidth]{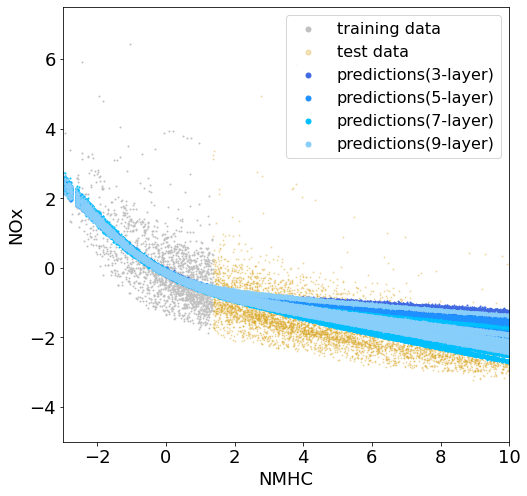}\hspace{0.3cm} &\hspace{0.3cm}
	\includegraphics[align=c,width=0.4\textwidth]{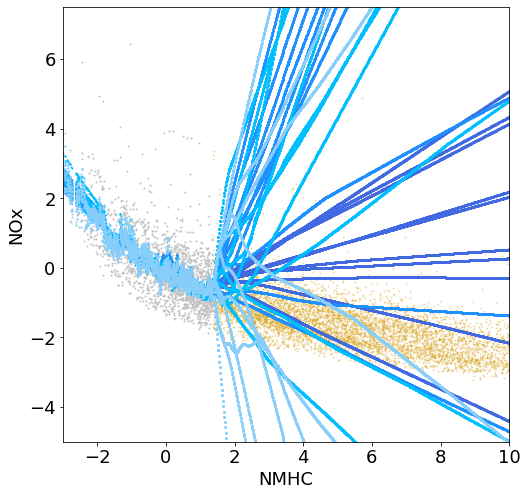} \\
	{(a) Engression} & (b) Regression
\end{tabular}
\caption{Out-of-support predictions by engression and $L_2$ regression on the air quality data. The models are trained on a support up to the first quartile of NMHC and evaluated in their predictions for larger NMHC values. To ensure a fair comparison, we apply both methods with the same neural network architecture whose number of layers is varying from 3 to 9; for each architecture, we repeatedly employ both methods for ten times (each with an independent random initialisation) and show  the behaviour of all of them. }\label{fig:illus_univar}
\end{figure}

Before presenting the formal theory, we demonstrate engression versus regression using  the air quality data set from \citet{misc_air_quality_360}, where we take the measurements of two pollutants, NMHC and NOx, as the predictor and response variables, respectively. We partition the data into training and test sets at the first quartile of the predictor and take the smaller portion for training. 
As visualised in Figure~\ref{fig:illus_univar}, engression and regression work comparably well within the training data while behaving very differently on out-of-support test data, despite the same underlying network architectures being used. Engression maintains excellent prediction performance on test data beyond the training boundary up to a certain range, when it starts to diverge moderately.  In contrast, regression produces highly variable predictions that spread out the whole space once going outside the support. Especially as the number of layers in the neural network grows, the out-of-support predictions by regression disperse even wider, whereas engression is highly robust to various model architectures. 

In the sequel, we cover theoretical and empirical perspectives of engression. Section~\ref{sec:theory_extra} presents the formal theory of extrapolation and justify the virtue of engression in this regard. Section~\ref{sec:finite_sample} provides finite-sample guarantees for the engression estimator. Section~\ref{sec:empirical} is devoted to large-scale experimental investigations.

\section{Extrapolation theory}\label{sec:theory_extra}
In Section~\ref{sec:extrapolability_def}, we define the extrapolation uncertainty and extrapolability for a given model class. As a demonstration, in Section~\ref{sec:extrapolability_anm}, we study the extrapolability of pre-ANMs and post-ANMs. Section~\ref{sec:local_extrap} discusses the extrapolability of engression, while Section~\ref{sec:extra_gain} quantifies the extrapolability gains of engression compared to the baseline methods. All results in this section are at the population level.

\subsection{Extrapolability}\label{sec:extrapolability_def}

In this section, we define \emph{extrapolability} through the notions of \emph{extrapolation uncertainty} at functional and distributional levels. The term `uncertainty' here refers to the spread among out-of-support predictions that arises due to partially observing the sample space and the fact that multiple functions or models fit the in-support data equally well, as made  more precise further below. It is in particular distinct from the `stochastic uncertainty' caused by the randomness from a finite sample. 

We start with the extrapolation uncertainty and extrapolability in terms of a function class which includes the notion of extrapolation used in existing literature as a special case. For any $x\in\bbR^d$, define the distance between $x$ and set $\cX$ by $d(x,\cX)=\inf_{x'\in\cX}\|x-x'\|$.

\begin{definition}[Functional extrapolability]\label{def:extra_func}
	Let $\cF$ be a function class of $x\in\bbR^d$. For $\delta>0$, define the functional extrapolation uncertainty of a function class $\cF$ as
	\begin{equation*}
		\cU_\cF(\delta) := \sup_{x':d(x',\cX)\le \delta} \sup_{\substack{f,f'\in\cF:\\D_{x}(f,f')=0,\forall x\in\cX}}D_{x'}(f,f'),
	\end{equation*}
	where $D_{x}(f,f'):=|f(x)-f'(x)|$ for all $x\in\bbR^d$.
	$\cF$ is said to be (globally) functionally extrapolable, if $\cU_\cF(\delta)\equiv0$.
\end{definition}
In words, the functional extrapolation uncertainty quantifies the worst-case disagreement of two functions that agree perfectly within the support.
The following examples illustrate the functional extrapolation uncertainty of several common function classes, which are proved in Appendix~\ref{pf:ex}.

\begin{example}[Linear functions]\label{ex:lin}
	Linear functions are functionally extrapolable, i.e.\ $\cU(\delta)\equiv0$. Note that the functional extrapolability has appeared in existing work, e.g.\ \citet{christiansen2021causal}. It essentially requires the function to be uniquely determined by its values on the support. 
\end{example}
\begin{example}[Lipschitz functions]\label{ex:lip}
	The functional extrapolation uncertainty of the class of $L$-Lipschitz continuous functions is $\cU(\delta)=2L\delta$.
\end{example}
\begin{example}[Monotone functions]\label{ex:monot}
	The functional extrapolation uncertainty of a class of (component-wise) monotone functions is $\cU(\delta)=\infty$.
\end{example}


In distributional regression, the relationship between $X$ and $Y$ is comprehensively captured by the conditional distribution of $Y|X=x$. 
Consider a class $\cP=\{P(y|x)\}$ of conditional distributions. Standard prediction tasks are often concerned with a certain characteristic quantity, such as the conditional mean function $\mu_{P}(x)=\bbE_{P(y|x)}[Y]$ or conditional quantile functions. Taking the conditional mean as an example, we define below the mean extrapolability which is a direct extension of the functional extrapolability towards a distributional level.
\begin{definition}[Mean extrapolability]\label{def:extra_mean}
	For $\delta>0$, define the mean extrapolation uncertainty of $\cP$ as
	\begin{equation*}
		\cU^\mu_\cP(\delta) := \sup_{x':d(x',\cX)\le \delta} \sup_{\substack{P,P'\in\cP:\\D_x(P,P')=0,\forall x\in\cX}}D_{x'}(P,P'),
	\end{equation*}
	where $D_x(P,P'):=|\mu_{P}(x)-\mu_{P'}(x)|$ for all $x\in\bbR^d$.
\end{definition}

In Definition~\ref{def:extra_mean}, the distance between two distributions is measured by the difference between the two conditional means. A more general way instead is to use a statistical distance, denoted by $D(P,P')$, between two probability distributions $P$ and $P'$, such as the Wasserstein distance or KL divergence. This leads us to the final notion of extrapolability in terms of a class of distributions.

\begin{definition}[Distributional extrapolability]\label{def:extra_dist}
	For $\delta>0$, define the distributional extrapolation uncertainty of $\cP$ as
	\begin{equation*}
		\cU_\cP(\delta) := \sup_{x':d(x',\cX)\le \delta} \sup_{\substack{P,P'\in\cP:\\D_x(P,P')=0,\forall x\in\cX}}D_{x'}(P,P'),
	\end{equation*}
	where $D_x(P,P'):=D(P(y|x),P'(y|x))$ for all $x\in\bbR^d$.
	$\cP$ is said to be (globally) distributionally extrapolable, if $\cU_\cP(\delta)\equiv0$.
\end{definition}

The distributional extrapolation uncertainty provides an upper bound on the degree of disparity between two conditional distributions outside the support, given that they coincide within the support. Distributional extrapolability implies that the conditional distribution of $Y|X=x$ for all $x$ globally is uniquely determined by the conditional distributions of $Y|X=x$ for all $x$ within the support. 
Perhaps surprisingly, unlike functional extrapolability, which imposes stringent constraints on the function class (such as linearity, as discussed in Example~\ref{ex:lin}), distributional extrapolability can be achieved under relatively mild conditions. These will be realised concretely in Theorem~\ref{thm:extra_anm} below.

\revise{
We conclude this section by illuminating the connections among the three levels of extrapolability. In essence, mean extrapolability serves as a bridge between functional and distributional extrapolability. 
If a method is solely aimed at fitting the conditional mean function, its ability to extrapolate beyond the training support relies fully on the extrapolability of the function class itself, which could be highly restrictive. In contrast, a method that aims to match the entire conditional distribution needs potentially much weaker assumptions on the function class for successful extrapolation. 
}

This insight implies that if a method is solely aimed at fitting the conditional mean function, its ability to extrapolate beyond the training support relies fully on the extrapolability of the function class itself, which could be highly restrictive. In contrast, a method that aims to match the entire conditional distribution needs potentially much weaker assumptions on the function class for successful extrapolation. This suggests from another perspective than the intuitions provided in Section~\ref{sec:dist_reg} that distributional regression could bring in more potential for extrapolation.

\subsection{extrapolability of ANMs}\label{sec:extrapolability_anm}

For ease of presentation, starting from this section, we will primarily focus on the univariate case with $X\in\bbR$ unless stated otherwise. We introduce some additional notation. 
\revise{Recall in \eqref{eq:postanm_uni} that a univariate post-ANM} is expressed as $Y = g(X)+\eta,$
where $g$ is a univariate nonlinear function and the parameter $W$ in \eqref{eq:postanm} is absorbed into $g$. We assume $\eta$ is independent from $X$ and follows an arbitrary distribution that is absolutely continuous with respect to the Lebesgue measure so that $\eta\overset{d}=h^\star(\varepsilon)$ for a strictly monotone function $h^\star$ and a random variable $\varepsilon\sim\mathrm{Unif}[0,1]$. Note that $\varepsilon$ can in principle be any continuous random variable and we choose the uniform distribution for convenience. Also assume without loss of generality that $\eta$ has a zero median so that $h^\star(0.5)=0$, which avoids trivial non-identifiability. 
Define function classes $\cG=\{g(x)\}$ and $\cH=\{h(\varepsilon):h(0.5)=0,h\text{ is strictly monotone}\}$. We denote by $\cM_{\mathrm{post}}=\{g(x)+h(\varepsilon):g\in\cG,h\in\cH\}$ a class of post-ANMs. 
Similarly, the pre-ANM is of the following form
\begin{equation}\label{eq:preanm_uni}
	Y = g(X+\eta)+\beta X.
\end{equation}
We denote by $\cM_{\mathrm{pre}}=\{g(x+h(\varepsilon))+\beta x:g\in\cG,h\in\cH,\beta\in\bbR\}$ a class of pre-ANMs.

Note that any model in $\cM_{\mathrm{pre}}$ or $\cM_{\mathrm{post}}$ induces a conditional distribution of $Y$, given $X=x$. For example, a post-ANM $g(X)+\eta$ with $\eta\sim\cN(0,1)$ follows a  Gaussian distribution $\cN(g(x),1)$ conditional on $X=x$. Thus, the extrapolability of $\cM_{\mathrm{pre}}$ or $\cM_{\mathrm{post}}$ is defined as the distributional extrapolability of the class of conditional distributions induced by the models in $\cM_{\mathrm{pre}}$ or $\cM_{\mathrm{post}}$. 

\begin{definition}[Linear and nonlinear pre-ANMs]
	A pre-ANM class $\preanm$ is said to be linear if $\cG$ belongs to a linear function class.
	A pre-ANM class $\preanm$ is said to be nonlinear if for all $g\in\cG$, for all $h\in\cH$ and $\varepsilon\in[0,1]$, it holds that $g$ is twice differentiable and there exists $x\in\cX$ such that $\ddot{g}(x+h(\varepsilon))\neq0$.
\end{definition}

The following theorem formalises the extrapolability of a pre-ANM class, either linear or nonlinear, in contrast to the non-extrapolability of a post-ANM class. 
\begin{theorem}\label{thm:extra_anm}
	Assume $h$ is unbounded for all $h\in\cH$, i.e.\ for any $B>0$, there exists $\varepsilon_1,\varepsilon_2$, such that $h(\varepsilon_1)>B$ and $h(\varepsilon_2)<-B$. If $g$ is strictly monotone and twice differentiable for all $g\in\cG$, then we have\vspace{-0.05in}
	\begin{enumerate}[label=(\roman*)]
	\setlength{\itemsep}{0pt}
	\setlength{\parskip}{2pt}
		\item both linear and nonlinear $\cM_{\mathrm{pre}}$ are distributionally extrapolable, i.e.\ $\cU_{\mathrm{pre}}(\delta)=0$;
		\item $\cM_{\mathrm{post}}$ has an infinite distributional extrapolation uncertainty, i.e.\ $\cU_{\mathrm{post}}(\delta)=\infty$.
	\end{enumerate}\vspace{-0.05in}
	In fact, $\cG$ being functionally extrapolable is the sufficient and necessary condition for $\cM_{\mathrm{post}}$ to be distributionally extrapolable.
\end{theorem}
The proof is provided in Appendix~\ref{pf:thm:extra_preanm}. The assumptions for a pre-ANM class to perform extrapolation include the unboundedness of $h$ and the monotonicity of the function class $\cG$. 
The first assumption essentially asserts that the noise $\eta=h(\varepsilon)$ is supported on $\bbR$, meaning that it has a strictly positive density on $\bbR$. While this may seem like a mild assumption, in transitioning from a population case to finite samples, we must acknowledge that the empirical support of the training data is always finite. Thus, starting from the next section, we will navigate towards a more realistic scenario where the noise $\eta$ may have a bounded support.
The assumption of monotonicity is also much broader than assuming the functional extrapolability of the function class. Nevertheless, we can relax the requirement for global monotonicity to that of monotonicity only in proximity to the boundary of the support.

\revise{
Theorem~\ref{thm:extra_anm} highlights a substantial difference between pre-ANMs and post-ANMs: under identical (and rather mild) assumptions, pre-ANMs satisfy distributional extrapolability, whereas post-ANMs fail to extrapolate unless the function class is inherently extrapolable.
}
Note that although a linear pre-ANM class is not identifiable (i.e.\ there exist two linear pre-ANMs that induce the same conditional distribution of $Y|X$), it still satisfies the distributional extrapolability according to Theorem~\ref{thm:extra_anm}. 
In fact, the two model classes $\cM_{\mathrm{pre}}$ and $\cM_{\mathrm{post}}$ only coincide when $\cG$ is a linear function class, which is in harmony with the functional extrapolability of linear functions presented in Example~\ref{ex:lin}. 

It is worth noting that one could establish extrapolability for a model that allows for both pre- and post-additive noises, for example, as follows 
\begin{equation*}
    Y = g(X+\eta)+\beta X + \xi,
\end{equation*}
where $\eta$ is the pre-additive noise while $\xi$ is the post-additive noise. In Appendix~\ref{app:pre_post_anm}, we provide an example of this possibility while more general investigations are worthwhile in future research.

\subsection{extrapolability of engression}\label{sec:local_extrap}
Next, we take a closer look at the extrapolability properties and guarantees of engression. Denote the solution to the population engression in \eqref{eq:eng_pop} explicitly by $(\tilde{g},\tilde{h},\tilde\beta)$. 
Based on the definitions in Section~\ref{sec:extrapolability_def}, we define the extrapolability of a method as the extrapolability of the class of estimators obtained from that method. 
For a quantity $\textsl{q}$ of interest, denote by $\cF^\textsl{q}_\rmE$ the class of functions for estimating the quantity $\textsl{q}$ obtained by applying engression on the training distribution. Specifically, we use $\textsl{q}=\textsl{m},\mu$ and $\alpha$ to stand for the (conditional) median, mean, and $\alpha$-quantile, respectively. For a method $\textsl{M}$, denote by $\cF_\textsl{M}$ the class of functions obtained by applying the method $\textsl{M}$ on the training distribution. Specifically, we use $\textsl{M}=\mathrm{QR},L_1$ and $L_2$ to represent quantile regression, $L_1$ regression, and $L_2$ regression, respectively. We denote by $\cF^\alpha_{\mathrm{QR}}$ the solution set of quantile regression for the $\alpha$-quantile.

Throughout this section, we assume the following conditions are fulfilled.
\begin{enumerate}[leftmargin=1cm,label=(\textit{A\arabic*})]\vspace{-0.05in}
\setlength{\itemsep}{0pt}
\setlength{\parskip}{2pt}
\item The noise $\eta$ follows a symmetric distribution with a bounded support, i.e.\ $\eta\in[-\etam,\etam]$ for $\etam>0$.\label{ass:bounded_noise}
\item The training support is a half line $\cX=(-\infty,\xm]$. \label{ass:support_half}
\item $g$ is strictly monotone and twice differentiable for all $g\in\cG$.\label{ass:monotone}
\item The true model, denoted by $Y=g^\star(X+h^\star(\varepsilon))+\beta^\star X$, satisfies $g^\star\in\cG$, $h^\star\in\cH$, and $g^\star$ is nonlinear, i.e.\ $\ddot{g}^\star(x)\neq0$ apart from a set with Lebesgue measure 0.\label{ass:well_specify}
\end{enumerate}\vspace{-0.05in}
Condition~\ref{ass:bounded_noise} is not a necessary assumption as a bounded support makes the problem more challenging for extrapolation, as suggested by Theorem~\ref{thm:extra_anm}. We include this assumption to make the setting more plausible for a broader range of real-world scenarios, compared to the requirement of unbounded noise. 

Condition~\ref{ass:support_half} simplifies the discussions by restricting the focus on extrapolation where $x>\xm$. Though this is a technical assumption for ease of analysis, the insights can be extended to scenarios where extrapolation is required on both ends of the support.

Condition~\ref{ass:monotone} requires the monotonicity and differentiability of the function $g$. Compared to a Lipschitz condition, this assumption presents a tougher challenge for extrapolation. However, it could be readily relaxed to the condition where $g$ is strictly monotone only in the interval $[\xm-\etam,\xm+\etam]$, under which the results regarding the extrapolation for $x$ outside the support would still hold.

Finally, Condition~\ref{ass:well_specify} ensures the correct specification of the model, meaning that the true data generating process $g^\star(X+h^\star(\varepsilon))+\beta^\star X$ belongs to the pre-ANM class $\cM_{\mathrm{pre}}$. This condition helps us concentrate our analysis on the extrapolation behaviour of the method, by excluding the effects of any approximation errors. The additional requirement that $g^\star$ be nonlinear excludes the linear setting where extrapolation is straightforward and positions us in the pivotal scenario of nonlinear extrapolation. 

Given that the noise is bounded, we cannot expect  extrapolability across the whole real line. However, we can achieve extrapolability within a specific range that extends beyond the support, an advantage that is not seen in alternative methods under the same mild  conditions.
 Moreover, in the upcoming subsection, we will compare our method with baseline approaches by quantifying the gains between their extrapolation uncertainties globally on $\bbR$, where engression constantly holds an advantage.

The following theorem suggests that engression can identify the true function $g^\star$ up to a certain range depending on the noise level. This brings up an interesting point that the (pre-additive) noise is a blessing rather than a curse for extrapolation: the more noise one has (i.e.\ a larger $\etam$), the farther one can extrapolate. In fact, this further echoes the result in Theorem~\ref{thm:extra_anm} which states that with an unbounded noise ($\etam\to\infty$), we can achieve global extrapolability and recover $g^\star$ over the entire real line. The proof is given in Appendix~\ref{app:pf_local_extrap}.
\begin{theorem}\label{thm:func_recover}
	We have $\tilde\beta=\beta^\star$, $\tilde{g}(x)=g^\star(x)$ for all $x\leq x_{\max}+\eta_{\max}$, and $\tilde{h}(\varepsilon)=h^\star(\varepsilon)$ for all $\varepsilon\in[0,1]$.
\end{theorem}

We say a model class is locally extrapolable on $\Delta\subseteq(0,\infty)$ if $\cU(\delta)=0$ for all $\delta\in\Delta$. For example, global extrapolability is local extrapolability with $\Delta=(0,\infty)$. 
Consider estimating the conditional quantile.
For $\alpha\in[0,1]$, let $q^\star_\alpha(x)$ be the conditional $\alpha$-quantile of $Y$ given $X=x$ under the training distribution $\ptr(y|x)$. The class of engression estimators for the $\alpha$-quantile is given by $\cF_\rmE^\alpha=\{\tilde{q}_\alpha(x)\}$. The following result shows that $\cF_\rmE^\alpha$ is locally extrapolable on $(0,\eta_{\max}-Q_\alpha(\eta)]$. 
\begin{corollary}\label{coro:local_quantile}
	For $\alpha\in[0,1]$, it holds for all $x\leq x_{\max}+\eta_{\max}-Q_\alpha(\eta)$ that $\tilde{q}_\alpha(x)=q^\star_\alpha(x)$, i.e.
	\begin{equation*}
		\bbP\left(Y \leq \tilde{q}_{1-\alpha}(X)\mid X=x\right) = 1-\alpha.
	\end{equation*}
\end{corollary}
In comparison, quantile regression does not have any local functional extrapolability. To see this, notice that quantile regression defined in \eqref{eq:qr} among the function class $\cG$ leads to the solution set
\begin{equation}\label{eq:qr_class}
	\cF^\alpha_{\mathrm{QR}}=\{f\in\cG:f(x)=q^\star_\alpha(x),\forall x\le\xm\}
\end{equation}
which contains all functions that match the true quantile function within the support but can take arbitrary values for $x$ outside the support. 
For a monotone function class $\cG$, the extrapolation uncertainty of $\cF^\alpha_{\mathrm{QR}}$ is actually infinity for all $\delta>0$, as seen in Example~\ref{ex:monot}. Alternatively, one may consider quantile regression over a pre-ANM class, e.g.
\begin{equation*}
	\min_{g\in\cG,h\in\cH} \bbE[\rho_\alpha(Y-g(X+h(\varepsilon)))].
\end{equation*}
In fact, the above formulation leads to the same solution set \eqref{eq:qr_class}, thus possessing an infinite extrapolation uncertainty. These results suggest that in the recipe for extrapolation, 
the combinations of non-distributional approaches with either pre-ANMs or post-ANMs do not work.


\subsection{extrapolability gain}\label{sec:extra_gain}
We present the extrapolability gains (in terms of extrapolation uncertainties) of engression over baseline methods for estimating the conditional median, mean, and distribution. 
We have seen from Example~\ref{ex:monot} and Theorem~\ref{thm:extra_anm} that with monotone functions, the baseline methods yield infinite extrapolation uncertainties. Thus, to ensure informative quantifications, throughout this subsection, we assume in addition to Conditions \ref{ass:bounded_noise}-\ref{ass:well_specify} that $\cG$ is uniformly $L$-Lipschitz continuous, i.e.\ \vspace{-0.1in}
\begin{enumerate}[leftmargin=*,label=(\textit{A\arabic*})]\addtocounter{enumi}{4}
	\item For all $g\in\cG$, it holds that $0<\dot{g}(x)\leq L$ for all $x\in\bbR$,\label{ass:lip}
\end{enumerate}\vspace{-0.1in}
although it is not a necessary condition for the extrapolability of engression.

According to Theorem~\ref{thm:func_recover}, we have for the engression solution that $\tilde{g}\in\cG^\star$ with
\begin{equation*}
	\cG^\star := \{g\in\cG:g(x)=g^\star(x), \forall x\le \xm+\etam\}.
\end{equation*}

We start with the extrapolability gain for conditional median estimation, where we compare engression to~$L_1$ regression defined in \eqref{eq:l1} with the  same function class $\cG$. As discussed at the end of Section~\ref{sec:engression}, given an engression estimator $(\tilde{g},\tilde{h},\tilde\beta)$ and a point $x\in\bbR$, we use the median of $\tilde{g}(x+\tilde{h}(\varepsilon))+\tilde\beta x$ as the median estimator. Since $\tilde{g}$ is strictly monotone and $\tilde{h}(\varepsilon)$ has a median 0, the engression estimator for the conditional median of $Y$ given $X=x$ becomes $\tilde{g}(x)+\tilde\beta x$. Thus, the class of conditional median estimators by engression is given by $\cF^\textsl{m}_{\mathrm{E}}=\{g(x)+\beta^\star x:g\in\cG^\star\}$. The one obtained by $L_1$ regression is denoted by $\cF_{L_1}$. The proofs for the extrapolability gains (Propositions~\ref{prop:gain_median}-\ref{prop:gain_dist}) are given in Appendix~\ref{app:pf_extra_gain}.
\begin{proposition}[Median extrapolability gain]\label{prop:gain_median}
	The functional extrapolation uncertainties of the two classes of median estimators are given by $\cU_{\cF^\textsl{m}_{\mathrm{E}}}(\delta) = L(\delta-\etam)_+$ and $\cU_{\cF_{L_1}}(\delta) = L\delta$, respectively. The median extrapolability gain, defined by the difference between the two functional extrapolation uncertainties, is
	\begin{equation*}
		\gamma^\textsl{m}(\delta) := \cU_{\cF_{L_1}}(\delta) - \cU_{\cF^\textsl{m}_{\mathrm{E}}}(\delta) = L\min\{\delta,\etam\}>0,
	\end{equation*} 
	for all $\delta>0$. The maximum extrapolability gain is $\sup_{\delta>0}\gamma^\textsl{m}(\delta)=L\etam$, which is achieved when $\delta\ge \etam$.
\end{proposition}

Next, we consider conditional mean estimation which is the most common prediction task. 
The class of conditional mean estimators from engression is $\cF^\mu_{\mathrm{E}}=\{\bbE_{\eta}[g(x+\eta)]+\beta^\star x:g\in\cG^\star\}$, by Theorem~\ref{thm:func_recover}. We consider $L_2$ regression with the function class $\cG$ as our baseline which gives a class of conditional mean estimators denoted by $\cF_{L_2}$. Let $F_{\eta}$ be the cumulative distribution function of $\eta$.
\begin{proposition}[Mean extrapolability gain]\label{prop:gain_mean}
	The functional extrapolation uncertainties of the two classes of mean estimators are given by $$\cU_{\cF^\mu_{\mathrm{E}}}(\delta) = L\int_{\etam-\delta}^{\etam}(\eta-\etam+\delta)d F_{\eta}(\eta)$$ and $\cU_{\cF_{L_2}}(\delta) = L\delta$, respectively. The mean extrapolability gain is
	\begin{equation*}
		\gamma^\mu(\delta) := \cU_{\cF^\mu_{L_2}}(\delta) - \cU_{\cF^\mu_{\mathrm{E}}}(\delta) = L\delta F_{\eta}(\etam-\delta) + L\int_{\etam-\delta}^{\etam}(\etam-\eta)d F_{\eta}(\eta)>0,
	\end{equation*} 
	for all $\delta>0$. The maximum extrapolability gain is $\sup_{\delta>0}\gamma^\mu(\delta)=L\etam$, which is achieved when $\delta\ge 2\etam$.
\end{proposition}

Finally, we turn to conditional distribution estimation. By Theorem~\ref{thm:func_recover}, engression naturally provides a class of conditional distributions
\begin{equation*}
	\cP_{\mathrm{E}} := \{P(y|x;g,h^\star,\beta^\star):g\in\cG^\star\},
\end{equation*}
where $P(y|x;g,h^\star,\beta^\star)$ is the conditional distribution of $g(X+h^\star(\varepsilon))+\beta^\star X$ given $X=x$. As to the baseline approach, we consider quantile regression for all quantiles, which leads to a class of conditional distributions defined in terms of quantile functions
\begin{equation*}
	\cP_{\mathrm{QR}}:=\{(q_\alpha(x))_{\alpha\in[0,1]}:q_\alpha\in\cF^\alpha_{\mathrm{QR}},\alpha\in[0,1]\},
\end{equation*}
where $\cF^\alpha_{\mathrm{QR}}$ is class of estimators for the conditional $\alpha$-quantile by the quantile regression, as defined in \eqref{eq:qr_class}.

In the definition of distributional extrapolation uncertainty (Definition~\ref{def:extra_dist}), we consider the Wasserstein-$\ell$ distance as the distance measure for any $\ell\ge1$. For two probability distributions $P$ and $P'$, the Wasserstein-$\ell$ distance between them is defined as $W_\ell(P,P')=(\inf_{\Gamma\in\Gamma(P,P')}\bbE_{(X,X')\sim\Gamma}[\|X-X'\|^\ell])^{1/\ell}$ where $\Pi(P,P')$ denotes the set of all joint distributions over $(X,X')$ whose marginals are $P$ and $P'$, respectively. 
\begin{proposition}[Distributional extrapolability gain]\label{prop:gain_dist}
	The distributional extrapolation uncertainties of the two classes of estimated conditional distributions are given by  
	\begin{equation*}
		\cU^\ell_{\cP_{\mathrm{E}}}(\delta) = L\left(\int_{\etam-\delta}^{\etam}(\eta-\etam+\delta)^\ell d F_{\eta}(\eta)\right)^{1/\ell}
	\end{equation*}
	and $\cU^\ell_{\cP_{\mathrm{QR}}}(\delta) = L\delta$, respectively. Define the distributional extrapolability gain as $\gamma^d_\ell(\delta) := \cU^\ell_{\cP_{\mathrm{QR}}}(\delta) - \cU^\ell_{\cP_{\mathrm{E}}}(\delta)$, which satisfies the following:
	\begin{enumerate}[label=(\roman*)]
		\item For $\ell=1$, the gain is explicitly given by
		\begin{equation*}
		\gamma^d_1(\delta) = L\delta F_{\eta}(\etam-\delta) + L\int_{\etam-\delta}^{\etam}(\etam-\eta)d F_{\eta}(\eta)>0,
		\end{equation*} 
		for all $\delta>0$. The maximum extrapolability gain is $\sup_{\delta>0}\gamma_1^d(\delta)=L\etam$, which is achieved when $\delta\ge 2\etam$.
		\item For $\ell\in(1,\infty)$, we have $0<\gamma^d_\ell(\delta)<\gamma^d_1(\delta)$ for all $\delta>0$, and the maximum extrapolability gain is $\sup_{\delta>0}\gamma^d_\ell(\delta)=\lim_{\delta\to\infty}\gamma^d_\ell(\delta)=L\etam$. 
		\item For $\ell=\infty$, we have $\gamma^d_\infty(\delta)=0$ for all $\delta>0$.
	\end{enumerate}
\end{proposition}

\begin{figure}
\centering
\begin{tabular}{@{}ccc@{}}
	\includegraphics[page=1, clip, trim=28cm 14.5cm 27.5cm 14.5cm, width=0.31\textwidth]{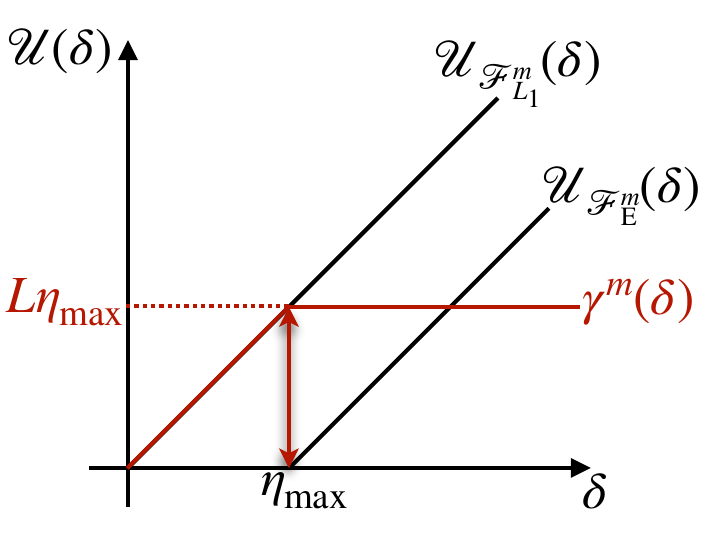} &
	\includegraphics[page=2, clip, trim=28cm 14.5cm 27.5cm 14.5cm, width=0.31\textwidth]{fig/gains} &
	\includegraphics[page=3, clip, trim=28cm 14.5cm 27.5cm 14.5cm, width=0.31\textwidth]{fig/gains}\vspace{-0.in}\\
	\small{(a) Median} & \small{(b) Mean} & \small{(c) Distributional} 
\end{tabular}
\caption{Extrapolation uncertainties and extrapolability gains. The gains for the mean and the distribution depend on the specific noise distribution and here we take $\eta\sim\mathrm{Unif}[-\etam,\etam]$ as an example. Black lines represent extrapolation uncertainties; red lines are extrapolability gains; red arrows represent the maximum extrapolability gains.}\label{fig:gains}
\end{figure}

We illustrate the theoretical formulas of extrapolability gains in Figure~\ref{fig:gains}. 
In summary, all the extrapolability gains for the median, the mean, and the distribution are positive for all $\delta>0$ (except for one case of the distributional extrapolability gain with the Wasserstein-$\infty$ distance). This indicates that engression yields strictly lower extrapolation uncertainties across the entire real line. 
Notably, the maximum extrapolability gains for the median, mean, and the distribution (with $\ell=1$) are the same. However, the maximum extrapolability gain is achieved later for the mean and distribution than for the median. The reason is because the mean and distribution also rely on the higher (or even extreme) quantiles, for which engression has a higher extrapolation uncertainty than that for the median. For distributional extrapolability, when $1<\ell<\infty$, the gain becomes smaller than the one measured with $\ell=1$, but it still remains positive and achieves the same maximum gain at infinity. The extreme case of $\ell=\infty$, however, corresponds to the only measure under which engression does not show an advantage over baseline methods but performs identically.

\section{Finite-sample analysis}\label{sec:finite_sample}
We study the finite-sample behaviour of the engression estimator. We first derive finite-sample bounds for parameter estimation and out-of-support prediction under a special case of quadratic models when the pre-ANM is either well- or misspecified. Then we present the consistency results for general monotone and Lipschitz models. 
\subsection{Quadratic models}\label{sec:quadratic}
Consider the following quadratic pre-ANM class
\begin{equation}\label{eq:quadratic_preanm}
	\{\beta_0+\beta_1 (x+\eta)+\beta_2 (x+\eta)^2:\beta=(\beta_0,\beta_1,\beta_2)\in\cB,\eta\sim P_\eta\in\cP_\eta\},
\end{equation}
where the covariate $x$ and noise variable $\eta$ are univariate, $\cB\subseteq\bbR^3$ is a bounded parameter space such that $\beta_1\beta_2>0$, and $\cP_\eta$ is a class of noise distributions. As such, this is a semiparametric model with a finite-dimensional parameter $\beta$ and infinite-dimensional component $\eta$. Assume all distributions in $\cP_\eta$ are symmetric with a bounded support $[-\etam,\etam]$. Assume the density of $Y=\beta_0+\beta_1 (x+\eta)+\beta_2 (x+\eta)^2$, denoted by $p_{\beta,\eta}(y|x)$, is uniformly bounded away from 0, i.e.\ 
\begin{equation}\label{eq:lb_density}
	b:=\inf_{\substack{x\in\cX,\beta\in\cB,\\\eta\sim P_\eta\in\cP_\eta}}p_{\beta,\eta}(y|x)>0.
\end{equation}

Consider the extreme case where the covariate $X$ takes only two values during training, i.e.\ with the training support $\cX=\{x_1,x_2\}$, where $0<x_1<x_2$; assume $x_1>\etam$ and that $X$ takes each value with probability $0.5$. The training data consists of $(x_1,Y_{1,i}),i=1,\dots,n$ and $(x_2,Y_{2,i}),i=1,\dots,n$. As we focus on the statistical behaviour concerning the sample size $n$ and one can in principle generate as many samples of the noise variable as possible without the worry of computational complexity, we assume the noise sample size $m$ in \eqref{eq:eng_emp} tends to infinity. 
Denote the engression estimator in \eqref{eq:eng_emp} based on the above sample by $\hat\beta=(\hat\beta_0,\hat\beta_1,\hat\beta_2)$ and $\hat\eta$. 

We are interested in predicting the conditional mean or quantile of $Y$ given $X=x$ for any point $x\in\bbR$, particularly those lying outside the support $\cX$. 

\subsubsection{Well-specified pre-ANM}
We first consider the case where the true data generating model is a pre-ANM, so that the pre-ANM class \eqref{eq:quadratic_preanm} is well specified. Specifically, let $\beta^\star=(\beta^\star_0,\beta^\star_1,\beta^\star_2)\in\cB$ be the true parameter and $\eta^\star\sim P^\star_\eta\in\cP_\eta$ be the true noise variable. 

We first show the failure in extrapolation of $L_2$ regression \eqref{eq:l2} and quantile regression \eqref{eq:qr} among the model class $\{\beta_0+\beta_1x+\beta_2x^2:\beta\in\cB\}$. Define the set of $L_2$ regression estimators by 
\begin{equation*}
	\cB^\mu = \argmin_\beta \frac{1}{2n}\sum_{j=1}^2\sum_{i=1}^n[Y_{j,i}-(\beta_0+\beta_1x_j+\beta_2x_j^2)]^2
\end{equation*}
and the set of quantile regression estimators given a level $\alpha$ by
\begin{equation*}
	\cB^{\textsl{q}}_\alpha = \argmin_\beta \frac{1}{2n}\sum_{j=1}^2\sum_{i=1}^n\rho_\alpha(Y_{j,i}-(\beta_0+\beta_1x_j+\beta_2x_j^2)).
\end{equation*}
Note that given two values $x_1,x_2$ of the covariate, both $L_2$ regression and quantile regression do not have a unique solution, i.e.\ sets $\cB^\dag$ and $\cB^\ddag$ are not singleton. Let $\mu^\star(x)$ be the true conditional mean of $Y$ given $X=x$ and $q^\star_\alpha(x)$ be the true conditional $\alpha$-quantile. 

The following proposition suggests that $L_2$ regression and quantile regression yield arbitrarily large errors for conditional mean and quantile prediction for out-of-support data. The proof is given in Appendix~\ref{app:pf_finite_sample_reg}.
\begin{proposition}\label{prop:finite_sample_reg}
For all $x\notin\cX$, we have
\begin{equation*}
	\sup_{\beta\in\cB^\mu}\bbE[(Y-(\beta_0+\beta_1x+\beta_2x^2))^2]=\infty,
\end{equation*}
and for all $\alpha\in[0,1]$
\begin{equation*}
	\sup_{\beta\in\cB^{\textsl{q}}_\alpha}|(q^\star_\alpha(x)-(\beta_0+\beta_1x+\beta_2x^2)|=\infty.
\end{equation*}
\end{proposition}

Now we turn to engression. Let $\hat{q}_\alpha(x)$ be the engression estimator for the conditional $\alpha$-quantile of $Y$ given $X=x$, which is defined as the $\alpha$-quantile of $\hat\beta_0+\hat\beta_1(x+\hat\eta)+\hat\beta_2(x+\hat\eta)^2$ in terms of $\hat\eta$. Let $\hat\mu(x)$ be the engression estimator for the conditional mean, defined as $\bbE_{\hat\eta}\big[\hat\beta_0+\hat\beta_1 (x+\hat\eta) + \hat\beta_2 (x+\hat\eta)^2\big]$.
Theorem~\ref{thm:finite_sample_eng} presents the finite-sample bounds for estimation errors of the parameter as well as the conditional mean and quantiles. 
Notably, in contrast to the failure of $L_2$ regression and quantile regression, our engression yields consistent estimators for the conditional mean and quantiles for data points both within or outside the support, as shown in \eqref{eq:bound_mean_est} and \eqref{eq:bound_quantile_est}. 
The proof is given in Appendix~\ref{app:pf_thm_finite_sample_eng}.

\begin{theorem}\label{thm:finite_sample_eng}
Let $\delta>0$. Under the setting posited above, with probability exceeding $1-\delta$, we have
\begin{equation}\label{eq:bound_param_est}
	\|\hat\beta - \beta^\star\| \le \frac{C_1}{(x_2-x_1)b^{\frac{2}{3}}}\left(\frac{\log(2/\delta)}{n}\right)^{\frac{1}{3}},
\end{equation}
where $b$ is defined in \eqref{eq:lb_density}, and $C_1$ is a constant depending on the parameter space $\cB$ and the noise level $\etam$. 
In addition, for any $x\in\bbR$, it holds with probability exceeding $1-\delta$ that
\begin{equation}\label{eq:bound_mean_est}
	(\hat\mu(x)-\mu^\star(x))^2 \le C_2\max\{1,|x|,x^2\}\left(\frac{\log(2/\delta)}{n}\right)^{\frac{2}{3}},
\end{equation}
where $C_2<\infty$ is a constant independent of $x$. For any $x\in\bbR$ and $\alpha\in[0,1]$, it holds with probability exceeding $1-\delta$ that
\begin{equation}\label{eq:bound_quantile_est}
	|\hat{q}_\alpha(x)-q^\star_\alpha(x)| \le C_3\max\{1,|x|,x^2\}|Q^{\eta^\star}_\alpha|\left(\frac{\log(2/\delta)}{n}\right)^{\frac{1}{3}},
\end{equation}
where $C_3<\infty$ is a constant independent of $\alpha$ and $x$.
\end{theorem}

Taking a closer look at the error rate for parameter estimation in \eqref{eq:bound_param_est}, the constant is monotonically decreasing with respect to $x_2-x_1$, indicating that the engression estimator would yield faster convergence as the training support of the covariate becomes more diverse. Also the constant $C_1$ is in fact monotonically decreasing with respect to the noise level $\etam$. In addition, the dependence of the rates for the conditional mean and quantile estimation suggest that it tends to be more challenging to extrapolate on a point farther away from training data or for extreme quantiles. These findings align with our previous theoretical results at the population level as well as the empirical observations presented below. 

Furthermore, due to the nonparametric component $\eta$, engression yields a nonparametric rate of the order $\cO(n^{-\frac{1}{3}})$. However, in this case, a simple modification of engression that enforces the conditional median to take the parametric form $\beta_0+\beta_1 x + \beta_2 x^2$ (e.g.\ by adding a regularisation term or imposing parametric assumptions on $\eta$) can lead to the parametric rate for $\beta$ and conditional median estimation. It remains an open problem whether engression can achieve a parametric rate more generally.

\subsubsection{Misspecified pre-ANM}
We further investigate the case with model misspecification, in particular, when the pre-ANM assumption, which plays a key role in engression, is violated.
Consider the same setting as above with the only difference being that the training data is now assumed to be generated according to $Y=\beta^\star_0+\beta^\star_1x+\beta^\star_2x^2+\eta^\star$ which is in fact a post-ANM. However, we continue to adopt engression with the pre-ANM class defined in \eqref{eq:quadratic_preanm} which does not contain the true model. 

The following proposition presents the bounds for the parametric component $\hat\beta$ of the engression estimator, which is proved in Appendix~\ref{app:pf_thm_finite_sample_eng_mis}.
\begin{proposition}\label{prop:finite_sample_eng_mis}
	Under the setting posited above, with probability exceeding $1-\delta$, we have
	\begin{equation*}
		\max\big\{|\hat\beta_0 - (\beta^\star_0-\beta^\star_2 x_1x_2)|,\ |\hat\beta_1 - (\beta^\star_1+\beta^\star_2(x_1+x_2))|,\ |\hat\beta_2|\big\} \lesssim \left(\frac{\log(2/\delta)}{n}\right)^{\frac{1}{3}}.
	\end{equation*}
\end{proposition}
We note that the engression estimator is generally inconsistent unless the true model is linear (i.e.\ $\beta^\star_2=0$), which is reasonable since the model is misspecified. More interestingly, engression still yields a unique solution, in contrast to the non-uniqueness of regression. Furthermore, engression leads approximately to a linear model with $\hat\beta_2\pto0$ as $n\to\infty$. 
As we have only observations at two distinct values of the predictor variable, a linear fit seems to be a useful fallback option in case of model misspecification.
In comparison, traditional regression methods again exhibit arbitrarily large errors outside the training support, as shown in Proposition~\ref{prop:finite_sample_reg}.

\subsection{General pre-ANMs}
We extend the previous setting and focus on a general pre-ANM class $\{g(x+h(\varepsilon)):g\in\cG,h\in\cH\}$ defined in Section~\ref{sec:theory_extra} without the additional linear term. We consider a continuous support $\cX=[x_{\min},x_{\max}]$ which is a bounded interval on $\bbR$. Without loss of generality, we assume that $\varepsilon\sim\mathrm{Unif}[0,1]$ and $\cH$ is the same as posited in Section~\ref{sec:extrapolability_anm}. 
Let $g^\star(x+h^\star(\varepsilon))$ be the true model. 
Let $n$ be the training sample size.

The following theorem states the consistency of the engression estimator $(\hat{g},\hat{h})$. Notably, the set of covariate values for which the consistency of $\hat{g}$ is achieved, i.e.\ $\tilde\cX$, extends beyond the support up to a range depending on the noise level. In addition, the consistency of $\hat{h}$ implies that engression can consistently estimate the cdf of the noise distribution. The proof is given in Appendix~\ref{app:thm_consistency_general}. 
\begin{theorem}\label{thm:consistency_general}
	Assume the following conditions hold:
	\begin{enumerate}[label=(\textit{B\arabic*})]
	\setlength{\itemsep}{2pt}
	\setlength{\parskip}{2pt}
		\item $x_{\max}-x_{\min}\ge \max\{|h^\star(1)|,|h^\star(0)|\}$. \label{ass:large_x_space}
		\item $\cG$ is uniformly Lipschitz and strictly monotone.\label{ass:g_lip_mono}
		\item $X$ lies on grid points $\xmin=x_0<x_1<\dots<x_m=\xm$, where $m=m(n)$ such that $m(n)\to\infty$, $m(n)/n\to0$, and $\max_{i=1,\dots,m}|x_i-x_{i-1}|=o(m)$ as $n\to\infty$, and $\bbP(X=x_i)\ge c/m$ for some $c>0$. \label{ass:grid}
	\end{enumerate}
	Then it holds for all $\tilde{x}\in\tilde\cX:=\{x+h^\star(\varepsilon):x\in\cX, \varepsilon\in[0,1]\}$ and $\varepsilon\in[0,1]$ that  
	\begin{equation*}
		\hat{g}(\tilde{x})\pto g^\star(\tilde{x})\quad \text{and} \quad\hat{h}(\varepsilon) \pto h^\star(\varepsilon)\quad \text{as $n\to\infty$.}
	\end{equation*}
\end{theorem}
Condition~\ref{ass:large_x_space} requires the training support of the covariate exceed the noise level. Condition~\ref{ass:g_lip_mono} imposes assumptions on the function class. It is worth noting that with the Lipschitz condition, we establish the consistency for covariate values outside the support, which traditional regression approaches, however, do not attain. For example, under the same assumptions, $L_1$ regression consistently exhibits an error for conditional median estimation on out-of-support data that scales linearly with respect to the distance away from the support boundary, as indicated in Proposition~\ref{prop:gain_median}. Condition~\ref{ass:grid} is a technical assumption for simplifying the analysis and we think the same results could be shown for $X$ with a general continuous density that is bounded away from 0.

\section{\revise{Related work}}\label{sec:related}
\subsection{Extrapolation and distribution shifts}\label{sec:related_extrap}
The theoretical understanding of extrapolation for nonlinear models is quite limited. 
Earlier studies, such as the one by \citet{sugiyama2007covariate}, were predicated on the assumption of a bounded density ratio between the training and test distributions. This assumption implies that the test support is nested within the training support, thus excluding the out-of-support extrapolation scenarios that we aim to address. \citet{christiansen2021causal} introduced an extrapolation method that presupposes linearity beyond the support, thus imposing stringent constraints on the function class. More recently, \citet{dong2023first} explored extrapolation within the class of additive nonlinear functions. They evaluated the average performance of a model under a newly extended distribution, assuming that the marginals are not substantially shifted. However, they did not study the behaviour of the model at distinct test points beyond the training support, an aspect often referred to as point-wise performance.

Extrapolation is also linked to the extensive body of work on domain adaptation or generalisation, with the primary objective being the development of models that can adapt or generalise effectively to potentially divergent test distributions. Such methods typically depend on the integration of a certain amount of data from the test distribution~\citep{ben2006analysis,pmlr-v37-ganin15,chen2021domain}, adjusting the test distribution via covariate shift compensation \citep{sugiyama2007covariate,gretton2009covariate}, or exploiting specific inherent data structures to train a model that is invariant or robust to distribution shifts~\citep{peters2016causal,arjovsky2019invariant,rothenhausler2021anchor}.
However, our work diverges from these methodologies by tackling a challenging problem: how a model trained on data within a bounded set can extrapolate to a data point that lies outside the training support. We do not specify the test distribution, nor do we utilise data from the test domain. \revise{There is also an extensive body of literature on distributionally robust learning, where the goal is to minimise the worst case risk among a family of potential test distributions. The family of distributions can be defined via a Wasserstein ball~\citep{kuhn2019wasserstein,sinha2018certifiable}, MMD ball~\citep{jegelka19distributionally, kirschner20distributionally}, or general $f$-divergences~\citep{NIPS2016_4588e674}. In contrast to these, we assume that the conditional distribution of the target, given the covariates, is unchanged for test data, and we examine the properties of the proposed methodology once the covariate distribution is shifting and has at test time a positive density outside of the training support. While the setting is different in this sense, we also show that the desired extrapolation properties depend critically on fitting a distributional regression model as opposed to a conditional mean or quantile regression model.  There is also extensive literature on distributional robustness where the assumption is that we have at training time access to samples from several different distributions and we try to use the heterogeneity in the training distributions to learn a robust model \citep{meinshausen2015maximin,Sagawa2020Distributionally,mehta2024distributionally}. In contrast, here we assume we have only access to i.i.d.\ samples from a single training distribution. In fact, our proposed method is also of interest for the setting where training and test data are both sampled i.i.d.\ from the same distribution as it provides a novel way for distributional regression with neural networks that does not require much hyperparameter tuning. The smoothing effect of the pre-additive noise model bears also resemblance with the positive effect of smoothing on adversarial robustness in a classification setting  
\citep{cohen2019certified}.

While engression can provide prediction intervals with conditional coverage, conformal prediction is a popular possibility for guaranteeing marginal coverage 
 of prediction intervals \citep{shafer2008tutorial,barber2023conformal}. The standard setting of conformal prediction requires training and test data to follow the same distribution but there have been attempts to provide marginal coverage guarantees under distribution shifts 
\citep{tibshirani2019conformal,gibbs2021adaptive,fannjiang2022conformal}. These approaches broadly rely on an appropriate reweighting scheme of samples to correct for the distributional shift. However, if the test distribution has a positive density in a region of the covariate space where the density of the training distribution tends to zero, the relevant weights will tend to infinity, making the distributional shift correction infeasible for conformal prediction. 
}

\subsection{Distributional regression}\label{sec:related_distreg}
Both parametric and non-parametric approaches have been developed for distributional regression through estimating the cumulative distribution function~\citep{foresi1995conditional,hothorn2014conditional}, density function~\citep{dunson2007bayesian}, quantile function~\citep{meinshausen2006quantile}, etc. See a recent review of distributional regression approaches in \citet{kneib23review}.



\revise{From the modelling perspective, engression aligns more with the recently advanced field of generative models which aims at learning to sample from the training data distribution via a transformation of a simple distribution. Over the past decade, numerous methods have been developed and shown to achieve remarkable empirical success in generating images \citep{rombach2022high,ramesh2022hierarchical} or texts \citep{brown2020language}. 
We think that one of the most valuable advantages of engression over 
existing deep generative models is its simplicity in both statistical and computational sense, yet performing well in our numerical studies for classical regression tasks. The statistical property of engression is relatively easy to understand, as demonstrated through our finite-sample results in \eqref{sec:finite_sample}. 
Compared to generative adversarial networks \citep{goodfellow2014}, engression solves a minimisation problem as in \eqref{eq:eng_emp} which is much easier computationally than solving a min-max optimisation problem. Engression does not need an additional discriminator model, resulting in much less hyperparameter tuning and more stable training. Compared to variational autoencoders \citep{kingma2013auto}, engression does not rely on variational approximation, does not need an additional encoder model, and allows generative models \eqref{eq:gen_model_class} which are more general than the Gaussian or Bernoulli decoders typically adopted in VAE. While an engression model in general does not lead to a closed-form density, it allows more expressive neural network models than the invertible NNs used in normalizing flows \citep{papamakarios2021normalizing} and its optimisation does not involve expensive Jacobian computation. In contrast to diffusion models \citep{sohl2015deep,ho2020denoising}, sampling from an engression model is done through a single forward pass of the model rather than iterating over a stochastic process, thus reducing the computational cost. 
In summary, these popular generative models have been focusing on image and text generation as their main application, while engression is more closely aligned with classical regression and its simplicity makes it particularly suitable in this context.
}

\subsection{Post-nonlinear models}\label{sec:related_model}
It is important to note that pre-ANMs bear certain resemblances to the post-nonlinear models featured in literature. In the context of nonlinear independent component analysis, given an independent source vector $S$, the post-nonlinear mixture~\citep{taleb1999source} hypothesises that the observed mixture $X$ is  formed through component-wise nonlinearities applied to the components of the linear mixture, represented as:
\begin{equation*}
X_i=g_i([WS]_i),
\end{equation*}
where $i=1,\dots,d$, $W$ is a linear transformation, and $g$ is a component-wise invertible nonlinear function. However, in contrast to our context, this is a deterministic process with the distinct aim of recovering the hidden signal $S$ from $X$.

In causality, \citet{Zhang2009OnTI} extended the post-nonlinear mixture to the post-nonlinear causal model, which models the causal relationship between a variable $Y$ and its direct cause $X$ as:
\begin{equation*}
Y = g_2(g_1(X)+\varepsilon),
\end{equation*}
where $g_1$ denotes the nonlinear causal effect, $g_2$ the invertible post-nonlinear distortion in $Y$, and $\varepsilon$ the independent noise. The crucial difference between this model and the pre-ANM is whether the noise $\varepsilon$ is directly added to $X$, up to a linear transformation. This attribute forms the cornerstone of pre-ANMs, which is not the case for post-nonlinear models unless $g_1$ is linear. Besides, the pre-ANM allows for additional linear terms that do not show up in a post-nonlinear model.

\section{Empirical results}\label{sec:empirical}
We start with simulated scenarios without model misspecification in Section~\ref{sec:simu} and then extend to the real-world scenarios without guarantees for the fulfillment of the model assumptions in Section~\ref{sec:realdata}. 
\subsection{Simulations}\label{sec:simu}
We simulate data from a heteroscedastic noise model that falls into the pre-ANM class in \eqref{eq:preanm_uni}, with details given in Table~\ref{tab:simu_set}.
 In all settings, the training data are supported on a bounded set and the true functions are nonlinear outside the support. We compare engression to the traditional $L_1$ and $L_2$ regression; for all methods, we use the identical implementation setups including architectures of neural networks and all hyperparameters of the optimisation algorithm.  For each setting, we randomly simulate data and apply the methods with random initialisation, which is repeated for 20 times. All the experimental details are described in Appendix~\ref{app:exp_detail}.

\begin{table}
\centering
\caption{Simulation settings. In all settings, the training data are supported on a bounded set within the boundary value $\xm=2$. With a finite sample, the noise also has a bounded support with roughly $\etam\approx2$ (approximately 97.5\% quantile of a standard Gaussian). All functions $g^\star$ are nonlinear outside the support except that the softplus function is approximately linear as $x$ grows large. }\label{tab:simu_set}
\vskip 0.1in
\begin{tabular}{lllll}
\toprule
Name & \multicolumn{1}{c}{$g^\star(\cdot)$} & \multicolumn{1}{c}{$X$} & \multicolumn{1}{c}{$\eta$} \\\midrule
\texttt{softplus} & $g^\star(x)=\log(1+e^x)$ & Unif$[-2,2]$ & $\cN(0,1)$ \\
\texttt{square} & $g^\star(x)=(x_+)^2/2$ & Unif$[0,2]$ & $\cN(0,1)$ \\
\texttt{cubic} & $g^\star(x)=x^3/3$ & Unif$[-2,2]$ & $\cN(0,1.1^2)$ \\
\texttt{log} & $g^\star(x)=\begin{cases}\frac{x-2}{3}+\log(3)&x\le2\\\log(x)&x>2\end{cases}$ & Unif$[0,2]$ & $\cN(0,1)$ \\
\bottomrule
\end{tabular}
\end{table}

We are interested in the performance of the estimated models outside the training support. To evaluate this, we use both qualitative visualisations and quantitative metrics. While we focus here on extrapolation beyond the larger end of the support for convenience, we expect similar phenomena to happen for the smaller side as well; we investigate extrapolation on both ends of the training support in real-data experiments.  

Figure~\ref{fig:simu_visual} visualises the overall performance of different methods, where we plot the fitted curves as well as the true function. Here we are concerned with the estimation of the conditional median and mean functions of $Y$ given $X=x$. The true median function is given by $g^\star(x)$, while the true mean function is given by $\bbE_{\eta}[g^\star(x+\eta)]$, which is estimated based on the sampling $10^5$ i.i.d.\ draws of $\eta$. The engression estimators for the median and mean are estimated from 512 samples from the engression model for each $x$. We observe that within the training support, all methods perform almost perfectly. However, once going beyond the support, $L_1$ and $L_2$ regression fail drastically and tend to produce uncontrollable predictions, resulting in a huge spread that represents roughly the extrapolation uncertainty. In contrast, the extrapolation uncertainty, as shown by the shaded area, of engression is significantly smaller. Particularly when $x\le4$, which roughly corresponds to $x\le\xm+\etam$, engression leads to nearly 0 extrapolation uncertainty for the conditional median, which supports our local extrapolability results in Section~\ref{sec:local_extrap}. The extrapolation uncertainty of engression for the conditional mean tends to be slightly higher than that for the conditional median, which coincides with the results in Propositions~\ref{prop:gain_median} and \ref{prop:gain_mean}.

\begin{figure}
\centering
\begin{tabular}{@{}c@{}c@{}c@{}c@{}c@{}}
	&\multicolumn{2}{c}{Conditional median}&\multicolumn{2}{c}{Conditional mean}\\
	&\small{Engression} & \small{$L_1$ regression} & \small{Engression} & \small{$L_2$ regression}\\
	\rotatebox[origin=c]{90}{\small{\texttt{softplus}}}&
	\includegraphics[align=c,width=0.24\textwidth]{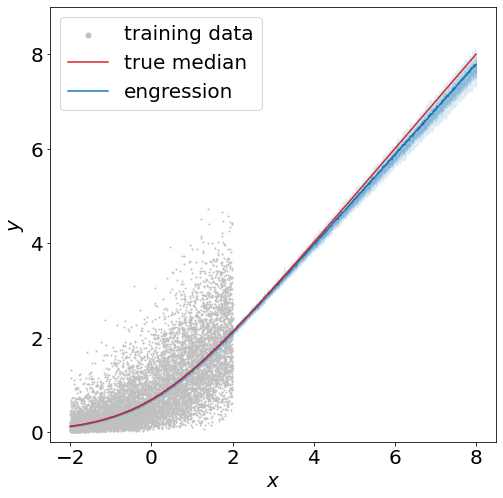} &
	\includegraphics[align=c,width=0.24\textwidth]{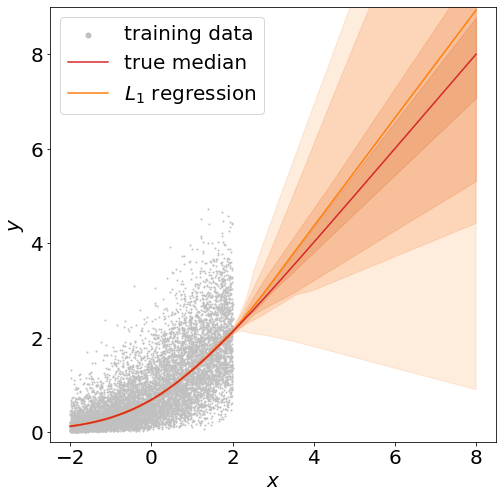} &
	\includegraphics[align=c,width=0.24\textwidth]{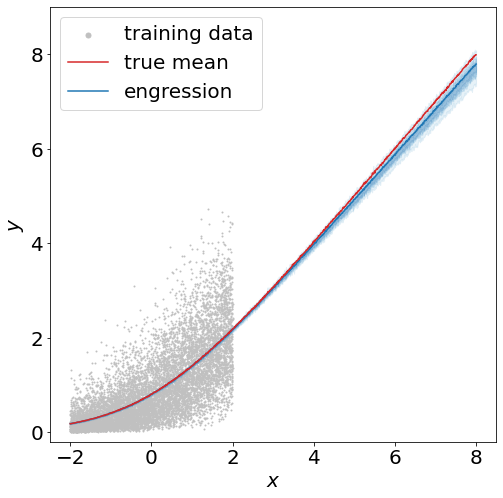} &
	\includegraphics[align=c,width=0.24\textwidth]{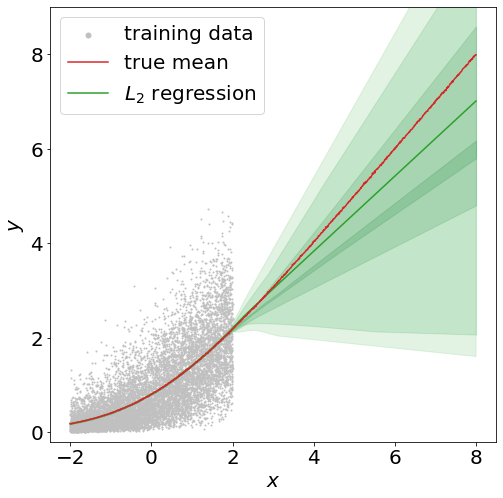}\\
	\rotatebox[origin=c]{90}{\small{\texttt{square}}}&
	\includegraphics[align=c,width=0.24\textwidth]{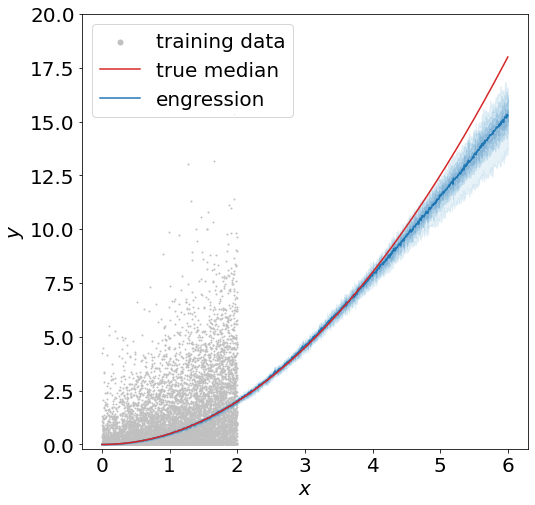} &
	\includegraphics[align=c,width=0.24\textwidth]{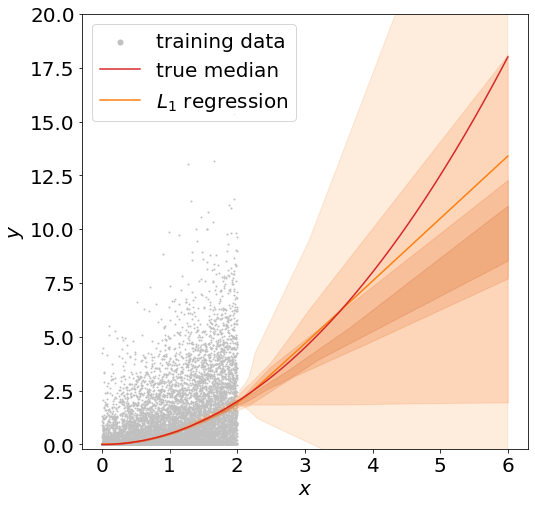} &
	\includegraphics[align=c,width=0.24\textwidth]{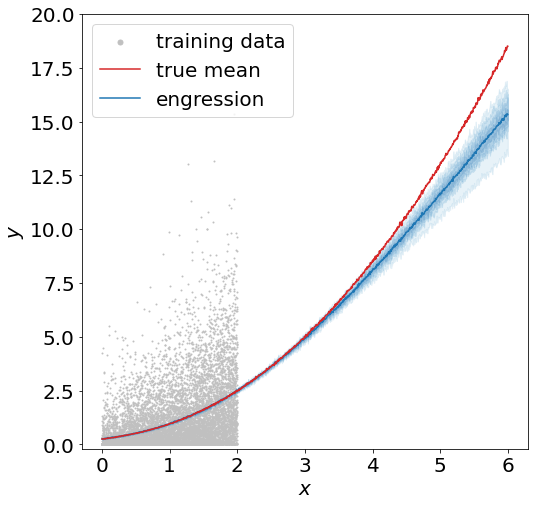} &
	\includegraphics[align=c,width=0.24\textwidth]{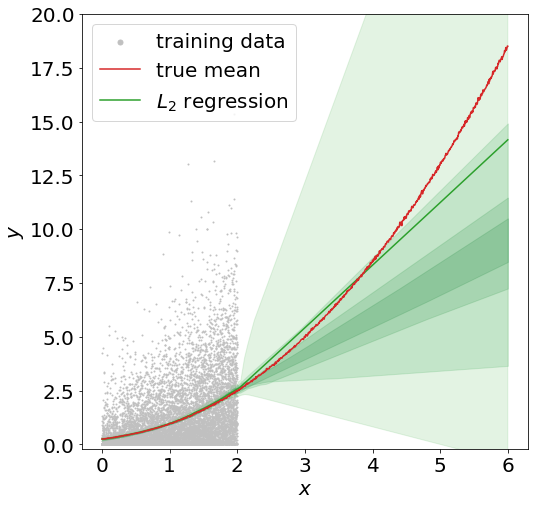}\\
	\rotatebox[origin=c]{90}{\small{\texttt{cubic}}}&
	\includegraphics[align=c,width=0.24\textwidth]{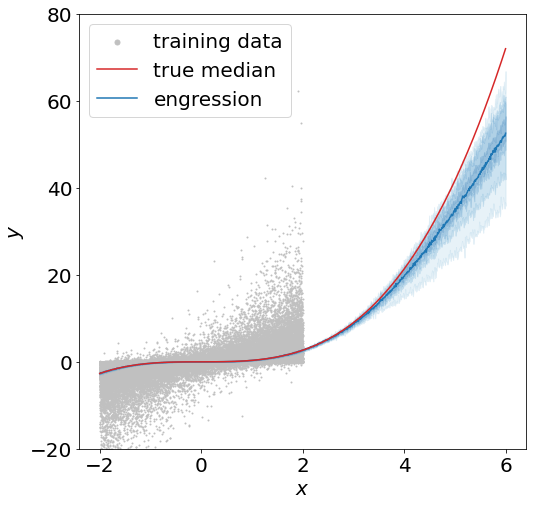} &
	\includegraphics[align=c,width=0.24\textwidth]{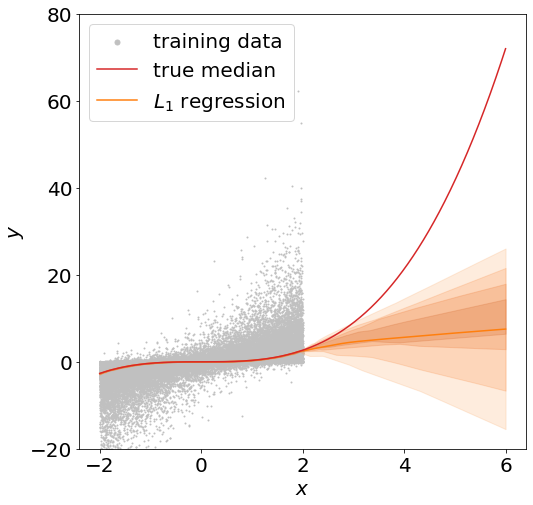} &
	\includegraphics[align=c,width=0.24\textwidth]{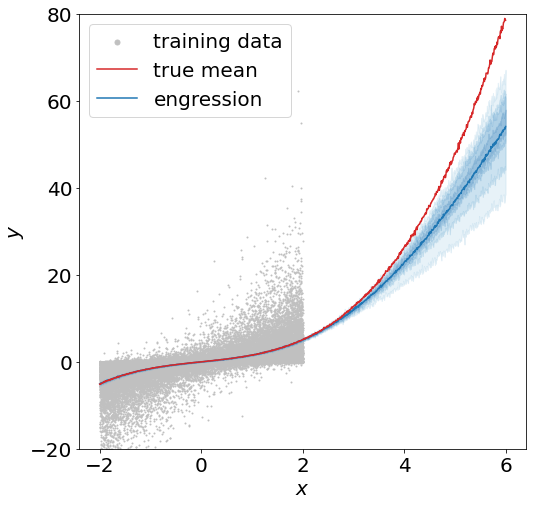} &
	\includegraphics[align=c,width=0.24\textwidth]{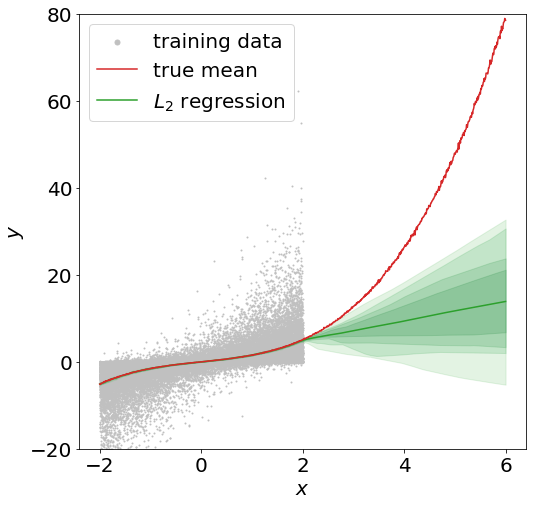}\\
	\rotatebox[origin=c]{90}{\small{\texttt{log}}}&
	\includegraphics[align=c,width=0.24\textwidth]{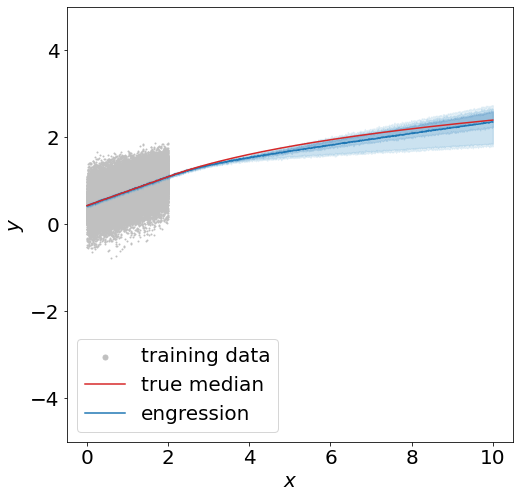} &
	\includegraphics[align=c,width=0.24\textwidth]{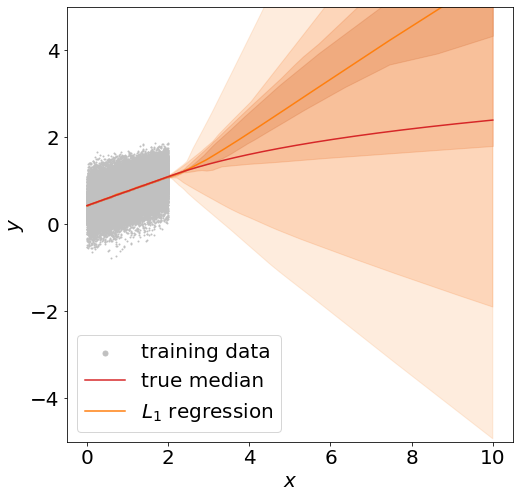} &
	\includegraphics[align=c,width=0.24\textwidth]{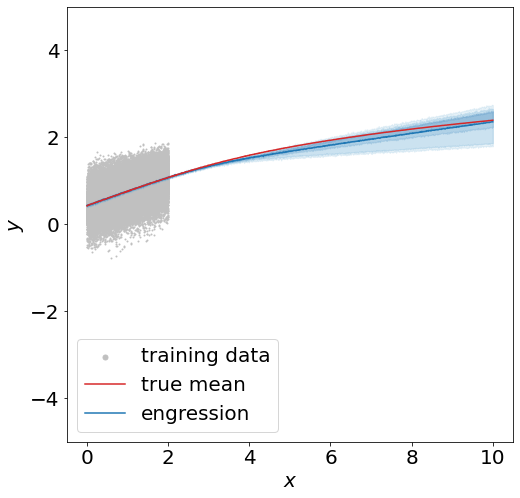} &
	\includegraphics[align=c,width=0.24\textwidth]{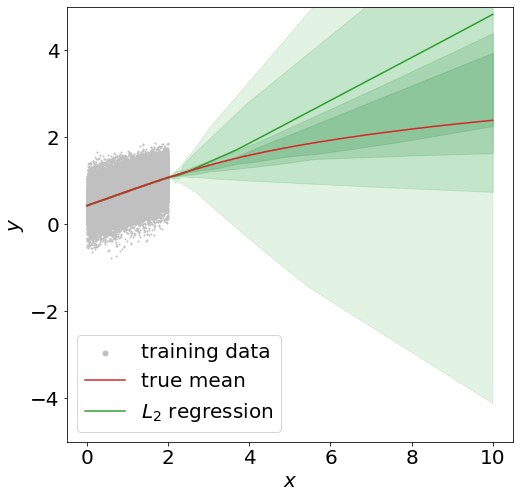}\\
\end{tabular}
\caption{The estimated conditional medians and means of different methods in all simulation settings. Each figure consists of the estimated functions for 20 random repetitions (whose 10\% to 90\% quantiles are plotted in graduated colors with the darkest curve in the middle representing the mean of the 20 estimated mean functions), the true function, and training data. }
\label{fig:simu_visual}
\end{figure}

\begin{figure}
\begin{tabular}{@{}c@{}c@{}c@{}c@{}}
	\texttt{softplus} & \texttt{square} & \texttt{cubic} & \texttt{log} \\
	\includegraphics[align=c,width=0.24\textwidth]{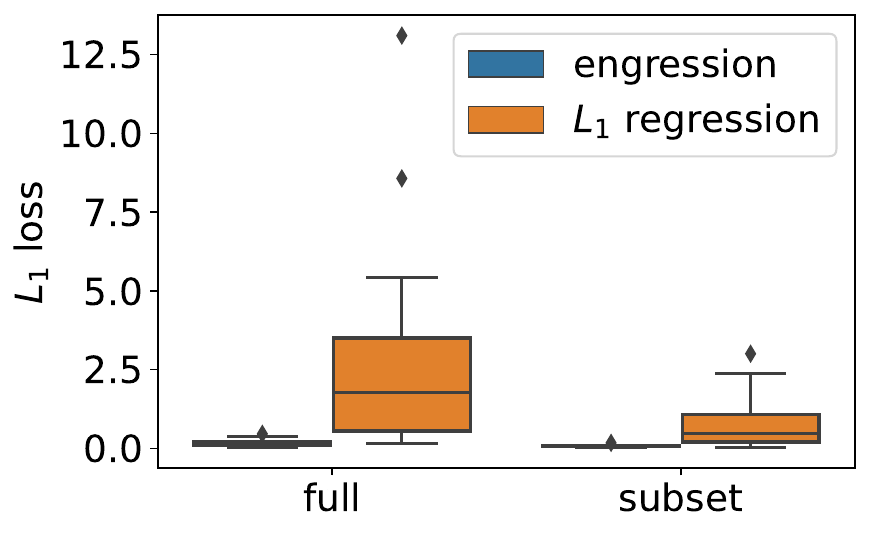} &
	\includegraphics[align=c,width=0.24\textwidth]{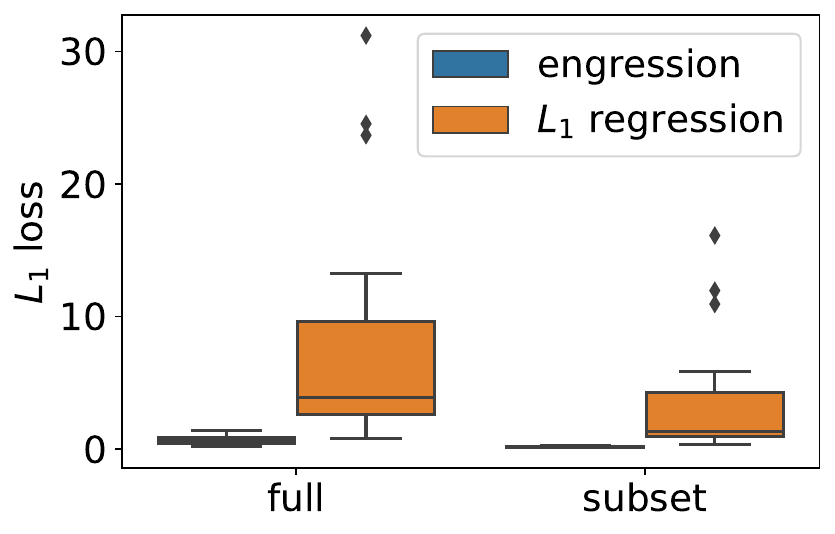} &
	\includegraphics[align=c,width=0.24\textwidth]{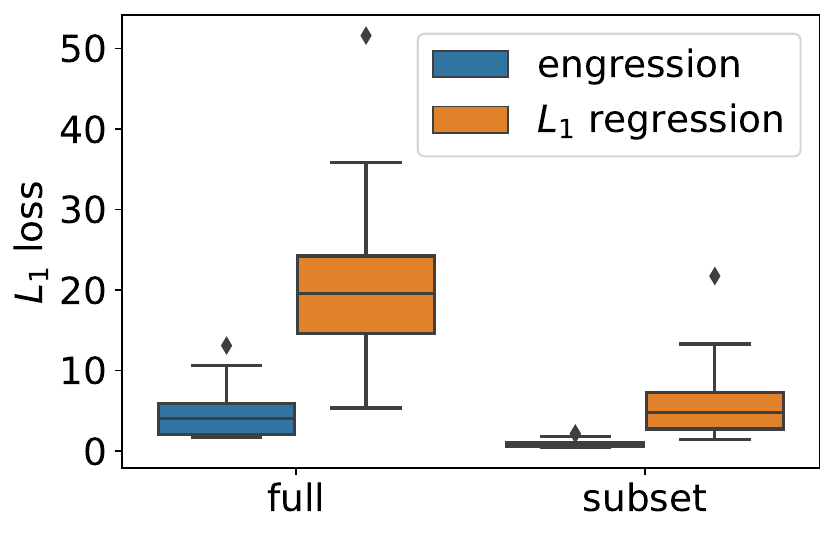} &
	\includegraphics[align=c,width=0.24\textwidth]{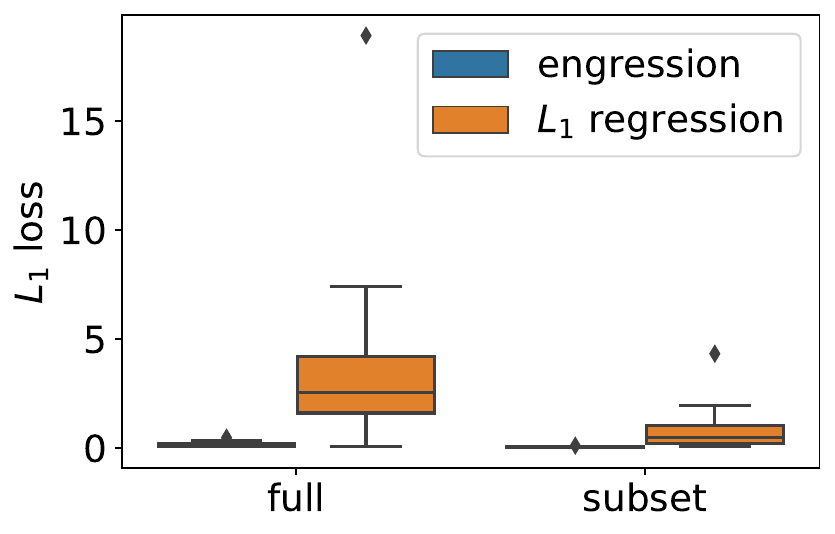} \\
	\includegraphics[align=c,width=0.24\textwidth]{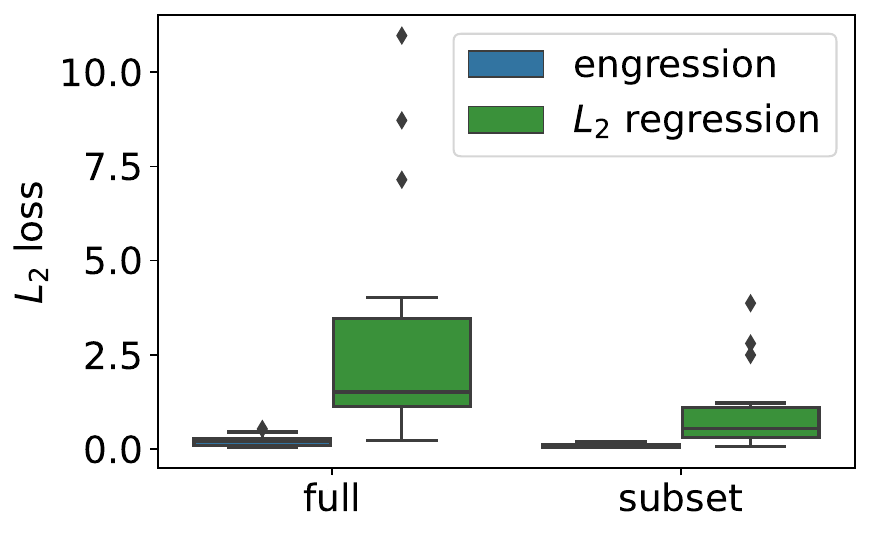} &
	\includegraphics[align=c,width=0.24\textwidth]{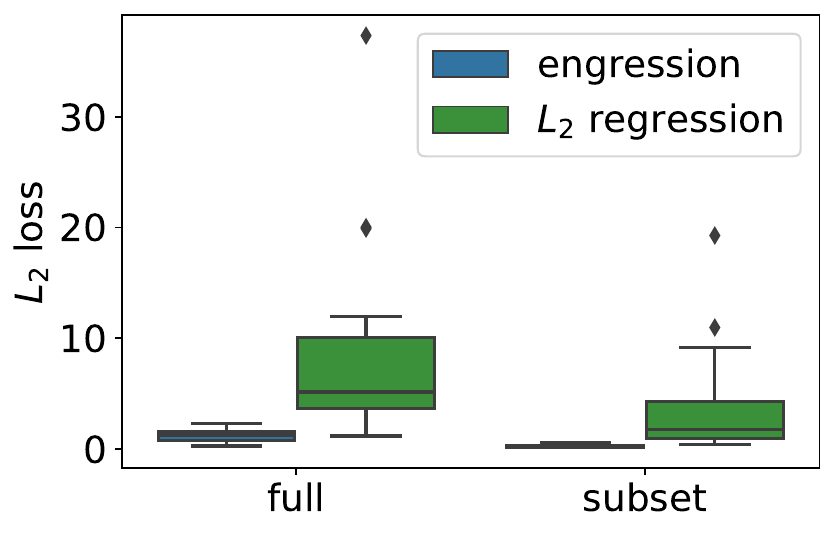} &
	\includegraphics[align=c,width=0.24\textwidth]{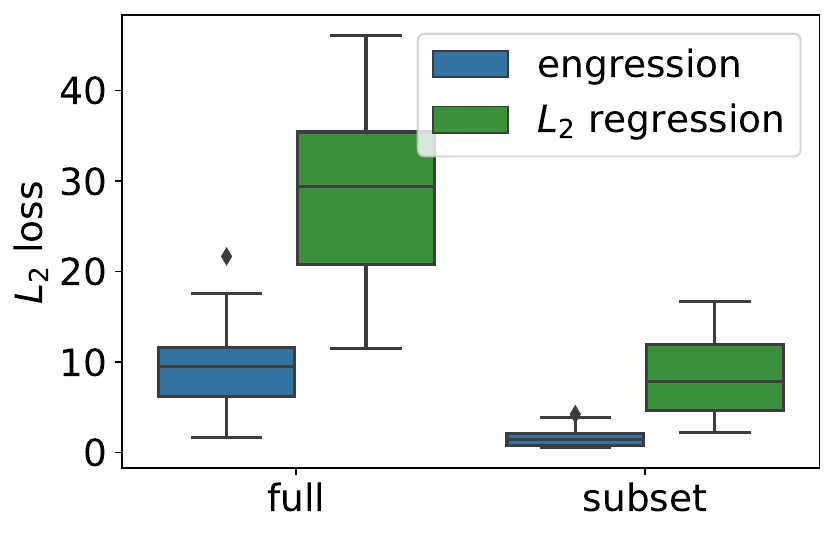} &
	\includegraphics[align=c,width=0.24\textwidth]{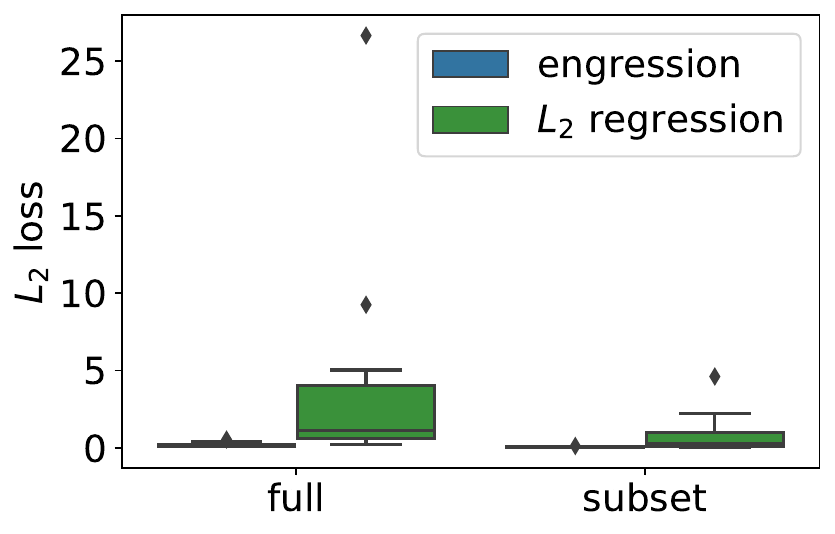} 
\end{tabular}
\caption{Average performance on out-of-support data. The top row shows the $L_1$ loss for conditional median estimation; the bottom row shows the $L_2$ loss for conditional mean estimation.}
\label{fig:simu_boxplot}
\end{figure}

\begin{figure}
\begin{tabular}{@{}cc@{}}
	\texttt{square} & \texttt{log} \\
	\includegraphics[align=c,width=0.49\textwidth]{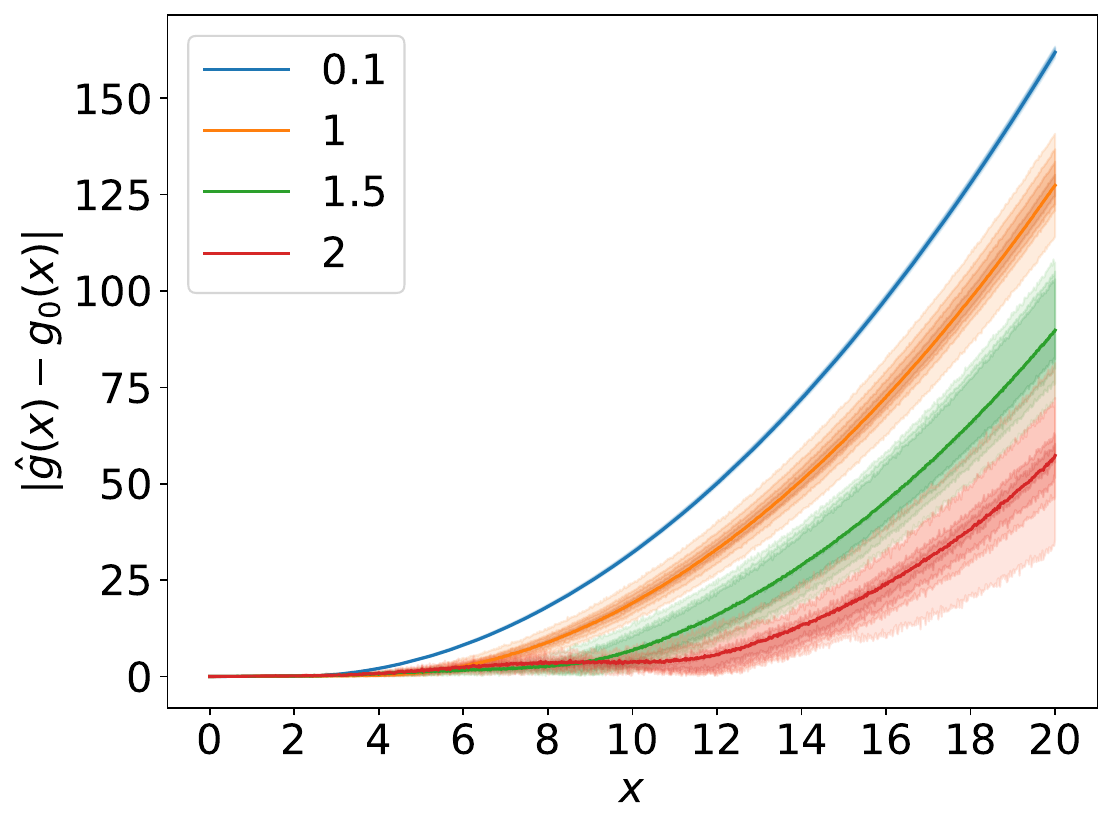} &
	\includegraphics[align=c,width=0.47\textwidth]{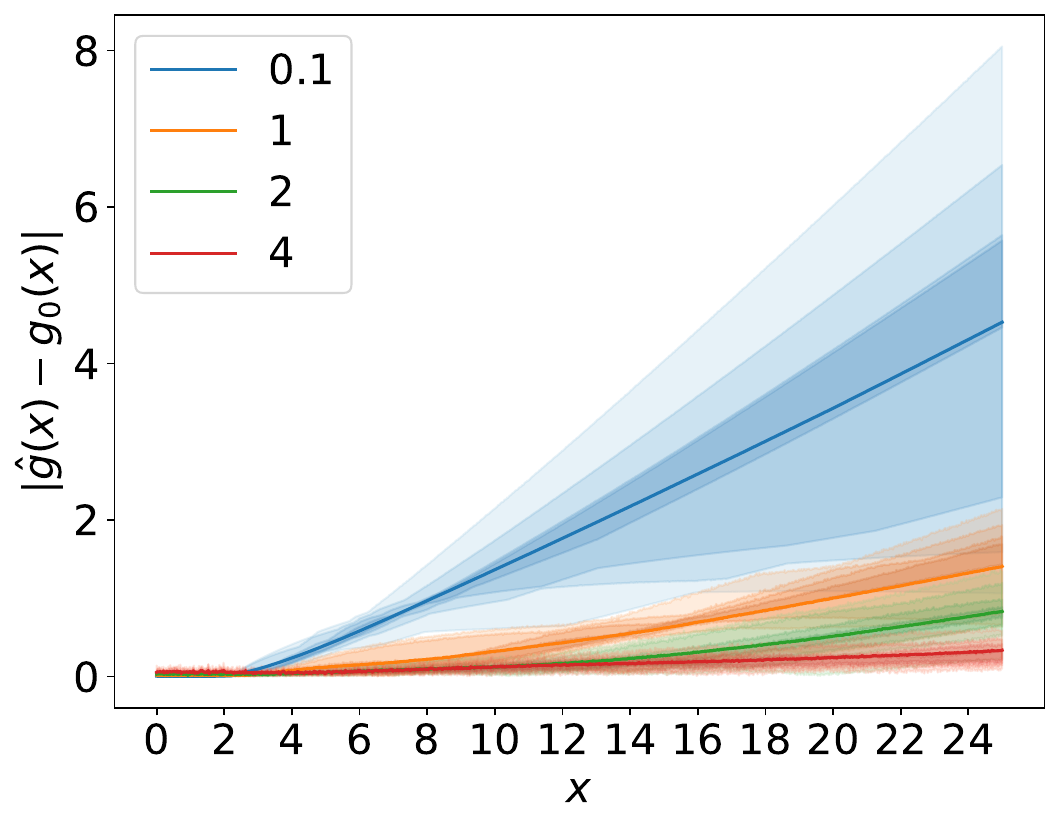} 
\end{tabular}
\caption{Extrapolation performance for varying noise levels. The training support is $[0,2]$. On the $y$-axis we compute the pointwise absolute error between the estimated and true conditional medians for each $x$ inside or outside the support. Each colour stands for a particular standard deviation of the noise $\eta$.}
\label{fig:simu_noise_level}
\end{figure}

\begin{figure}
\begin{tabular}{@{}c@{}c@{}c@{}c@{}}
	\includegraphics[align=c,width=0.24\textwidth]{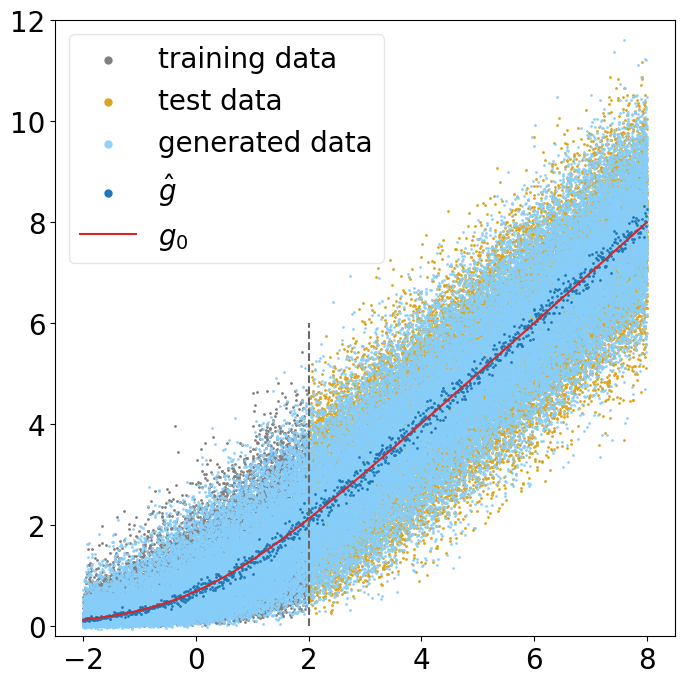} &
	\includegraphics[align=c,width=0.24\textwidth]{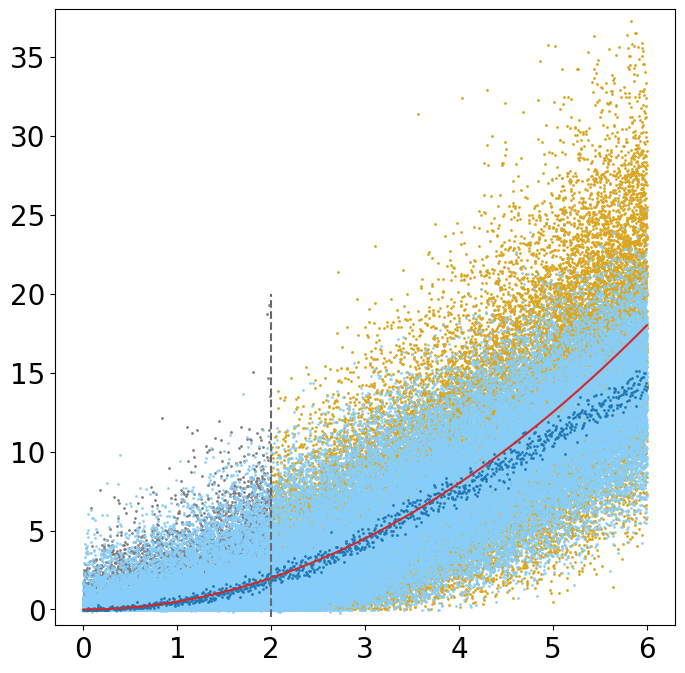} &
	\includegraphics[align=c,width=0.24\textwidth]{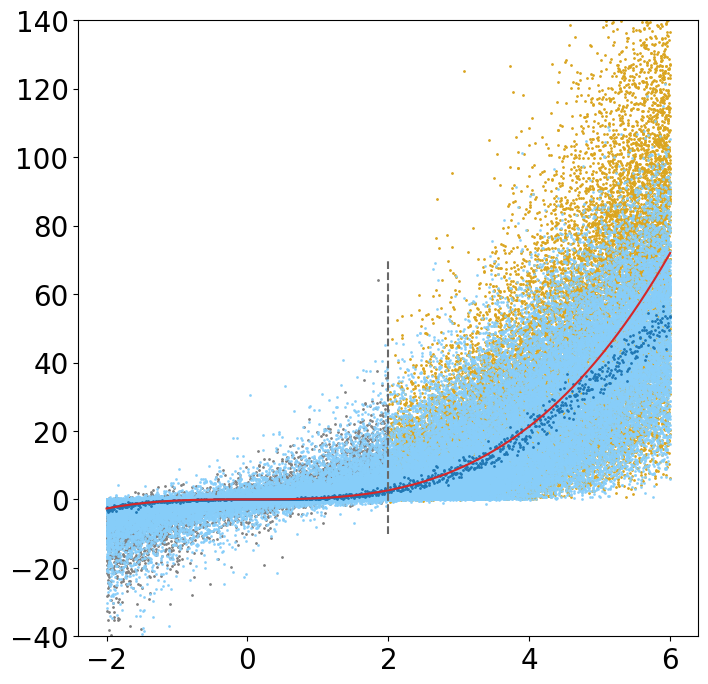} &
	\includegraphics[align=c,width=0.24\textwidth]{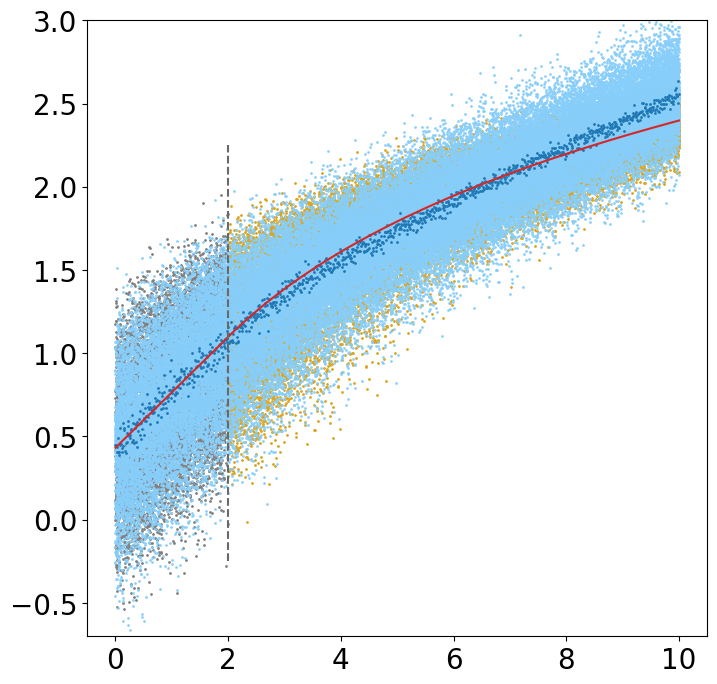} \\
	\texttt{softplus} & \texttt{square} & \texttt{cubic} & \texttt{log} 
\end{tabular}
\caption{Data generated from the true model and engression model within and outside the training support. The dashed vertical lines are the boundaries of training and test data. For each $x$, $\hat{g}(x)$ is obtained by the empirical median of 64 samples from the engression model.}
\label{fig:simu_gen}
\end{figure}

In addition to the visual illustrations, we quantitatively evaluate the conditional median estimation using the $L_1$ loss and the conditional mean estimation using the $L_2$ loss. Both metrics are computed according to $X$ following a uniform distribution on two (sub)sets outside the training support: a subset $[\xm,\xm+\etam]$ on which engression consistently recover the true (median) function by Theorem~\ref{thm:func_recover}, and the `full' test set which goes further beyond the local extrapolable boundary $\xm+\etam$. The metrics for the 20 repetitions are reported via boxplots in Figure~\ref{fig:simu_boxplot}. These numerical results further demonstrate the substantial advantages of engression in extrapolation outside the training support for the conditional mean or median function. 

Additionally, we illustrate how engression can take advantage of higher noise levels. In the simulation settings for \texttt{square} and \texttt{log}, we experiment with various variance levels for the noise $\eta$. As depicted in Figure~\ref{fig:simu_noise_level}, with an increase in the noise level, engression tends to extrapolate over a broader range until the prediction curve begins to diverge from the true function. This observation supports Theorem~\ref{thm:func_recover}, which asserts that engression can consistently recover the true function up to a boundary that becomes less restrictive with an increasing noise level.

Finally, we investigate how well engression fits the entire conditional distribution of $Y|X$ via sampling. Figure~\ref{fig:simu_gen} presents the scatterplots of the true data and data generated from the engression model. Inside the training support, engression captures exactly the true conditional distribution, generating data identically distributed as the true data. Outside the support, the estimated distribution still exhibits a large overlap with the true distribution at the beginning while eventually starts to deviate from the truth, except the \texttt{softplus} case where engression appears to produce linear extrapolation that matches the true one.

\subsection{Real data experiments}\label{sec:realdata}
We apply engression to an extensive range of real data sets spanning various domains such as climate science, biology, ecology, public health, and more. We compare engression with traditional prediction methods, including $L_1$ regression and $L_2$ regression with neural networks, linear quantile regression (LQR)~\citep{koenker1978regression,koenker2005quantile}, linear $L_2$ regression (LR), and quantile regression forest (QRF)~\citep{meinshausen2006quantile}. Our investigation begins with large-scale experiments on univariate prediction, where the primary goal is to assess the wide applicability and validity of engression across different domains. Furthermore, we extend our evaluation to multivariate prediction and prediction intervals, showcasing the versatility of engression in various prediction tasks. We provide descriptions of the data sets and all details about hyperparameter and experimental settings in Appendix~\ref{app:exp_detail}.

\subsubsection{Univariate prediction}\label{sec:exp_uni}

We conduct pairwise predictions on data sets from a variety of fields, which results in~30 pairs of variables and~59 univariate prediction tasks, accomplished by switching the roles of the predictor and response variables (with one discrete variable being exclusively considered as the response). 
For each prediction task, we partition the training and test data based on the quantiles of the predictor, specifically at the 0.3, 0.4, 0.5, 0.6, 0.7 quantiles, and designate the smaller or larger portion as the training sample. This way, all the test data lie outside the training support and we encounter diverse patterns of training and test data. Hence, we have a total of 590 unique data configurations, each representing a specific combination of the data set, predictor/response variables, and data partitioning.
 
Additionally, we utilise the same range of hyperparameter settings for NN-based methods, including engression and $L_1$ or $L_2$ regression. This range includes various combinations of learning rates, numbers of training iterations, layers, and neurons per layer. In each data configuration, we experiment with 18 unique hyperparameter settings, resulting in a total of 31,860 models generated by the three NN-based methods. Each model represents an aggregate of 10 models acquired through 10-fold cross-validation (CV). Regarding evaluation, we report both the average performance across all hyperparameter settings and the performance for the hyperparameter setting that yields the best CV metric. Specifically, for engression, $L_1$ regression, and $L_2$ regression, we employ the energy loss, $L_1$ loss, and $L_2$ loss as the validation metrics, respectively.
This strategy ensures a fair and consistent foundation for comparison and mitigates potential bias originating from hyperparameter tuning.

Our main evaluation metric is the $L_2$ loss on the out-of-support test data which reflects the extrapolation performance in such practical scenarios. Table~\ref{tab:uni_pred} demonstrates the comparative performance of each method across various data configurations. Engression consistently stands out as one of the top methods for the most of the data settings. In contrast, other nonlinear approaches, including $L_1$ and $L_2$ NN regression and quantile regression forests, significantly underperform compared to engression, which aligns with our theoretical and simulation findings presented earlier. Linear models, while inferior to engression, often rank second, possibly because the out-of-support patterns on real-world data are relatively close to linearity. In addition, it is worth noting that the average performance of engression across all hyperparameter settings is already superior compared to other methods, although model selection via cross validation can further magnify the advantage.

\begin{table}
\centering
\caption{Proportion (in percentage) of the data configurations for each method with the out-of-support $L_2$ loss exceeding the best loss by a certain percentage given in the first column. For example, the first number in the table, 79, means that on 79\% of all data configurations, the test $L_2$ loss of engression (averaging across all hyperparameter settings) exceeds the minimum test $L_2$ loss among all approaches for the same data configuration by 1\%. The smaller the numbers, the better the method is. For NN-based methods, including engression and $L_1$ and $L_2$ regression, we report both the average performance across all hyperparameter settings or the best performance according to the CV metric.}\vskip 0.1in
\label{tab:uni_pred}
\begin{tabular}{rrrrrrrrrr}
\toprule
\multicolumn{1}{c}{exceeding} & \multicolumn{2}{c}{engression} & \multicolumn{2}{c}{$L_1$ regression} & \multicolumn{2}{c}{$L_2$ regression} & \multirow{2}{*}{LQR}  &  \multirow{2}{*}{LR}  & \multirow{2}{*}{QRF}  \\ 
\multicolumn{1}{c}{percentage} & average & CV & average & CV & average & CV &    &    &    \\ 
\midrule
1\% & 79 & 55 & 98 & 72 & 97 & 82 & 85 & 81 & 93 \\
3\% & 65 & 44 & 95 & 63 & 93 & 72 & 78 & 71 & 91 \\
10\% & 35 & 28 & 86 & 44 & 83 & 56 & 61 & 54 & 83 \\
30\% & 14 & 12 & 64 & 22 & 60 & 41 & 38 & 34 & 62 \\
100\% & 2 & 3 & 26 & 9 & 21 & 25 & 13 & 11 & 28 \\
\bottomrule
\end{tabular}
\end{table}

\begin{figure}
\centering
\begin{tabular}{@{}cc@{}}
	\includegraphics[width=0.49\textwidth]{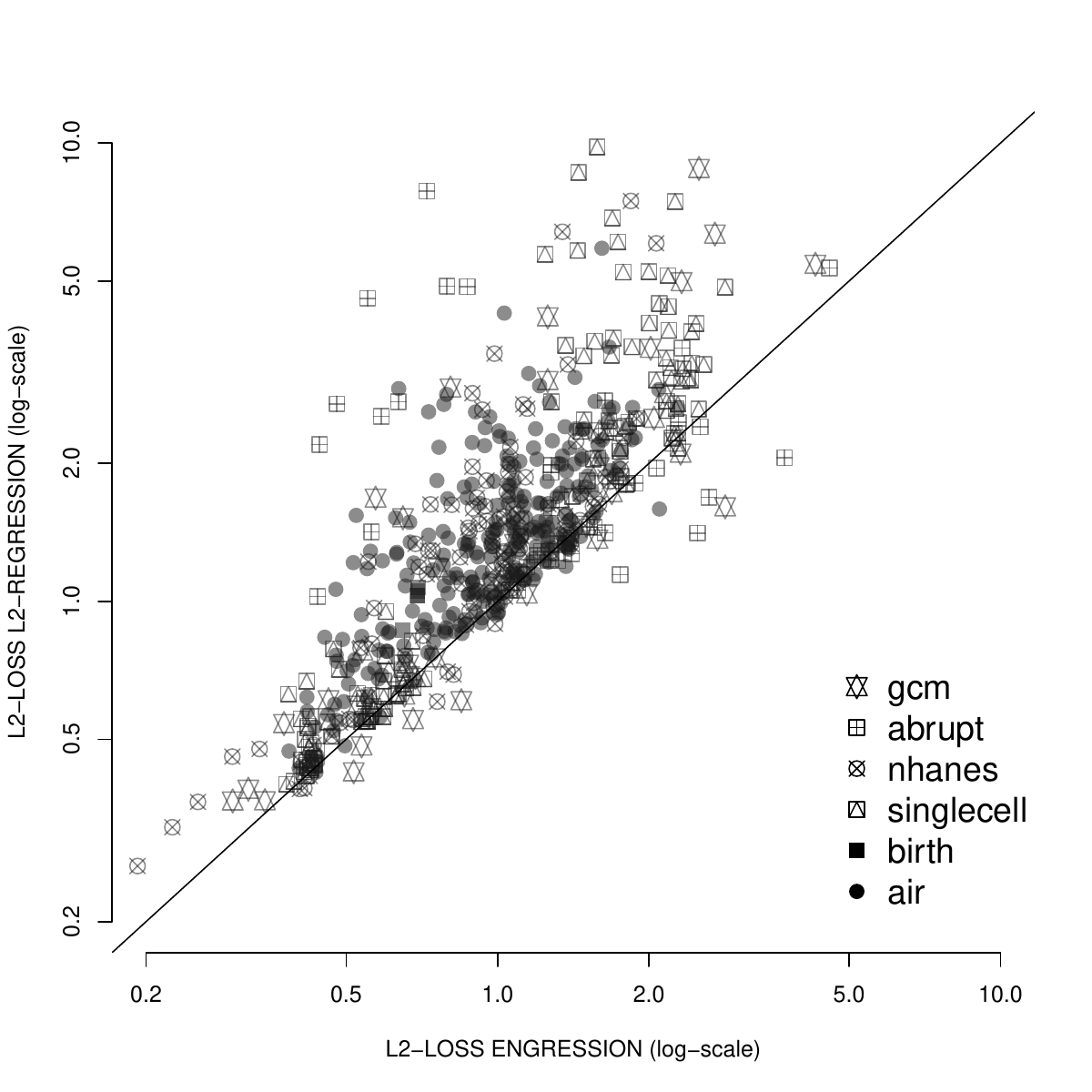} & 
	\includegraphics[width=0.49\textwidth]{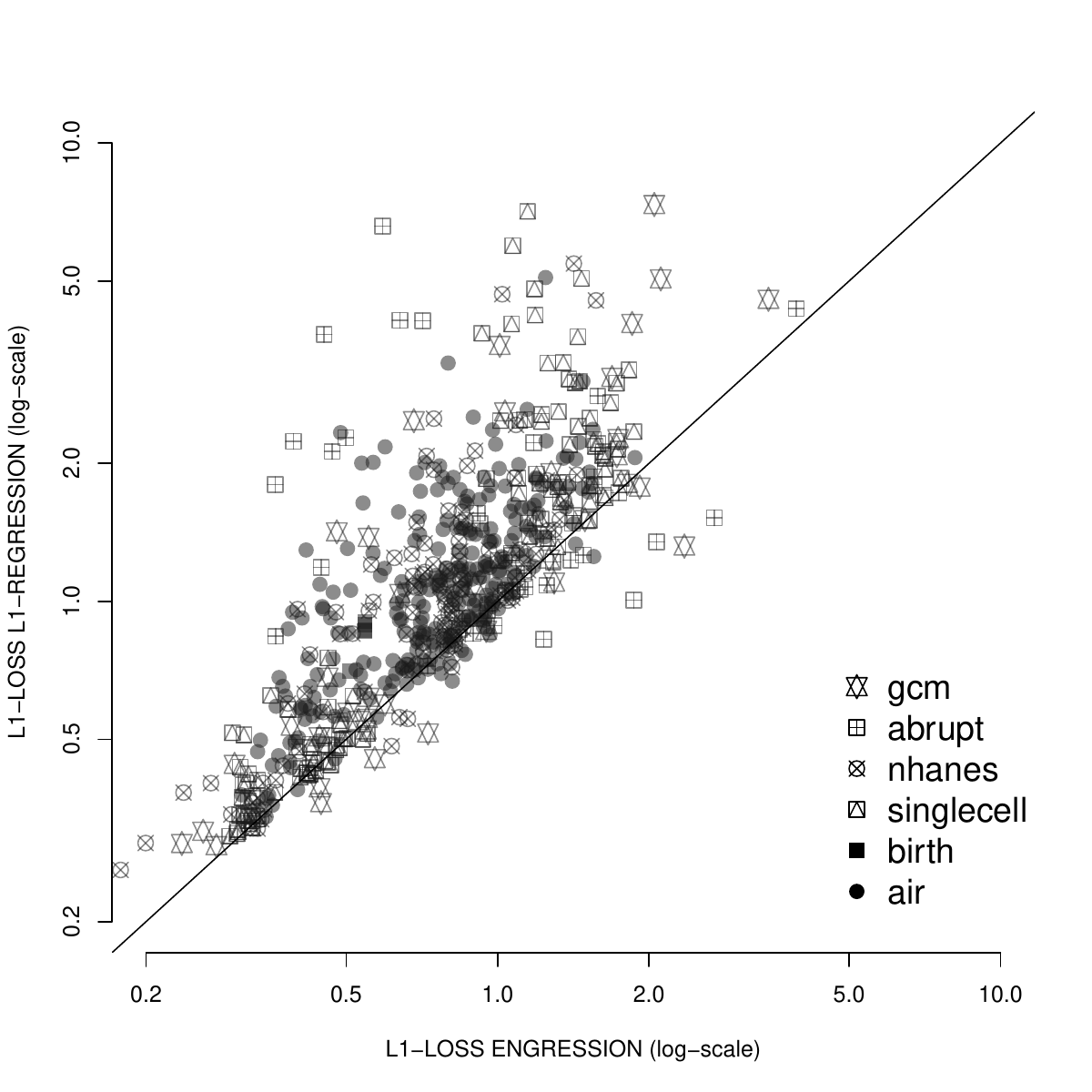}\\
	\small{(a) $L_2$ loss} & \small{(b) $L_1$ loss}
\end{tabular}
\caption{Out-of-support losses (in log-scale) of engression and regression for various data configurations, averaging over all hyperparameter settings.}
\label{fig:uni_eng_reg}
\end{figure}

\begin{figure}
\centering
	\includegraphics[width=0.6\textwidth]{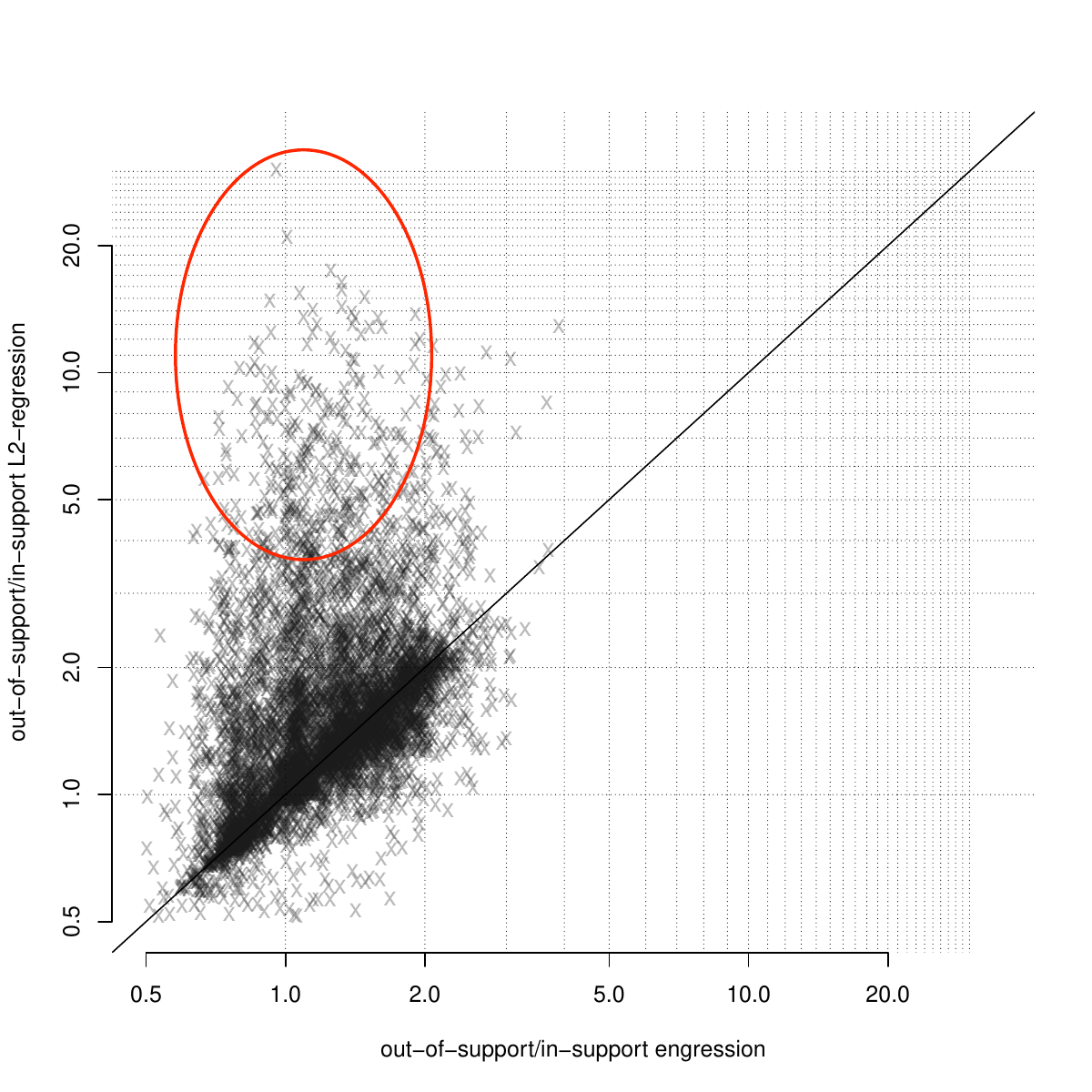}  
\caption{The ratio (in log-scale) between out-of-support and in-support $L_2$ losses of engression and regression for all hyperparameter settings on the air pollution data.}
\label{fig:uni_outsample}
\end{figure}

Moreover, we show a more detailed comparison between engression and regression, both always using the same architecture of the neural networks. As demonstrated in Figure~\ref{fig:uni_eng_reg}, engression consistently surpasses $L_1$ and $L_2$ regression in terms of the $L_1$ and $L_2$ loss on test data that fall outside the support, respectively. This is true even though these losses are specifically optimised by the regression methods on the training data.

Beyond the absolute evaluation of performance on data outside the support, we also scrutinise the relative performance between predictions made within and beyond the support. We utilise the $L_2$ loss on the omitted folds in cross-validation to gauge the performance within the support, which aids in sidestepping potential overfitting effects. Figure~\ref{fig:uni_outsample} shows the ratio between the $L_2$ losses for predictions made outside the support versus those made within it for both methods. We arrive at two key observations. Firstly, the ratios of engression are typically around 1, suggesting negligible performance degradation when transitioning from within-support to outside-support prediction. In contrast, regression tends to encounter more pronounced difficulties when extending beyond the training support, as illustrated by the generally larger ratios (note the log-scaling on the plot). In particular, there exist several instances where engression maintains comparable performance levels both within and outside the support, whereas regression exhibits a substantial increase in loss for predictions outside the support compared to those within it, as shown in the red circle. This further highlights the superior extrapolation performance of engression. Additionally, it should be noted that we plot the performance for all hyperparameter settings, and engression demonstrates less variance in performance compared to regression. This implies that engression displays greater robustness across a range of hyperparameter configurations.

The large-scale experimental study comprehensively demonstrates the remarkable performance of engression in out-of-support prediction. 
The empirical results suggest that engression is suitable for a wide range of real data and offers practitioners an alternative modelling choice for practical data analysis. Furthermore, it is worth noting that for each pair of variables across all data sets, we include the prediction tasks in both directions, taking either of them as the response and the other one as the predictor and the other way round. This ensures that at least in one of these two directions, the assumption on the pre-ANM is not fully fulfilled, unless the relationship is linear. 
Considering this, the observation that engression works well across all these data set configurations indicates the satisfying robustness of engression under model misspecification, which makes it even more widely applicable in real scenarios.

\subsubsection{Multivariate prediction}\label{sec:exp_multi}
Next, we demonstrate the performance of engression in multivariate prediction through two data sets. (While we focus here on the setting of multivariate predictors and univariate outcome, the method is also directly applicable to multivariate outcomes, which is implemented in our software.) The air quality data comprises measurements of five pollutants and three meteorological covariates. For the prediction task, we select one pollutant as the response variable and utilise the remaining variables as predictors. The NHANES data set includes the total activity count (TAC), sedentary time (ST), body mass index (BMI), and age, where we consider predicting the TAC from the remaining three variables. The data is split into training and test sets based on the median of one of the predictors. We evaluate the performance by the $L_2$ losses on both training and test data. 

\begin{figure}
\centering
\begin{tabular}{@{}c@{}c@{}c@{}}
	Engression & NN regression & Linear regression\\
	\small{(0.0798, 0.1662)} & \small{(0.0099, 0.7670)} & \small{(0.1639, 0.5022)}\vspace{-0.05in}\\
	\includegraphics[width=0.33\textwidth]{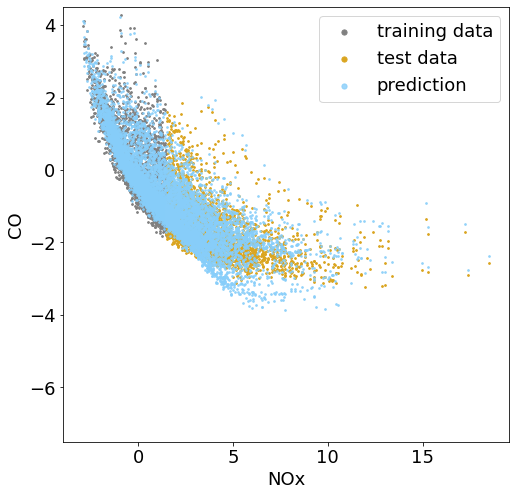} &
	\includegraphics[width=0.33\textwidth]{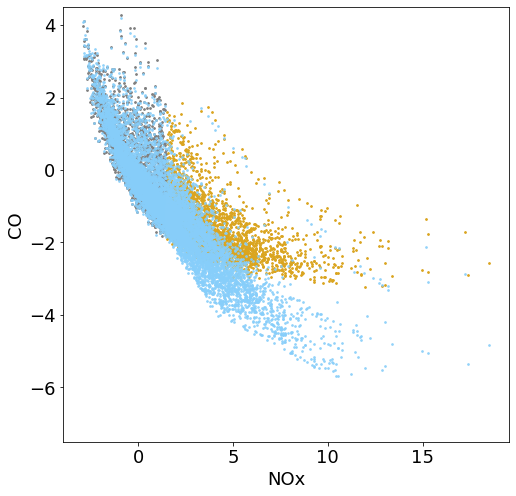} &
	\includegraphics[width=0.33\textwidth]{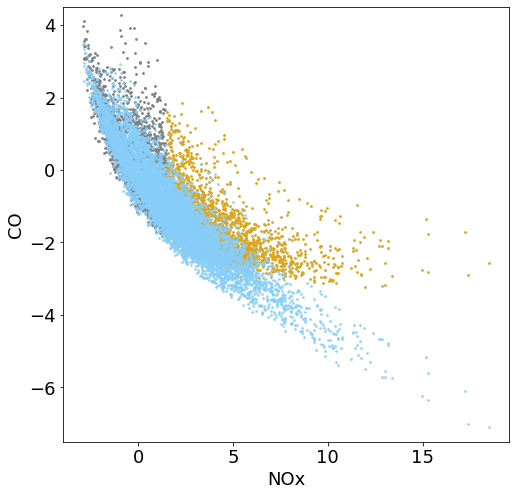} \\
	\small{(0.1687, 0.2096)} & \small{(0.0017, 2.9041)} & \small{(0.3041, 1.4730)}\vspace{-0.05in}\\
	\includegraphics[width=0.33\textwidth]{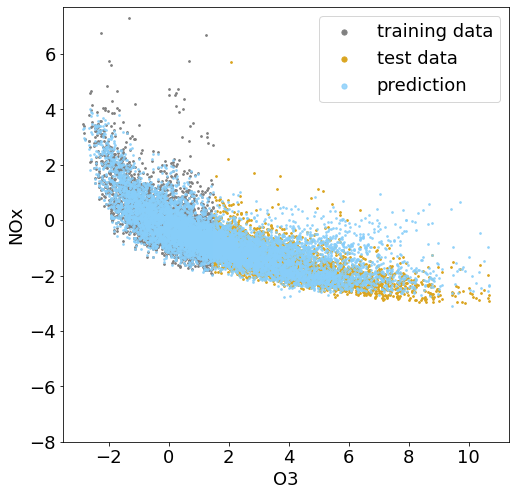} &
	\includegraphics[width=0.33\textwidth]{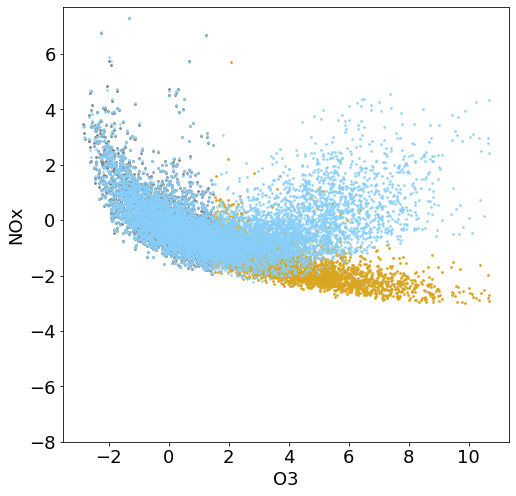} &
	\includegraphics[width=0.33\textwidth]{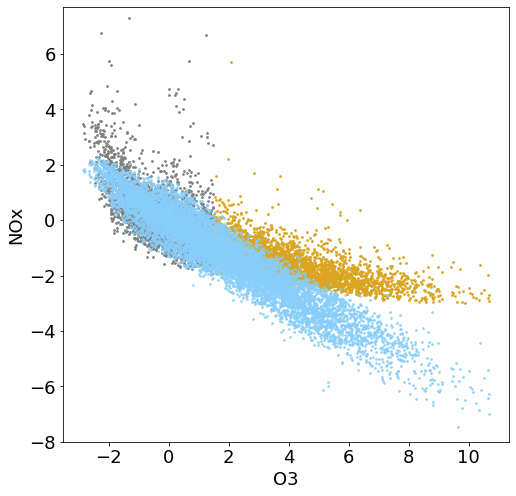} \\
	\small{(0.1267, 0.6392)} & \small{(0.0021, 1.8257)} & \small{(0.3462, 0.9700)}\vspace{-0.05in}\\
	\includegraphics[width=0.33\textwidth]{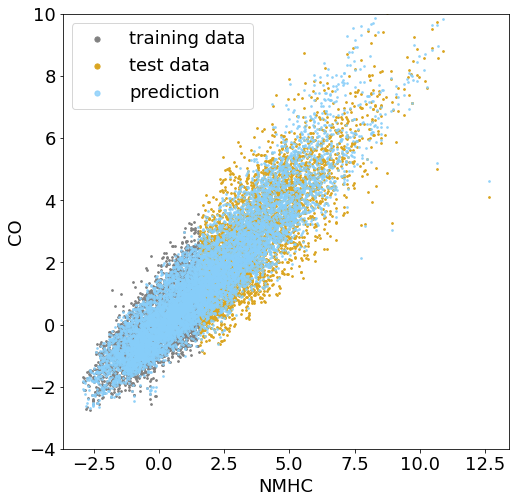} &
	\includegraphics[width=0.33\textwidth]{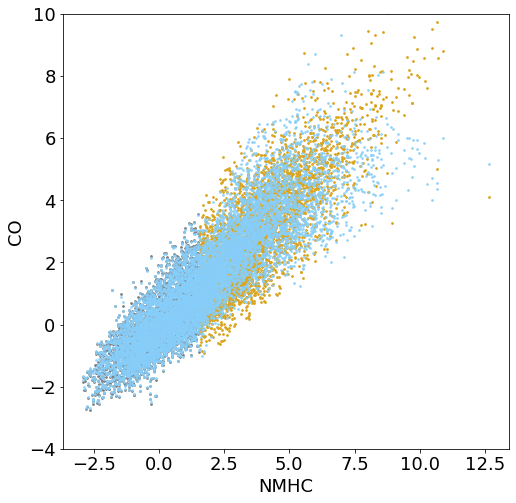} &
	\includegraphics[width=0.33\textwidth]{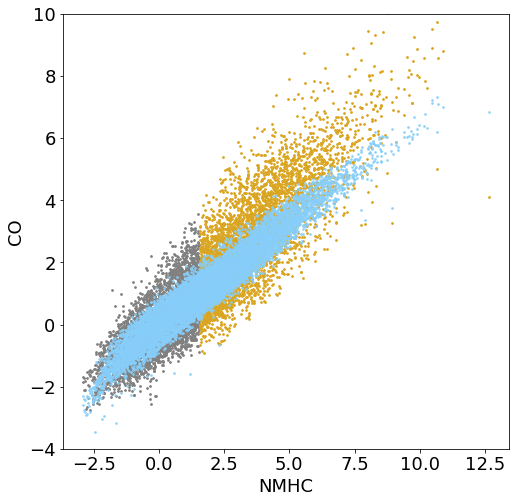} 
\end{tabular}
\caption{Results for multivariate prediction on the air quality data. Each row shows a predictor-response pair; each column is for a method. In each plot, the variable on the x-axis is used to split the training and test data at its median. The numbers in the parenthesis above each plot are the training and test $L_2$ losses, respectively. The spread in the predictions arises due to the presence of other predictor variables in this multivariate prediction.}\label{fig:multivar_air}
\end{figure}

\begin{figure}
\centering
\begin{tabular}{@{}c@{}c@{}c@{}}
	Engression & NN regression & Linear regression\\
	\small{(0.3632, 0.1735)} & \small{(0.2087, 3.8117)} & \small{(0.3741, 0.8245)}\vspace{-0.05in}\\
	\includegraphics[width=0.33\textwidth]{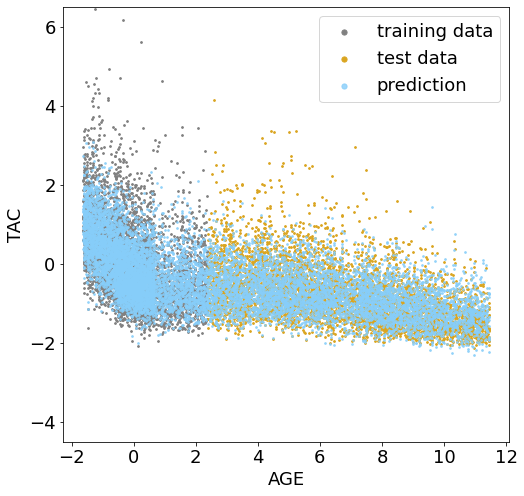} &
	\includegraphics[width=0.33\textwidth]{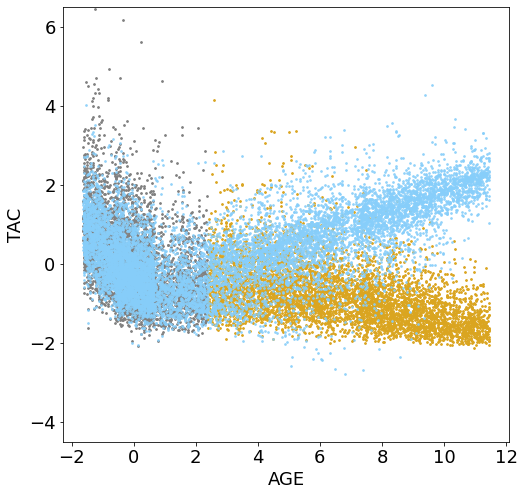} &
	\includegraphics[width=0.33\textwidth]{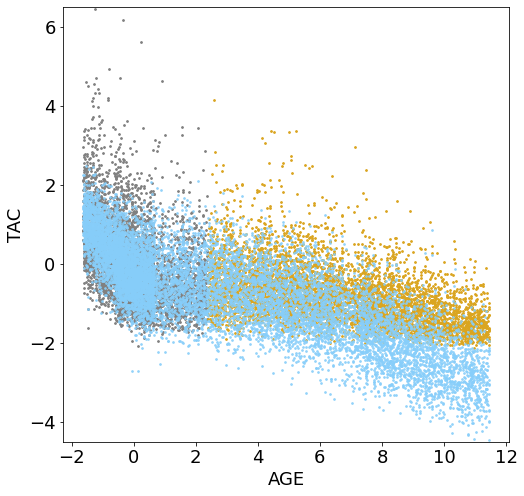} \\
	\small{(0.4500, 0.1838)} & \small{(0.3393, 0.6753)} & \small{(0.4646, 0.4223)}\vspace{-0.05in}\\
	\includegraphics[width=0.33\textwidth]{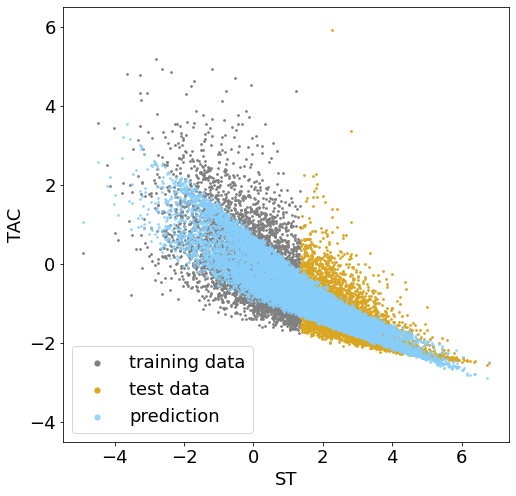} &
	\includegraphics[width=0.33\textwidth]{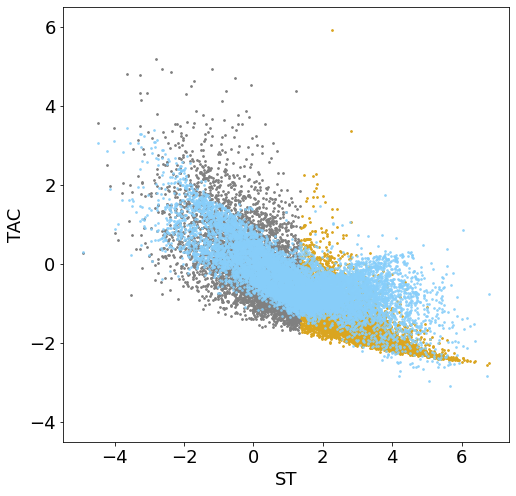} &
	\includegraphics[width=0.33\textwidth]{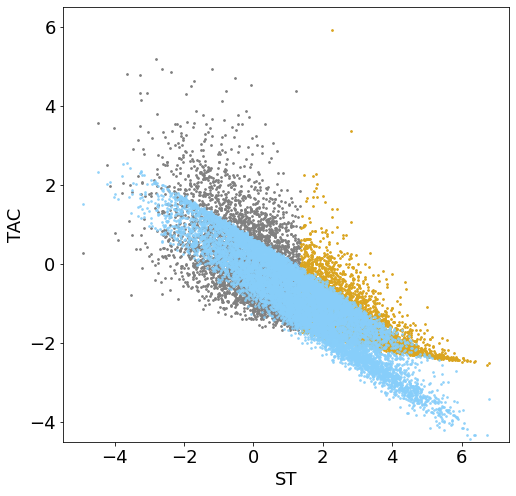} \\
	\small{(0.3156, 0.1743)} & \small{(0.2095, 1.0541)} & \small{(0.3375, 0.2532)}\vspace{-0.05in}\\
	\includegraphics[width=0.33\textwidth]{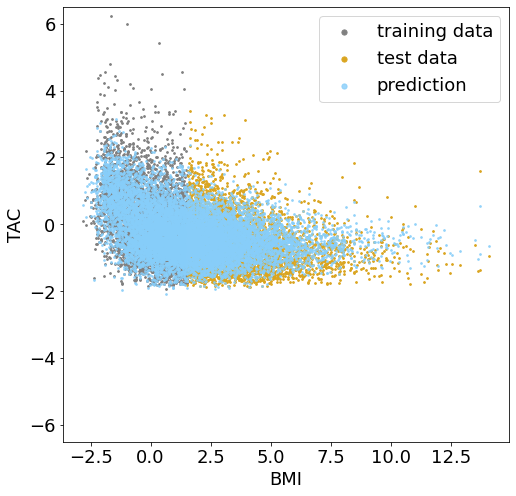} &
	\includegraphics[width=0.33\textwidth]{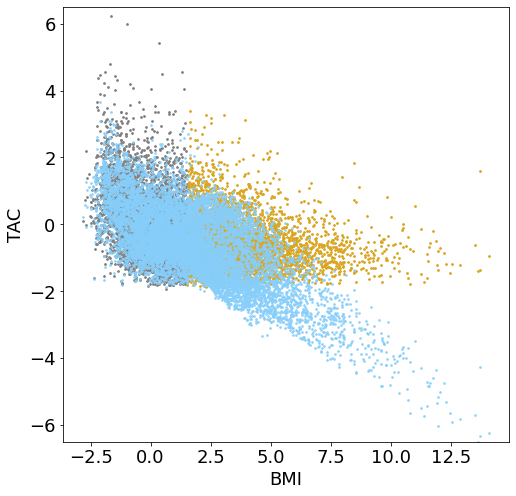} &
	\includegraphics[width=0.33\textwidth]{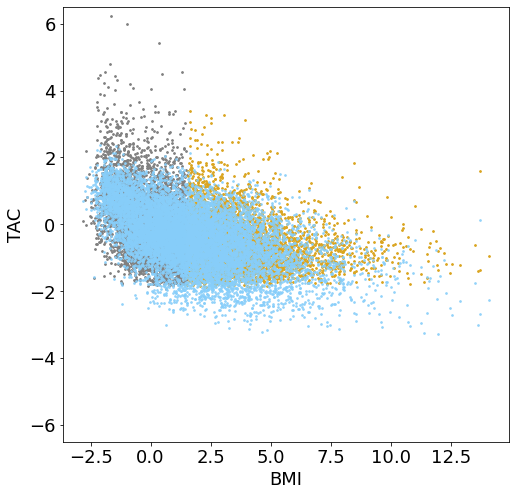} 
\end{tabular}
\caption{Same plot for the NHANES data. }\label{fig:multivar_health}
\end{figure}

In Figures~\ref{fig:multivar_air}-\ref{fig:multivar_health}, we visualise the results by plotting the response variable against the predictor $i$ that is used for splitting training and test data. We include both the true data points and the prediction of $\bbE[Y|X_i=x_i,X_{-i}]$, where $X_{-i}$ denotes the set of remaining predictors excluding $X_i$. In most scenarios, the marginal relations between the response and the predictor are nonlinear. We observe that linear regression tends to show predominantly linear extrapolations when visualised univariately. However,  nonlinear effects can  surface in the results, originating from  nonlinear dependencies between the depicted predictor variables and the other variables present in the data set.
 On the other hand, NN $L_2$ regression, while consistently having the smallest training loss, fails to extrapolate reliably, resulting in a significantly larger test loss. In comparison, engression outperforms both linear regression and NN regression. It achieves a lower training loss than linear regression, owing to the expressiveness of the NN class. More importantly, engression often maintains a strong performance on out-of-support test data, significantly surpassing the other two approaches. 

In the bottom row of Figure~\ref{fig:multivar_air}, the marginal relationship appears to be almost linear, with a slight presence of heteroscedasticity. In this case, all methods appear to provide reasonable predictions visually. However, when evaluated quantitatively, engression exhibits a much lower test loss. In fact, although the marginal relationship between CO and NMHC is approximately linear, the relationships between CO and other predictors, like NOx (as shown in the top row) are nonlinear. Thus, predicting CO based on multiple predictors still requires nonlinear extrapolation in which engression excels, although the advantage lies in the other predictors and is not explicitly reflected in the marginal relationship between CO and NMHC as shown in the plots.

In summary, these empirical findings suggest the potential of engression in multivariate prediction. Thus, it is worthwhile to delve deeper into the theoretical underpinnings of engression in multivariate scenarios in future research.

\subsubsection{Prediction intervals}\label{sec:exp_pi}

Engression can also be applied to construct prediction intervals based on the conditional quantile estimations. Specifically, we compute the 2.5\% and 97.5\% quantiles of the estimated conditional distribution from engression and obtain the 95\% prediction intervals based on them. Note that we do not account for the estimation uncertainty arising due to a finite sample size. We use linear quantile regression and quantile regression forests as baseline methods, where the intervals are derived based on the estimated quantiles. We also consider standard prediction intervals with linear $L_2$ regression for Gaussian errors as another baseline. 

In Figure~\ref{fig:pi}, we visualise the 95\% prediction intervals on three data sets, along with the average coverage probabilities. We observe that for engression and most baseline methods, about 95\% of all training data points fall within the 95\% prediction intervals. However, when it comes to out-of-support test data, only the prediction intervals obtained through engression attain good coverage, while the other three methods exhibit significantly lower coverage rates. For instance, for the GCM data where we predict the global mean temperature from the radiation, their relationship is nonlinear during training while extrapolating linearly during test. Linear methods, though obtaining the best linear fits for the training distribution, provide misleading prediction intervals that cover the test data poorly. In comparison, engression successfully captures the pattern outside the support, leading to prediction intervals with decent coverage. 
These results align with the extrapolability guarantee of engression for quantile estimation in Corollary~\ref{coro:local_quantile}, whereas the failure of other methods is mainly due to unreliable estimations of the quantiles (for QR and QRF) or the mean (for LR) outside the support.

It is notable that conformal prediction~\citep{shafer2008tutorial,lei2018distribution} can also be used to construct prediction intervals. Nonetheless, when it comes to extrapolation, the residuals of a regression model at data points that fall outside the support would generally not fulfill the exchangeability assumption necessary for conformal prediction as we are operating outside the support of the training distribution. Consequently, adapting conformal prediction to ensure coverage guarantees for out-of-support predictions is not a straightforward task.

\begin{figure}
\centering
\begin{tabular}{@{}c@{}c@{}c@{}c@{}c@{}}
    &{\small{Engression}} & {\small{Linear regression}} & {\small{Linear quantile regression}} & {\small{quantile regression forest}} \\
    &\small{(0.964, 0.917)} & \small{(0.937, 0.366)} & \small{(0.950, 0.333)} & \small{(0.938, 0.152)}\vspace{-0.05in}\\
    \rotatebox[origin=c]{90}{\small{GCM}} &
    \includegraphics[align=c,width=0.24\textwidth]{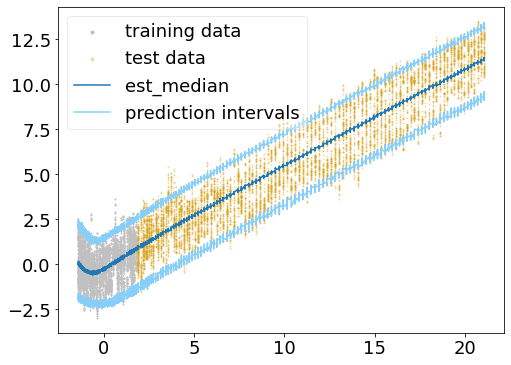} &
    \includegraphics[align=c,width=0.24\textwidth]{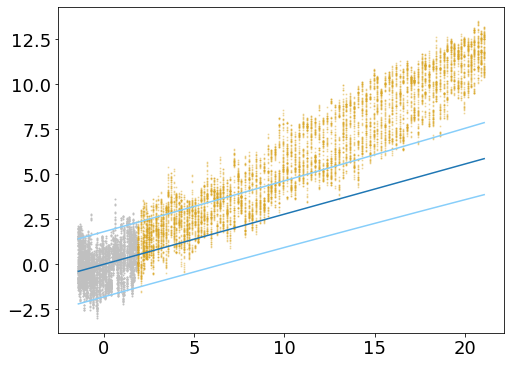} &
    \includegraphics[align=c,width=0.24\textwidth]{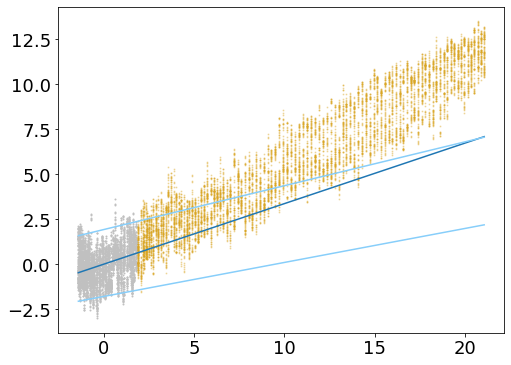} &
    \includegraphics[align=c,width=0.24\textwidth]{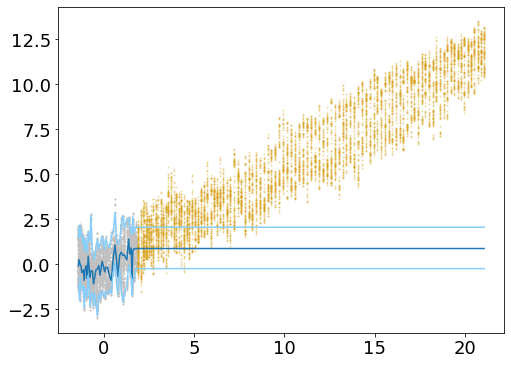}\vspace{0.05in} \\
    &\small{(0.946, 0.916)} & \small{(0.900, 0.363)} & \small{(0.951, 0.313)} & \small{(0.995, 0.559)}\vspace{-0.05in}\\
    \rotatebox[origin=c]{90}{\small{Air}} &
    \includegraphics[align=c,width=0.24\textwidth]{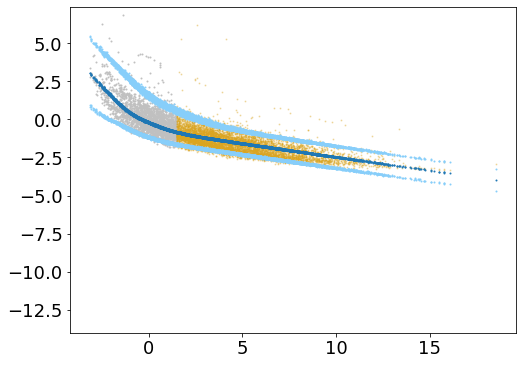} &
    \includegraphics[align=c,width=0.24\textwidth]{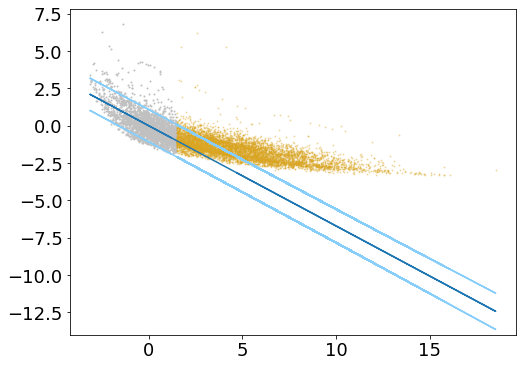} &
    \includegraphics[align=c,width=0.24\textwidth]{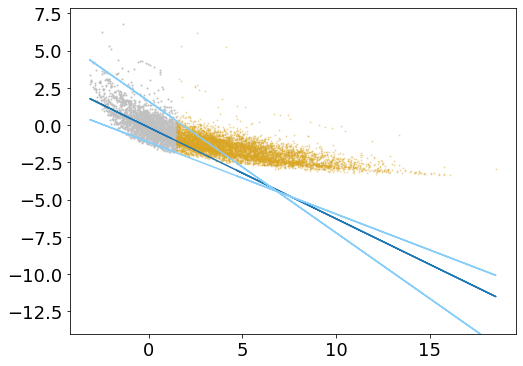} &
    \includegraphics[align=c,width=0.24\textwidth]{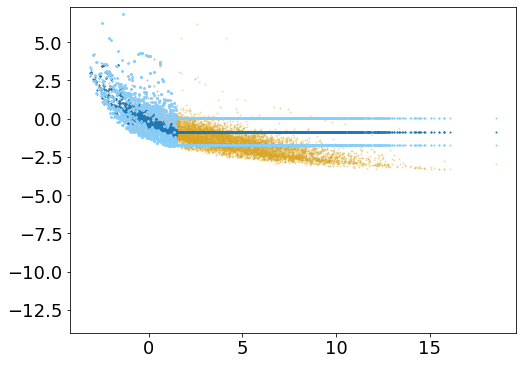} \vspace{0.05in}\\
    &\small{(0.951, 0.885)} & \small{(0.939, 0.264)} & \small{(0.951, 0.272)} & \small{(0.962, 0.690)}\vspace{-0.05in}\\
    \rotatebox[origin=c]{90}{\small{NHANES}} &
    \includegraphics[align=c,width=0.24\textwidth]{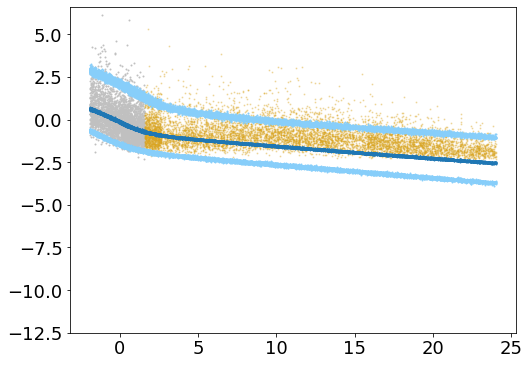} &
    \includegraphics[align=c,width=0.24\textwidth]{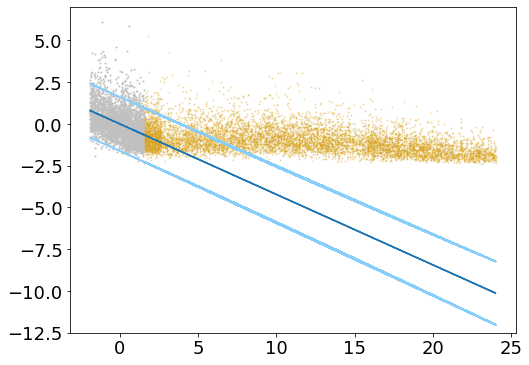} &
    \includegraphics[align=c,width=0.24\textwidth]{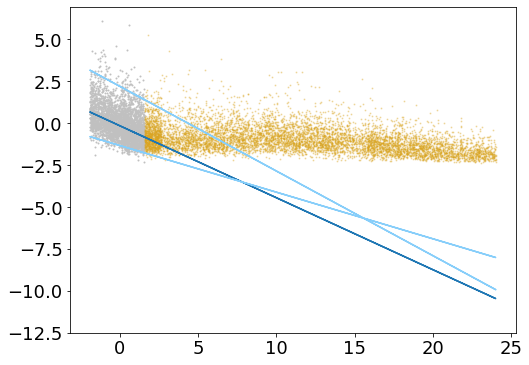} &
    \includegraphics[align=c,width=0.24\textwidth]{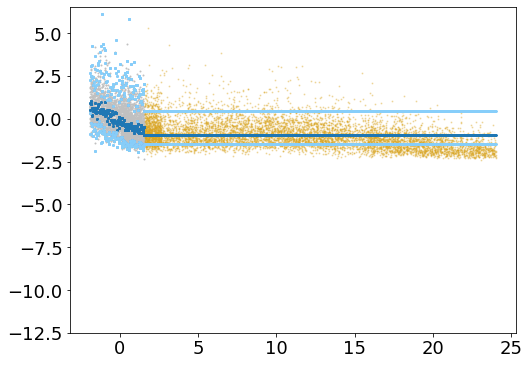} 
\end{tabular}
\caption{Prediction intervals. Each row shows the results on a data set; each column is for a method. The numbers in the parenthesis above each plot are the average coverage probabilities of the prediction intervals on training and test data, respectively.}
\label{fig:pi}
\end{figure}

\section{Discussion}
\revise{
Our main contribution is twofold:
\begin{enumerate}
\setlength{\parskip}{2pt}
	\item[(1)] We propose a new method `engression' for estimating the conditional distribution of a response variable given covariates.
	\item[(2)] Engression in conjunction with pre-additive noise models provides a new perspective for the extrapolation problem in nonlinear regression.
\end{enumerate}

As a distributional regression method, engression is simple yet generically applicable to various regression tasks. Our preliminary theoretical and empirical studies have shown some promise. From a practical perspective, compared to traditional statistical methodologies, engression benefits from recent advancements in machine learning. The expressive capacity of neural networks alleviates the limitations of parametric model specifications. Furthermore, modern optimisation tools and computational hardware have made the associated computations tractable and scalable, even when dealing with vast sample sizes and high-dimensional data. Compared to modern generative models for image and text data, such as diffusion models and generative adversarial networks, the simplicity of engression makes it much more user-friendly and computationally lighter.

We therefore envisage engression as a useful tool for both statisticians and applied researchers. On one hand, engression can inspire the development of new methodologies for many statistical inference problems that involve (conditional) distribution estimation such as conditional independence test \citep{shi2021double}, missing data imputation \citep{naf2024good}, and dimension reduction \citep{shen2024distributional}. On the other hand, engression can be an interesting addition to the current data analysis toolkit, as it possesses distinct characteristics than existing nonlinear regression techniques, such as more comprehensive quantification of the full distribution and different behaviour when it comes to data outside the training support, which could be more appealing to many practitioners. 

In addition, we study one important aspect of engression: the ability to extrapolate beyond the support of the training data. 
}
This need arises from the limitations of existing regression models: linear models, while capable of extrapolation, are a rather restrictive class, offering limited flexibility. In contrast, nonlinear regression models, such as tree ensembles and neural networks, suffer from their own deficiencies when it comes to extrapolation. \revise{Tree ensembles provide a constant prediction when extending beyond of the support, which may not be accurate in many contexts. Neural network predictions, on the contrary, tend to become uncontrollable when used for extrapolation.} 

\revise{
We highlight that engression exhibits a distinct extrapolation behaviour than the existing models and thus offers a new perspective and modelling tool. Our recipe for extrapolation} 
incorporates two key components: a pre-additive noise model and a distributional fit. 
The advantage of a
pre-ANM class is, in simple terms, that it reveals the 
nonlinearities outside of the training support. 
This benefit can only  be captured, however, if we fit the full conditional distribution. Importantly, we have shown that both ingredients (i.e.\ pre-ANM and distributional fitting) are necessary for extrapolation and that neither adopting a post-ANM nor fitting merely the mean or quantiles would attain the same merits.

The usefulness and robustness of our proposed engression methodology are demonstrated through theoretical proofs and empirical evidence. Our theoretical results show a significant improvement \revise{of engression over traditional regression approaches} in extrapolation. Interestingly, the range to which we can extrapolate benefits from a high noise level. In addition, we have studied finite-sample bounds for some specific settings, which shows the possibility to extrapolate beyond the support of the finite training data if the pre-additive noise model is correct. Even when the pre-additive noise assumption is incorrect and the true model proves to be a post-additive model with a homoscedastic noise, the model produced by engression defaults to a linear extrapolation, thereby  ensuring a useful baseline performance.
Moreover, our empirical results show that  pre-additive noise models are a reasonable model for many data sets, validate our theoretical findings, and show the robustness and accuracy of engression in various regression tasks.

\newpage
\bibliography{ref.bib}

\begin{thebibliography}{}

\bibitem[Abe et~al., 2022]{abe2022deep}
Abe, T., Buchanan, E.~K., Pleiss, G., Zemel, R., and Cunningham, J.~P. (2022).
\newblock Deep ensembles work, but are they necessary?
\newblock {\em Advances in Neural Information Processing Systems},
  35:33646--33660.

\bibitem[Arjovsky et~al., 2019]{arjovsky2019invariant}
Arjovsky, M., Bottou, L., Gulrajani, I., and Lopez-Paz, D. (2019).
\newblock Invariant risk minimization.
\newblock {\em arXiv preprint arXiv:1907.02893}.

\bibitem[Barber et~al., 2023]{barber2023conformal}
Barber, R.~F., Candes, E.~J., Ramdas, A., and Tibshirani, R.~J. (2023).
\newblock Conformal prediction beyond exchangeability.
\newblock {\em The Annals of Statistics}, 51(2):816--845.

\bibitem[Baringhaus and Franz, 2004]{BARINGHAUS2004190}
Baringhaus, L. and Franz, C. (2004).
\newblock On a new multivariate two-sample test.
\newblock {\em Journal of Multivariate Analysis}, 88(1):190--206.

\bibitem[Bellemare et~al., 2017]{bellemare2017cramer}
Bellemare, M.~G., Danihelka, I., Dabney, W., Mohamed, S., Lakshminarayanan, B.,
  Hoyer, S., and Munos, R. (2017).
\newblock The cramer distance as a solution to biased wasserstein gradients.
\newblock {\em arXiv preprint arXiv:1705.10743}.

\bibitem[Ben-David et~al., 2006]{ben2006analysis}
Ben-David, S., Blitzer, J., Crammer, K., and Pereira, F. (2006).
\newblock Analysis of representations for domain adaptation.
\newblock In {\em Advances in Neural Information Processing Systems},
  volume~19.

\bibitem[Breiman, 2001]{breiman2001random}
Breiman, L. (2001).
\newblock Random forests.
\newblock {\em Machine learning}, 45:5--32.

\bibitem[Brillinger, 2012]{brillinger2012generalized}
Brillinger, D.~R. (2012).
\newblock A generalized linear model with “gaussian” regressor variables.
\newblock {\em Selected Works of David Brillinger}, pages 589--606.

\bibitem[Brown et~al., 2020]{brown2020language}
Brown, T., Mann, B., Ryder, N., Subbiah, M., Kaplan, J.~D., Dhariwal, P.,
  Neelakantan, A., Shyam, P., Sastry, G., Askell, A., et~al. (2020).
\newblock Language models are few-shot learners.
\newblock {\em Advances in neural information processing systems},
  33:1877--1901.

\bibitem[B{\"u}hlmann and Hothorn, 2007]{buhlmann2007boosting}
B{\"u}hlmann, P. and Hothorn, T. (2007).
\newblock {Boosting Algorithms: Regularization, Prediction and Model Fitting}.
\newblock {\em Statistical Science}, 22(4):477 -- 505.

\bibitem[B{\"u}hlmann and Yu, 2003]{buhlmann2003boosting}
B{\"u}hlmann, P. and Yu, B. (2003).
\newblock Boosting with the l2 loss: regression and classification.
\newblock {\em Journal of the American Statistical Association},
  98(462):324--339.

\bibitem[Carlier et~al., 2017]{carlier2017vector}
Carlier, G., Chernozhukov, V., and Galichon, A. (2017).
\newblock Vector quantile regression beyond the specified case.
\newblock {\em Journal of Multivariate Analysis}, 161:96--102.

\bibitem[Cauchy, 1847]{cauchy1847methode}
Cauchy, A. (1847).
\newblock M{\'e}thode g{\'e}n{\'e}rale pour la r{\'e}solution des systemes
  d’{\'e}quations simultan{\'e}es.
\newblock {\em Comp. Rend. Sci. Paris}, 25(1847):536--538.

\bibitem[Chen et~al., 2024]{chen2024generative}
Chen, J., Janke, T., Steinke, F., and Lerch, S. (2024).
\newblock Generative machine learning methods for multivariate ensemble
  postprocessing.
\newblock {\em The Annals of Applied Statistics}, 18(1):159--183.

\bibitem[Chen and B{\"u}hlmann, 2021]{chen2021domain}
Chen, Y. and B{\"u}hlmann, P. (2021).
\newblock Domain adaptation under structural causal models.
\newblock {\em The Journal of Machine Learning Research}, 22(1):11856--11935.

\bibitem[Chernozhukov et~al., 2010]{chernozhukov2010quantile}
Chernozhukov, V., Fern{\'a}ndez-Val, I., and Galichon, A. (2010).
\newblock Quantile and probability curves without crossing.
\newblock {\em Econometrica}, 78(3):1093--1125.

\bibitem[Christiansen et~al., 2021]{christiansen2021causal}
Christiansen, R., Pfister, N., Jakobsen, M.~E., Gnecco, N., and Peters, J.
  (2021).
\newblock A causal framework for distribution generalization.
\newblock {\em IEEE Transactions on Pattern Analysis and Machine Intelligence},
  44(10):6614--6630.

\bibitem[Cohen et~al., 2019]{cohen2019certified}
Cohen, J., Rosenfeld, E., and Kolter, Z. (2019).
\newblock Certified adversarial robustness via randomized smoothing.
\newblock In {\em international conference on machine learning}, pages
  1310--1320. PMLR.

\bibitem[Cui et~al., 2021]{cui2021additive}
Cui, E., Crainiceanu, C.~M., and Leroux, A. (2021).
\newblock Additive functional cox model.
\newblock {\em Journal of Computational and Graphical Statistics},
  30(3):780--793.

\bibitem[Danabasoglu et~al., 2020]{danabasoglu2020community}
Danabasoglu, G., Lamarque, J.-F., Bacmeister, J., Bailey, D., DuVivier, A.,
  Edwards, J., Emmons, L., Fasullo, J., Garcia, R., Gettelman, A., et~al.
  (2020).
\newblock The community earth system model version 2 (cesm2).
\newblock {\em Journal of Advances in Modeling Earth Systems},
  12(2):e2019MS001916.

\bibitem[Dawid, 2007]{dawid2007geometry}
Dawid, A.~P. (2007).
\newblock The geometry of proper scoring rules.
\newblock {\em Annals of the Institute of Statistical Mathematics}, 59:77--93.

\bibitem[Dong and Ma, 2023]{dong2023first}
Dong, K. and Ma, T. (2023).
\newblock First steps toward understanding the extrapolation of nonlinear
  models to unseen domains.
\newblock In {\em The Eleventh International Conference on Learning
  Representations}.

\bibitem[Dunson et~al., 2007]{dunson2007bayesian}
Dunson, D.~B., Pillai, N., and Park, J.-H. (2007).
\newblock Bayesian density regression.
\newblock {\em Journal of the Royal Statistical Society Series B: Statistical
  Methodology}, 69(2):163--183.

\bibitem[Fannjiang et~al., 2022]{fannjiang2022conformal}
Fannjiang, C., Bates, S., Angelopoulos, A.~N., Listgarten, J., and Jordan,
  M.~I. (2022).
\newblock Conformal prediction under feedback covariate shift for biomolecular
  design.
\newblock {\em Proceedings of the National Academy of Sciences},
  119(43):e2204569119.

\bibitem[Foresi and Peracchi, 1995]{foresi1995conditional}
Foresi, S. and Peracchi, F. (1995).
\newblock The conditional distribution of excess returns: An empirical
  analysis.
\newblock {\em Journal of the American Statistical Association},
  90(430):451--466.

\bibitem[Friedman, 2001]{friedman2001greedy}
Friedman, J.~H. (2001).
\newblock Greedy function approximation: a gradient boosting machine.
\newblock {\em Annals of statistics}, pages 1189--1232.

\bibitem[Ganin and Lempitsky, 2015]{pmlr-v37-ganin15}
Ganin, Y. and Lempitsky, V. (2015).
\newblock Unsupervised domain adaptation by backpropagation.
\newblock In Bach, F. and Blei, D., editors, {\em Proceedings of the 32nd
  International Conference on Machine Learning}, volume~37 of {\em Proceedings
  of Machine Learning Research}, pages 1180--1189, Lille, France. PMLR.

\bibitem[Gibbs and Candes, 2021]{gibbs2021adaptive}
Gibbs, I. and Candes, E. (2021).
\newblock Adaptive conformal inference under distribution shift.
\newblock {\em Advances in Neural Information Processing Systems},
  34:1660--1672.

\bibitem[Gneiting and Raftery, 2007]{gneiting2007strictly}
Gneiting, T. and Raftery, A.~E. (2007).
\newblock Strictly proper scoring rules, prediction, and estimation.
\newblock {\em Journal of the American statistical Association},
  102(477):359--378.

\bibitem[Goodfellow et~al., 2014]{goodfellow2014}
Goodfellow, I., Pouget-Abadie, J., Mirza, M., Xu, B., Warde-Farley, D., Ozair,
  S., Courville, A., and Bengio, Y. (2014).
\newblock Generative adversarial nets.
\newblock In Ghahramani, Z., Welling, M., Cortes, C., Lawrence, N., and
  Weinberger, K., editors, {\em Advances in Neural Information Processing
  Systems}, volume~27. Curran Associates, Inc.

\bibitem[Gretton et~al., 2012]{gretton2012kernel}
Gretton, A., Borgwardt, K.~M., Rasch, M.~J., Sch{\"o}lkopf, B., and Smola, A.
  (2012).
\newblock A kernel two-sample test.
\newblock {\em The Journal of Machine Learning Research}, 13(1):723--773.

\bibitem[Gretton et~al., 2009]{gretton2009covariate}
Gretton, A., Smola, A., Huang, J., Schmittfull, M., Borgwardt, K., and
  Sch{\"o}lkopf, B. (2009).
\newblock Covariate shift by kernel mean matching.
\newblock {\em Dataset shift in machine learning}, 3(4):5.

\bibitem[Hardle et~al., 1993]{hardle1993optimal}
Hardle, W., Hall, P., and Ichimura, H. (1993).
\newblock Optimal smoothing in single-index models.
\newblock {\em The annals of Statistics}, 21(1):157--178.

\bibitem[He et~al., 2016]{he2016deep}
He, K., Zhang, X., Ren, S., and Sun, J. (2016).
\newblock Deep residual learning for image recognition.
\newblock In {\em Proceedings of the IEEE conference on computer vision and
  pattern recognition}, pages 770--778.

\bibitem[He, 1997]{he1997quantile}
He, X. (1997).
\newblock Quantile curves without crossing.
\newblock {\em The American Statistician}, 51(2):186--192.

\bibitem[Ho et~al., 2020]{ho2020denoising}
Ho, J., Jain, A., and Abbeel, P. (2020).
\newblock Denoising diffusion probabilistic models.
\newblock {\em Advances in neural information processing systems},
  33:6840--6851.

\bibitem[Hothorn et~al., 2014]{hothorn2014conditional}
Hothorn, T., Kneib, T., and B{\"u}hlmann, P. (2014).
\newblock Conditional transformation models.
\newblock {\em Journal of the Royal Statistical Society: Series B: Statistical
  Methodology}, pages 3--27.

\bibitem[Kingma and Ba, 2015]{kingma2014adam}
Kingma, D.~P. and Ba, J. (2015).
\newblock Adam: A method for stochastic optimization.
\newblock In {\em Proceedings of the 3rd International Conference on Learning
  Representations (ICLR)}.

\bibitem[Kingma and Welling, 2014]{kingma2013auto}
Kingma, D.~P. and Welling, M. (2014).
\newblock Auto-encoding variational bayes.
\newblock In {\em International Conference on Learning Representations}.

\bibitem[Kirschner et~al., 2020]{kirschner20distributionally}
Kirschner, J., Bogunovic, I., Jegelka, S., and Krause, A. (2020).
\newblock Distributionally robust bayesian optimization.
\newblock In Chiappa, S. and Calandra, R., editors, {\em Proceedings of the
  Twenty Third International Conference on Artificial Intelligence and
  Statistics}, volume 108 of {\em Proceedings of Machine Learning Research},
  pages 2174--2184. PMLR.

\bibitem[Kneib et~al., 2023]{kneib23review}
Kneib, T., Silbersdorff, A., and Säfken, B. (2023).
\newblock {Rage Against the Mean – A Review of Distributional Regression
  Approaches}.
\newblock {\em Econometrics and Statistics}, 26(C):99--123.

\bibitem[Koenker, 2005]{koenker2005quantile}
Koenker, R. (2005).
\newblock {\em Quantile regression}, volume~38.
\newblock Cambridge university press.

\bibitem[Koenker and Bassett, 1978]{koenker1978regression}
Koenker, R. and Bassett, G. (1978).
\newblock Regression quantiles.
\newblock {\em Econometrica: journal of the Econometric Society}, pages 33--50.

\bibitem[Kuhn et~al., 2019]{kuhn2019wasserstein}
Kuhn, D., Esfahani, P.~M., Nguyen, V.~A., and Shafieezadeh-Abadeh, S. (2019).
\newblock Wasserstein distributionally robust optimization: Theory and
  applications in machine learning.
\newblock In {\em Operations research \& management science in the age of
  analytics}, pages 130--166. Informs.

\bibitem[Kushner et~al., 2018]{kushner2018canadian}
Kushner, P.~J., Mudryk, L.~R., Merryfield, W., Ambadan, J.~T., Berg, A.,
  Bichet, A., Brown, R., Derksen, C., D{\'e}ry, S.~J., Dirkson, A., et~al.
  (2018).
\newblock Canadian snow and sea ice: assessment of snow, sea ice, and related
  climate processes in canada's earth system model and climate-prediction
  system.
\newblock {\em The Cryosphere}, 12(4):1137--1156.

\bibitem[Legendre, 1806]{legendre1806nouvelles}
Legendre, A.~M. (1806).
\newblock {\em Nouvelles m{\'e}thodes pour la d{\'e}termination des orbites des
  com{\`e}tes; par AM Legendre...}
\newblock chez Firmin Didot, libraire pour lew mathematiques, la marine, l~….

\bibitem[Lei et~al., 2018]{lei2018distribution}
Lei, J., G’Sell, M., Rinaldo, A., Tibshirani, R.~J., and Wasserman, L.
  (2018).
\newblock Distribution-free predictive inference for regression.
\newblock {\em Journal of the American Statistical Association},
  113(523):1094--1111.

\bibitem[Leroux et~al., 2019]{leroux2019organizing}
Leroux, A., Di, J., Smirnova, E., Mcguffey, E.~J., Cao, Q., Bayatmokhtari, E.,
  Tabacu, L., Zipunnikov, V., Urbanek, J.~K., and Crainiceanu, C. (2019).
\newblock Organizing and analyzing the activity data in nhanes.
\newblock {\em Statistics in biosciences}, 11:262--287.

\bibitem[Massart, 1990]{massart1990tight}
Massart, P. (1990).
\newblock The tight constant in the dvoretzky-kiefer-wolfowitz inequality.
\newblock {\em The annals of Probability}, pages 1269--1283.

\bibitem[Matheson and Winkler, 1976]{matheson1976scoring}
Matheson, J.~E. and Winkler, R.~L. (1976).
\newblock Scoring rules for continuous probability distributions.
\newblock {\em Management science}, 22(10):1087--1096.

\bibitem[McCullagh and Nelder, 1983]{mccullagh1983generalized}
McCullagh, P. and Nelder, J. (1983).
\newblock {\em Generalized Linear Models}.
\newblock Monographs on Statistics and Applied Probability. Springer US.

\bibitem[Mehta et~al., 2024]{mehta2024distributionally}
Mehta, R., Roulet, V., Pillutla, K., and Harchaoui, Z. (2024).
\newblock Distributionally robust optimization with bias and variance
  reduction.
\newblock In {\em The Twelfth International Conference on Learning
  Representations}.

\bibitem[Meinshausen, 2006]{meinshausen2006quantile}
Meinshausen, N. (2006).
\newblock Quantile regression forests.
\newblock {\em Journal of machine learning research}, 7(6).

\bibitem[Meinshausen and B{\"u}hlmann, 2015]{meinshausen2015maximin}
Meinshausen, N. and B{\"u}hlmann, P. (2015).
\newblock Maximin effects in inhomogeneous large-scale data.
\newblock {\em The Annals of Statistics}, 43(4):1801--1830.

\bibitem[N{\"a}f and Josse, 2024]{naf2024good}
N{\"a}f, J. and Josse, J. (2024).
\newblock What is a good imputation under mar missingness?
\newblock {\em arXiv preprint arXiv:2403.19196}.

\bibitem[Namkoong and Duchi, 2016]{NIPS2016_4588e674}
Namkoong, H. and Duchi, J.~C. (2016).
\newblock Stochastic gradient methods for distributionally robust optimization
  with f-divergences.
\newblock In Lee, D., Sugiyama, M., Luxburg, U., Guyon, I., and Garnett, R.,
  editors, {\em Advances in Neural Information Processing Systems}, volume~29.
  Curran Associates, Inc.

\bibitem[Osband et~al., 2024]{osband2024epistemic}
Osband, I., Wen, Z., Asghari, S.~M., Dwaracherla, V., Ibrahimi, M., Lu, X., and
  Van~Roy, B. (2024).
\newblock Epistemic neural networks.
\newblock {\em Advances in Neural Information Processing Systems}, 36.

\bibitem[Papamakarios et~al., 2021]{papamakarios2021normalizing}
Papamakarios, G., Nalisnick, E., Rezende, D.~J., Mohamed, S., and
  Lakshminarayanan, B. (2021).
\newblock Normalizing flows for probabilistic modeling and inference.
\newblock {\em Journal of Machine Learning Research}, 22(57):1--64.

\bibitem[Peters et~al., 2016]{peters2016causal}
Peters, J., B{\"u}hlmann, P., and Meinshausen, N. (2016).
\newblock Causal inference by using invariant prediction: identification and
  confidence intervals.
\newblock {\em Journal of the Royal Statistical Society. Series B (Statistical
  Methodology)}, pages 947--1012.

\bibitem[Ramesh et~al., 2022]{ramesh2022hierarchical}
Ramesh, A., Dhariwal, P., Nichol, A., Chu, C., and Chen, M. (2022).
\newblock Hierarchical text-conditional image generation with clip latents.
\newblock {\em arXiv preprint arXiv:2204.06125}, 1(2):3.

\bibitem[Replogle et~al., 2022]{replogle2022mapping}
Replogle, J.~M., Saunders, R.~A., Pogson, A.~N., Hussmann, J.~A., Lenail, A.,
  Guna, A., Mascibroda, L., Wagner, E.~J., Adelman, K., Lithwick-Yanai, G.,
  et~al. (2022).
\newblock Mapping information-rich genotype-phenotype landscapes with
  genome-scale perturb-seq.
\newblock {\em Cell}, 185(14):2559--2575.

\bibitem[Robbins and Monro, 1951]{robbins1951stochastic}
Robbins, H. and Monro, S. (1951).
\newblock A stochastic approximation method.
\newblock {\em The annals of mathematical statistics}, pages 400--407.

\bibitem[Rombach et~al., 2022]{rombach2022high}
Rombach, R., Blattmann, A., Lorenz, D., Esser, P., and Ommer, B. (2022).
\newblock High-resolution image synthesis with latent diffusion models.
\newblock In {\em Proceedings of the IEEE/CVF conference on computer vision and
  pattern recognition}, pages 10684--10695.

\bibitem[Rothenh{\"a}usler et~al., 2021]{rothenhausler2021anchor}
Rothenh{\"a}usler, D., Meinshausen, N., B{\"u}hlmann, P., and Peters, J.
  (2021).
\newblock Anchor regression: Heterogeneous data meet causality.
\newblock {\em Journal of the Royal Statistical Society Series B: Statistical
  Methodology}, 83(2):215--246.

\bibitem[Sagawa et~al., 2020]{Sagawa2020Distributionally}
Sagawa, S., Koh, P.~W., Hashimoto, T.~B., and Liang, P. (2020).
\newblock Distributionally robust neural networks.
\newblock In {\em International Conference on Learning Representations}.

\bibitem[Sejdinovic et~al., 2013]{sejdinovic2013equivalence}
Sejdinovic, D., Sriperumbudur, B., Gretton, A., and Fukumizu, K. (2013).
\newblock Equivalence of distance-based and rkhs-based statistics in hypothesis
  testing.
\newblock {\em The annals of statistics}, pages 2263--2291.

\bibitem[Shafer and Vovk, 2008]{shafer2008tutorial}
Shafer, G. and Vovk, V. (2008).
\newblock A tutorial on conformal prediction.
\newblock {\em Journal of Machine Learning Research}, 9(3).

\bibitem[Shen et~al., 2022]{shen2022asymptotic}
Shen, X., Chen, K., and Zhang, T. (2022).
\newblock Asymptotic statistical analysis of $ f $-divergence gan.
\newblock {\em arXiv preprint arXiv:2209.06853}.

\bibitem[Shen and Meinshausen, 2024]{shen2024distributional}
Shen, X. and Meinshausen, N. (2024).
\newblock Distributional principal autoencoders.
\newblock {\em arXiv preprint arXiv:2404.13649}.

\bibitem[Shi et~al., 2021]{shi2021double}
Shi, C., Xu, T., Bergsma, W., and Li, L. (2021).
\newblock Double generative adversarial networks for conditional independence
  testing.
\newblock {\em Journal of Machine Learning Research}, 22(285):1--32.

\bibitem[Sinha et~al., 2018]{sinha2018certifiable}
Sinha, A., Namkoong, H., and Duchi, J. (2018).
\newblock Certifiable distributional robustness with principled adversarial
  training.
\newblock In {\em International Conference on Learning Representations}.

\bibitem[Snoek et~al., 2012]{snoek2012practical}
Snoek, J., Larochelle, H., and Adams, R.~P. (2012).
\newblock Practical bayesian optimization of machine learning algorithms.
\newblock {\em Advances in neural information processing systems}, 25.

\bibitem[Sohl-Dickstein et~al., 2015]{sohl2015deep}
Sohl-Dickstein, J., Weiss, E., Maheswaranathan, N., and Ganguli, S. (2015).
\newblock Deep unsupervised learning using nonequilibrium thermodynamics.
\newblock In {\em International conference on machine learning}, pages
  2256--2265. PMLR.

\bibitem[Staib and Jegelka, 2019]{jegelka19distributionally}
Staib, M. and Jegelka, S. (2019).
\newblock Distributionally robust optimization and generalization in kernel
  methods.
\newblock In Wallach, H., Larochelle, H., Beygelzimer, A., d\textquotesingle
  Alch\'{e}-Buc, F., Fox, E., and Garnett, R., editors, {\em Advances in Neural
  Information Processing Systems}, volume~32. Curran Associates, Inc.

\bibitem[Steinwart and Ziegel, 2021]{steinwart2021strictly}
Steinwart, I. and Ziegel, J.~F. (2021).
\newblock Strictly proper kernel scores and characteristic kernels on compact
  spaces.
\newblock {\em Applied and Computational Harmonic Analysis}, 51:510--542.

\bibitem[Sugiyama et~al., 2007]{sugiyama2007covariate}
Sugiyama, M., Krauledat, M., and M{\"u}ller, K.-R. (2007).
\newblock Covariate shift adaptation by importance weighted cross validation.
\newblock {\em Journal of Machine Learning Research}, 8(5).

\bibitem[Sz{\'e}kely, 2003]{szekely2003statistics}
Sz{\'e}kely, G.~J. (2003).
\newblock E-statistics: The energy of statistical samples.
\newblock {\em Bowling Green State University, Department of Mathematics and
  Statistics Technical Report}, 3(05):1--18.

\bibitem[Sz{\'e}kely and Rizzo, 2023]{Szkely2023TheEO}
Sz{\'e}kely, G.~J. and Rizzo, M.~L. (2023).
\newblock {\em The Energy of Data and Distance Correlation}.
\newblock CRC Press.

\bibitem[Taleb and Jutten, 1999]{taleb1999source}
Taleb, A. and Jutten, C. (1999).
\newblock Source separation in post-nonlinear mixtures.
\newblock {\em IEEE Transactions on signal Processing}, 47(10):2807--2820.

\bibitem[Tibshirani et~al., 2019]{tibshirani2019conformal}
Tibshirani, R.~J., Foygel~Barber, R., Candes, E., and Ramdas, A. (2019).
\newblock Conformal prediction under covariate shift.
\newblock {\em Advances in neural information processing systems}, 32.

\bibitem[Vito, 2016]{misc_air_quality_360}
Vito, S. (2016).
\newblock {Air Quality}.
\newblock UCI Machine Learning Repository.
\newblock {DOI}: https://doi.org/10.24432/C59K5F.

\bibitem[Zhang and Hyv{\"a}rinen, 2009]{Zhang2009OnTI}
Zhang, K. and Hyv{\"a}rinen, A. (2009).
\newblock On the identifiability of the post-nonlinear causal model.
\newblock In {\em Conference on Uncertainty in Artificial Intelligence}.

\end{thebibliography}
\bibliographystyle{apalike}

\newpage
\appendix
\section{Proofs of the population results}
\subsection{Proof of Proposition~\ref{prop:dist_est_pop}}\label{pf:prop:dist_est_pop}
\begin{proof}[Proof of Proposition~\ref{prop:dist_est_pop}]
Our correct model specification assumption ensures the existence of such a $\tilde\gen$ that satisfies $\tilde\gen(x,\varepsilon)\sim \ptr(y|x)$ for all $x\in\cX$ almost everywhere.
Given any model $\gen\in\cM$, assume there exists a subset $\cX'\subseteq\cX$ with a nonzero base measure such that for all $x\in\cX'$, $P_{\tilde\gen}(y|x)\not\equiv\ptr(y|x)$. Then according to Lemma~\ref{lem:es}, we have for all $x\in\cX'$ that
\[\bbE_{Y\sim \ptr(y|x)}[\es(P_{\tilde\gen}(y|x), Y)] > \bbE_{Y\sim \ptr(y|x)}[\es(P_{\gen}(y|x), Y)].\]
Taking the expectation with respect to $\ptr(x)$ then yields
\[\bbE_{\ptr}[\es(P_{\tilde\gen}(y|X), Y)] > \bbE_{\ptr}[\es(P_{\gen}(y|X), Y)].\]
Thus we have 
\[\bbE[\cL_e(P_{\tilde\gen}(y|X);\ptr(y|X))] < \bbE[\cL_e(P_{\gen}(y|X);\ptr(y|X))],\]
which concludes the proof.
\end{proof}

\subsection{Proof of examples in Section~\ref{sec:extrapolability_def}}\label{pf:ex}
\begin{proof}[Proof of Example~\ref{ex:lin}]
Let $\cF=\{\beta^\top x\}$. For all $\beta,\beta'$ such that $\beta^\top x=\beta'^\top x$ for all $x\in\cX$, we have $\beta=\beta'$ and thus $\beta^\top x=\beta'^\top x$ for all $x\in\bbR$. Thus $\cU(\delta)=0$. 
\end{proof}

\begin{proof}[Proof of Example~\ref{ex:lip}]
Let $x_b = \argmin_{x\in\cX}\|x-x'\|$. For all $f,f'\in\cF$ such that $f(x)=f'(x)$ for all $x\in\cX$, we have $f(x_b)=f'(x_b)$ and then
\begin{align*}
	f(x')-f'(x') &= f(x')-f(x_b)+f'(x_b)-f'(x')\\
	&= (\nabla f(\tx)  - \nabla f'(\tx'))^\top (x'-x_b),
\end{align*}
where $\tx$ and $\tx'$ are points on the line segment between $x'$ and $x_b$, so $\tx,\tx'\notin\cX$. Thus
\begin{align*}
	\cU(\delta) &= \sup_{x':d(x,x')\leq\delta}\sup_{{\substack{f,f'\in\cF:\\D_{x}(f,f')=0,\forall x\in\cX}}}|(\nabla f(\tx)  - \nabla f'(\tx'))^\top (x'-x_b)|\\
	&= \sup_{x':d(x,x')\leq\delta} 2L\|x'-x_b\|=2L\delta.
\end{align*}
\end{proof}
\begin{proof}[Proof of Example~\ref{ex:monot}]
For all $x'\notin\cX$, $\sup_{f,f'}|f(x')-f'(x')|=\infty$. Thus $\cU(\delta)=\infty$.
\end{proof}

\subsection{Proof of Theorem~\ref{thm:extra_anm}}\label{pf:thm:extra_preanm}
Recall the pre-ANM class defined in the main text
$$\cM_{\mathrm{pre}}=\{g(x+h(\varepsilon))+\beta x:g\in\cG,h\in\cH,\beta\in\bbR\}.$$
We first study the identifiability of $\cM_{\mathrm{pre}}$ which serves as the basis for the investigation of extrapolability. 
The identifiability refers to the uniqueness of a model that induces a single (conditional distribution). The following proposition formalises the identifiability of $\cM_{\mathrm{pre}}$. It says if the two pre-ANMs induce the same conditional distribution of $Y$ given $X=x$, then their functions or parameters would also be the same up to a certain extent of unidentifiability. 

\begin{proposition}[Identifiability of pre-ANMs]\label{prop:identify_preanm}
	Assume functions in $\cG$ are strictly monotone.
	Let $\beta,\beta'\in\bbR$, $g,g'\in\cG$, and $h,h'\in\cH$ be such that $g(x+h(\varepsilon))+\beta x\overset{d}=g'(x+h'(\varepsilon))+\beta' x$ for all $x\in\cX$. 
	\begin{enumerate}[label=(\roman*)]
		\item Assume functions in $\cG$ are twice differentiable and functions in $\cH$ are differentiable. If $g$ is nonlinear, i.e.\ $\ddot{g}\neq0$ a.e., then $\beta=\beta'$, $g(x+h(\varepsilon))=g'(x+h(\varepsilon))$ for all $x\in\cX$ and $\varepsilon\in[0,1]$, and $h(\varepsilon)=h'(\varepsilon)$ for all $\varepsilon\in[0,1]$;
		\item Assume functions in $\cG$ and $\cH$ are differentiable. If $g$ is linear, i.e.\ $g(x)=\alpha x + \varsigma$, then $g'(x+h'(\varepsilon))=\alpha' (x+h'(\varepsilon)) + \varsigma'$ for all $x\in\cX$ and $\varepsilon\in[0,1]$, where $\varsigma=\varsigma'$ and $\alpha+\beta=\alpha'+\beta'$, and $\alpha h(\varepsilon)=\alpha'h'(\varepsilon)$.
	\end{enumerate}
\end{proposition}
\begin{proof}[Proof of Proposition~\ref{prop:identify_preanm}]
For all $x\in\cX$, since $g(x+h(\varepsilon))+\beta x\overset{d}=g'(x+h'(\varepsilon))+\beta'x$, the $\alpha$-quantiles of both sides are equal. As $g,g',h,h'$ are strictly monotone, we have $g(x+h(Q_\alpha(\varepsilon)))+\beta x=g'(x+h'(Q_\alpha(\varepsilon)))+\beta'x$ for all $\alpha\in[0,1]$. Let $\varepsilon=Q_\alpha(\varepsilon)$. This implies  
\begin{equation}\label{eq0}
	g(x+h(\varepsilon))+\beta x=g'(x+h'(\varepsilon))+\beta'x
\end{equation}
for all $x\in\cX$ and $\varepsilon\in[0,1]$. 

	Taking partial derivatives of both sides in \eqref{eq0} with respect to $\varepsilon$ and $x$ implies
	\begin{align}
		\dot{g}(x+h(\varepsilon))\dot{h}(\varepsilon)&=\dot{g}'(x+h'(\varepsilon))\dot{h}'(\varepsilon)\label{eq1}\\
		\beta+\dot{g}(x+h(\varepsilon))&=\beta'+\dot{g}'(x+h'(\varepsilon)).\label{eq2}
	\end{align}
	By plugging \eqref{eq2} into \eqref{eq1} and reorganising, we have 
	\begin{equation*}
		\dot{g}(x+h(\varepsilon))(\dot{h}(\varepsilon)-\dot{h}'(\varepsilon))=(\beta-\beta')\dot{h}'(\varepsilon).
	\end{equation*}
	
	Let us first consider the scenario of (i).
	Taking the derivative on both sides with respect to $x$ leads to
	\begin{equation*}
		\ddot{g}(x+h(\varepsilon))(\dot{h}(\varepsilon)-\dot{h}'(\varepsilon))=0
	\end{equation*}
	for all $x\in\cX,\varepsilon\in[0,1]$. Since $\ddot{g}\neq0$ a.e., we have $\dot{h}(\varepsilon)=\dot{h}'(\varepsilon)$ for all $\varepsilon\in[0,1]$. Since we assume $h(0.5)=h'(0.5)$, we have $h(\varepsilon)=h'(\varepsilon)$ for all $\varepsilon\in[0,1]$.
	
	Taking $\varepsilon=0$ in \eqref{eq0} implies
	\begin{equation*}
		\beta x+g(x)=\beta' x+g'(x),\quad \forall x\in\cX.
	\end{equation*}
	For all $x\in\cX$, let $\varepsilon_0\neq0$ be such that $x+h(\varepsilon_0)\in\cX$. Then \eqref{eq0} becomes
	\begin{equation*}
		\beta x+g(x+h(\varepsilon_0)) = \beta x+(\beta'-\beta)(x+h(\varepsilon_0))+g'(x+h(\varepsilon_0)) = \beta'x+g'(x+h(\varepsilon_0))
	\end{equation*}
	which leads to
	\begin{equation*}
		(\beta'-\beta)h(\varepsilon_0) = 0.
	\end{equation*}
	Since $h(\varepsilon_0) \neq0$, we have $\beta'=\beta$. Further by \eqref{eq0}, we have
	\begin{equation*}
		g(x+h(\varepsilon)) = g'(x+h(\varepsilon))
	\end{equation*}
	for all $x\in\cX$ and $\varepsilon\in[0,1]$. We hence conclude (i).
	
	As to the linear case in (ii), \eqref{eq2} become
	\begin{equation*}
		\beta+\alpha=\beta'+\dot{g}'(x+h'(\varepsilon))
	\end{equation*}
	for all $\varepsilon$. Thus, $\dot{g}'(x+h'(\varepsilon))$ has to be a constant and we denote it by $\dot{g}'(x+h'(\varepsilon))=\alpha'$, which implies $g'(x+h'(\varepsilon))=\alpha'(x+h'(\varepsilon))+\varsigma'$ and $\beta+\alpha=\beta'+\alpha'$.
	
	Back to \eqref{eq0}, we then have 
	\begin{equation*}
		\alpha(x+h(\varepsilon))+\varsigma+\beta x = \alpha'(x+h'(\varepsilon))+\varsigma'+\beta'x
	\end{equation*}
	for all $x\in\cX$ and $\varepsilon\in[0,1]$. This implies $\varsigma=\varsigma'$ and $\alpha h(\varepsilon)=\alpha'h'(\varepsilon)$ for all $\varepsilon$. We hence conclude (ii).
\end{proof}

\bigskip
Then we are ready to prove Theorem~\ref{thm:extra_anm}.
\begin{proof}[Proof of Theorem~\ref{thm:extra_anm}]
We first show that $\cM_{\mathrm{pre}}$ is distributionally extrapable. Let $g(x+h(\varepsilon))+\beta x,g'(x+h'(\varepsilon))+\beta' x$ be any two models in $\cM_{\mathrm{pre}}$. Following the definition of distributional extrapolability, suppose $g(x+h(\varepsilon))+\beta x\overset{d}=g'(x+h'(\varepsilon))+\beta' x$ for all $x\in\cX$. Then the strict monotonicity of $\cG$ allows us to apply Proposition~\ref{prop:identify_preanm}. 

If $\cM_{\mathrm{pre}}$ is nonlinear, we have from Proposition~\ref{prop:identify_preanm} (i) that $\beta=\beta'$, $g(x+h(\varepsilon))=g'(x+h(\varepsilon))$ for all $x\in\cX$ and $\varepsilon\in[0,1]$, and $h(\varepsilon)=h'(\varepsilon)$ for all $\varepsilon\in[0,1]$. Since $h$ is unbounded, this implies $g(x)=g'(x)$ for all $x\in\bbR$. Therefore, we have $g(x+h(\varepsilon))+\beta x\overset{d}=g'(x+h'(\varepsilon))+\beta'x$ for all $x\in\bbR$, which leads to $\cU_{\mathrm{pre}}(\delta)=0$. 

On the other hand, if $\preanm$ is linear, we know from Proposition~\ref{prop:identify_preanm} (ii) that $g(x)=\alpha x + \varsigma$, $g'(x)=\alpha' x + \varsigma'$ for all $x\in\bbR$, where $\varsigma=\varsigma'$ and $\alpha+\beta=\alpha'+\beta'$, and $\alpha h(\varepsilon)=\alpha'h'(\varepsilon)$. Note that in the linear case, $g(x+h(\varepsilon))+\beta x=(\alpha+\beta)x+\alpha h(\varepsilon)+\varsigma$.  Therefore, we have $g(x+h(\varepsilon))+\beta x\overset{d}=g'(x+h'(\varepsilon))+\beta'x$ for all $x\in\bbR$, which leads to $\cU_{\mathrm{pre}}(\delta)=0$. We thus conclude (i).

For post-ANMs, let $g(x)+h(\varepsilon)$ and $g'(x)+h(\varepsilon)$ be two models in $\cM_{\mathrm{post}}$ for any $g,g'\in\cG$ and the same $h\in\cH$. Suppose $g(x)+h(\varepsilon)\overset{d}=g'(x)+h(\varepsilon)$ for all $x\in\cX$, which is equivalent to say $g(x)=g'(x)$ for all $x\in\cX$. Then if $\cG$ is the class of all monotone functions, $g(x)$ and $g'(x)$ can differ by an arbitrarily large amount for any $x\notin\cX$. Thus $\cU_{\mathrm{post}}(\delta)=\infty$ for all $\delta>0$. 

Last, we look at the condition for $\cM_{\mathrm{post}}$ to be distributionally extrapable. We start from the sufficiency. Let $g(x)+h(\varepsilon)$ and $g'(x)+h'(\varepsilon)$ be any two models in $\cM_{\mathrm{post}}$. Suppose $g(x)+h(\varepsilon)\overset{d}=g'(x)+h'(\varepsilon)$ for all $x\in\cX$. By taking the median of both sides and noting that $h(0.5)=h'(0.5)=0$, we have $g(x)=g'(x)$ for all $x\in\cX$. Then we have $h(\varepsilon)\overset{d}=h'(\varepsilon)$, which implies $h(\varepsilon)=h'(\varepsilon)$ for any $\varepsilon\in[0,1]$. Now if $\cG$ is functionally extrapable, we have $g(x)=g'(x)$ for all $x\in\bbR$. This then implies $g(x)+h(\varepsilon)\overset{d}=g'(x)+h'(\varepsilon)$ for all $x\in\bbR$, that is, $\cM_{\mathrm{post}}$ is distributionally extrapable. This shows the sufficiency. 

To see the necessity, suppose $\cG$ is not functionally extrapable, meaning that there exists $\delta>0$ such that $\cU_{\cG}(\delta)>0$, which is equivalently to say there exists $x'\notin\cX$ such that $g(x')\neq g'(x)$. Then the distributions of $g(x')+h(\varepsilon)$ and $g'(x')+h'(\varepsilon)$ are not equal, so $\cU_{\mathrm{post}}(\delta)>0$, that is, $\cM_{\mathrm{post}}$ is not distributionally extrapable. We thus conclude the proof.
\end{proof}

\subsection{Proofs in Section~\ref{sec:local_extrap}}\label{app:pf_local_extrap}
\begin{proof}[Proof of Theorem~\ref{thm:func_recover}]
By Lemma~\ref{lem:es}, we have $\tilde{g}(x+\tilde{h}(\varepsilon))+\tilde\beta x\overset{d}=g^\star(x+h^\star(\varepsilon))+\beta^\star x$ for all $x\in\cX$. Then the desired results follow immediately from Proposition~\ref{prop:identify_preanm}.	
\end{proof}
\begin{proof}[Proof of Corollary~\ref{coro:local_quantile}]
	Note that $\tilde{q}_\alpha(x)=\tilde{g}(x+Q_\alpha(\eta))+\tilde\beta x$ and $q^\star_\alpha(x)=g^\star(x+Q_\alpha(\eta))+\beta^\star x$. Then the desired result follows from Theorem~\ref{thm:func_recover}.
\end{proof}

\subsection{Proofs in Section~\ref{sec:extra_gain}}\label{app:pf_extra_gain}
\begin{proof}[Proof of Proposition~\ref{prop:gain_median}]
For all $f,f'\in\cF^\textsl{m}_\rmE$, there exists $g,g'\in\cG^\star$ such that $f(x)=g(x)+\beta^\star x$ and $f'(x)=g'(x)+\beta^\star x$. Recall the definition $\cG^\star := \{g\in\cG:g(x)=g^\star(x), \forall x\le \xm+\etam\}.$ Thus, by Definition~\ref{def:extra_func} and the uniform Lipschitz of $\cG$, for $\delta\le\etam$, we have
\begin{align*}
	\cU_{\cF^\textsl{m}_\rmE}(\delta)&=\sup_{x'\in[\xm,\xm+\delta]}\sup_{f,f'\in\cF^\textsl{m}_\rmE}|f(x')-f'(x')|\\
	&=\sup_{f,f'\in\cF^\textsl{m}_\rmE}|f(\xm+\delta)-f'(\xm+\delta)|\\
	&=\sup_{g,g'\in\cG^\star}|g(\xm+\delta)-g'(\xm+\delta)|=0.
\end{align*}
For $\delta>\etam$, we have
\begin{align*}
	f(\xm+\delta)-f'(\xm+\delta) &= g(\xm+\delta)-g'(\xm+\delta)\\
	&= g(\xm+\etam)+\dot{g}(\tilde{x})(\delta-\etam) - [g'(\xm+\etam)+\dot{g}'(\tilde{x}')(\delta-\etam)]\\
	&= (\delta-\etam)(\dot{g}(\tilde{x}) - \dot{g}'(\tilde{x}')),
\end{align*}
where $\tx,\tx'\in[\xm+\etam,\xm+\delta]$. 
Then the extrapolation uncertainty of engression is
\begin{equation*}
	\cU_{\cF^\textsl{m}_\rmE}(\delta) = (\delta-\etam)\cdot\sup_{g,g'\in\cG^\star}|\dot{g}(\tilde{x}) - \dot{g}'(\tilde{x}')| = L(\delta-\etam),
\end{equation*}
by the $L$-Lipschitz condition in Assumption~\ref{ass:lip}. Thus, we have $\cU_{\cF^\textsl{m}_\rmE}(\delta)=L(\delta-\etam)_+$. 

For $L_1$ regression, for all $f,f'\in\cF_{L_1}$, we only have $f(x)=f'(x)$ for all $x\leq\xm$. Thus, for all $\delta>0$
\begin{align*}
	\cU_{\cF_{L_1}}(\delta) &= \sup_{f,f'\in\cF_{L_1}}|f(\xm+\delta)-f'(\xm+\delta)| \\
	&= \sup_{f,f'\in\cF_{L_1}}|g(\xm)+\dot{f}(\tilde{x})\delta - [f'(\xm)+\dot{f}'(\tilde{x}')\delta]|\\
	&= \delta\sup_{f,f'\in\cG}|\dot{f}(\tilde{x}) - \dot{f}'(\tilde{x}')|=L\delta,
\end{align*}
where $\tx,\tx'\in[\xm+\etam,\xm+\delta]$. 
The remaining results are straightforward to see.
\end{proof}

\bigskip
\begin{proof}[Proof of Proposition~\ref{prop:gain_mean}]
For $\delta>0$, for all $f,f'\in\cF^\mu_\rmE$, there exists $g,g'\in\cG^\star$ such that
\begin{align*}
	f(\xm+\delta) - f'(\xm+\delta) &= \bbE_{\eta}[g(\xm+\delta+\eta)] - \bbE_{\eta}[g'(\xm+\delta+\eta)]\\
	&= \int_{-\etam}^{\etam-\delta} g(\xm+\delta+\eta)dF_{\eta}(\eta) + \int_{\etam-\delta}^{\etam} g(\xm+\delta+\eta)dF_{\eta}(\eta) \\&\quad- \left(\int_{-\etam}^{\etam-\delta} g'(\xm+\delta+\eta)dF_{\eta}(\eta) + \int_{\etam-\delta}^{\etam} g'(\xm+\delta+\eta)dF_{\eta}(\eta)\right)\\
	&= \int_{\etam-\delta}^{\etam}\left[g(\xm+\delta+\eta) - g'(\xm+\delta+\eta)\right]dF_{\eta}(\eta)\\
	&= \int_{\etam-\delta}^{\etam}(\dot{g}(\tx) - \dot{g}'(\tx'))(\eta-\etam+\delta) dF_{\eta}(\eta),
\end{align*}
where $\tx,\tx'\in[\xm+\etam,\xm+\delta+\eta]$ for $\eta>\etam-\delta$. Then the extrapolation uncertainty of engression is 
\begin{align*}
	\cU_{\cF^\mu_\rmE}(\delta) &= \sup_{f,f'\in\cF^\mu_\rmE}|f(\xm+\delta) - f'(\xm+\delta)| \\
	&= \sup_{g,g'\in\cG^\star}\left|\int_{\etam-\delta}^{\etam}(\dot{g}(\tx) - \dot{g}'(\tx'))(\eta-\etam+\delta) dF_{\eta}(\eta)\right|\\
	&= L\int_{\etam-\delta}^{\etam}(\eta-\etam+\delta) dF_{\eta}(\eta).
\end{align*}
When $\delta\ge2\etam$, we have $\cU_{\cF^\mu_\rmE}(\delta)=L(\delta-\etam)$.

As to $L_2$ regression, for all $f,f'\in\cF_{L_2}$, we only have $f(x)=f'(x)$ for all $x\le\xm$. Thus, for all $\delta>0$,
\begin{equation*}
	\cU_{\cF_{L_2}}(\delta) = \sup_{f,f'\in\cF_{L_2}}|f(\xm+\delta)-f'(\xm+\delta)| = L\delta,
\end{equation*}
similar to the case of $L_1$ regression. 

Then the extrapolability gain is given by
\begin{align}
	\gamma^\mu(\delta) &= L\delta - L\int_{\etam-\delta}^{\etam}(\eta-\etam+\delta) dF_{\eta}(\eta)\nonumber\\
	&= L\left[\delta\int_{-\etam}^{\etam-\delta}dF_{\eta}(\eta) + \int_{\etam-\delta}^{\etam}(\etam-\eta)dF_{\eta}(\eta)\right]\nonumber\\
	&= L\delta F_{\eta}(\etam-\delta) + L\int_{\etam-\delta}^{\etam}(\etam-\eta)d F_{\eta}(\eta)\label{eq:dist_gain1}\\
	&= L(\delta-\etam)F_{\eta}(\etam-\delta) + L\etam - L\int_{\etam-\delta}^{\etam}\eta d F_{\eta}(\eta).\label{eq:dist_gain2}
\end{align}
Since both terms in \eqref{eq:dist_gain1} are non-negative, we have $\gamma^\mu(\delta)>0$ for all $\delta>0$. When $0<\delta\le\etam$, \eqref{eq:dist_gain2} is monotone increasing with respect to $\delta$ and $\gamma^\mu(\etam)=L\etam-L\int_0^{\etam}\eta d F_{\eta}(\eta)>0$. When $\delta>\etam$, we have
\begin{align*}
	L\etam - \gamma^\mu(\delta) &= L\int_{\etam-\delta}^{\etam}\eta dF_{\eta}(\eta) - L(\delta-\etam)F_{\eta}(\etam-\delta)\\
	&= L\int_{\delta-\etam}^{\etam}\eta dF_{\eta}(\eta) - L(\delta-\etam)(1-F_{\eta}(\delta-\etam)) \ge 0
\end{align*} 
where the second equality is due to the symmetry of the distribution of $\eta$, and the inequality becomes an equality if and only if $\delta\ge2\etam$. 
Thus, we conclude the proof.
\end{proof}

\bigskip
\begin{proof}[Proof of Proposition~\ref{prop:gain_dist}]
For one-dimensional distributions $p,p'$, the Wasserstein-$\ell$ distance between them is 
\begin{equation*}
	W_\ell(p,p')=\left(\int_0^1|Q_\alpha(p)-Q_\alpha(p')|^\ell d\alpha)\right)^{1/\ell},
\end{equation*}
where $Q_\alpha(p)$ is the $\alpha$-quantile of $p$. 

For $p(y|x),p'(y|x)\in\cP_\rmE$, there exists $g,g'\in\cG^\star$ such that $g(x+\eta)+\beta^\star x\sim p(y|x)$ and $g'(x+\eta)+\beta^\star x\sim p'(y|x)$. For $\alpha\in[0,1]$, denote by $q_\alpha(x)$ and $q'_\alpha(x)$ the $\alpha$-quantiles of $p(y|x)$ and $p'(y|x)$, respectively. Note that $q_\alpha(x)=g(x+Q_\alpha(\eta))+\beta^\star x$. By Corollary~\ref{coro:local_quantile}, we have $q_\alpha(x)=q'_\alpha(x)$ for all $x\le\xm+\etam-Q_\alpha(\eta)$. Let $\cF^\alpha_\rmE$ be the class of $\alpha$-quantile functions of all $p\in\cP_\rmE$, i.e.\ $\cF^\alpha_\rmE=\{g(x+Q_\alpha(\eta))+\beta^\star x:g\in\cG^\star\}$.

Then for all $\delta\ge \etam-Q_\alpha(\eta)$, we have 
\begin{align*}
	q_\alpha(\xm+\delta)-q'_\alpha(\xm+\delta) &= q_\alpha(\xm+\etam-Q_\alpha(\eta)) + \dot{q}_\alpha(\tx)(\delta-\etam+Q_\alpha(\eta)) \\
	&\quad - [q'_\alpha(\xm+\etam-Q_\alpha(\eta)) + \dot{q}'_\alpha(\tx')(\delta-\etam+Q_\alpha(\eta))] \\
	&=(\delta-\etam+Q_\alpha(\eta))(\dot{q}_\alpha(\tx)-\dot{q}'_\alpha(\tx'))\\
	&=(\delta-\etam+Q_\alpha(\eta))(\dot{g}_\alpha(\tx+Q_\alpha(\eta))-\dot{g}'_\alpha(\tx'+Q_\alpha(\eta))),
\end{align*}
where $\tx,\tx'\in[\xm+\etam-Q_\alpha(\eta),\xm+\delta]$. 

For $\delta>0$, let $\alpha_0(\delta)=F_{\eta}(\xm+\delta)$ so that $\alpha<\alpha_0(\xm+\delta)\Leftrightarrow \delta\le\etam-Q_\alpha(\eta)$. When $\delta\ge2\etam$, we have $\alpha_0(\delta)=0$.
Then the above implies the distributional extrapolation uncertainty of engression is given by
\begin{align*}
	\cU^\ell_{\cP_\rmE}(\delta) &=	\sup_{p,p'\in\cP_\rmE}W_\ell(p(y|\xm+\delta),p'(y|\xm+\delta)) \\
	&= \left(\int_0^1\sup_{q_\alpha,q'_\alpha\in\cF^\alpha_\rmE}|q_\alpha(\xm+\delta)-q'_\alpha(\xm+\delta)|^\ell d\alpha\right)^{1/\ell}\\
	&= \left(\int_{\alpha_0(\delta)}^1\sup_{q_\alpha,q'_\alpha\in\cF^\alpha_\rmE}|q_\alpha(\xm+\delta)-q'_\alpha(\xm+\delta)|^\ell d\alpha\right)^{1/\ell}\\
	&= \left(\int_{\alpha_0(\delta)}^1(\delta-\etam+Q_\alpha(\eta))\sup_{g,g'\in\cG}|\dot{g}_\alpha(\tx+Q_\alpha(\eta))-\dot{g}'_\alpha(\tx'+Q_\alpha(\eta))|^\ell d\alpha\right)^{1/\ell}\\
	&= L\left(\int_{\alpha_0(\delta)}^1(\delta-\etam+Q_\alpha(\eta))^\ell d\alpha\right)^{1/\ell}\\
	&= L\left(\int_{\etam-\delta}^{\etam}(\eta-\etam+\delta)^\ell dF_{\eta}(\eta)\right)^{1/\ell}>0,
\end{align*}
for all $\delta>0$ and $\ell>0$. 

As to quantile regression, given $\alpha$, for all $f,f'\in\cF^\alpha_{\mathrm{QR}}$, we only have $f(x)=f'(x)$ for all $x\le\xm$. Thus, we have
\begin{align*}
	\cU^\ell_{\cP_{\mathrm{QR}}}(\delta) &= \sup_{p,p'\in\cP_{\mathrm{QR}}}W_\ell(p(y|\xm+\delta),p'(y|\xm+\delta)) \\
	&= \left(\int_0^1\sup_{f,f'\in\cF^\alpha_{\mathrm{QR}}}|f(\xm+\delta)-f'(\xm+\delta)|^\ell d\alpha\right)^{1/\ell}\\
	&= L\delta.
\end{align*}
Thus, for all $\ell>0$, $\delta>0$, we have $\cU^\ell_{\cP_\rmE}(\delta)<\cU^\ell_{\cP_{\mathrm{QR}}}(\delta)$, so $\gamma^d_\ell(\delta)>0$.

When $\ell=1$, we have
\begin{equation*}
	\cU^1_{\cP_\rmE}(\delta) = L\int_{\etam-\delta}^{\etam}(\eta-\etam+\delta) dF_{\eta}(\eta).
\end{equation*}
Then the (maximum) extrapolation gain can be derived the same as in the proof of Proposition~\ref{prop:gain_mean}, which leads to (i). 

When $\ell\in(1,\infty)$, note that $\cU^\ell_{\cP_\rmE}(\delta)$ is monotone increasing with respect to $\ell$ while $\cU^\ell_{\cP_{\mathrm{QR}}}(\delta)$ is free of $\ell$, so $0<\gamma^d_\ell(\delta)<\gamma^1_\ell(\delta)$ for all $\delta$. Also, when $\delta>2\etam$,
\begin{equation*}
	\frac{\cU^\ell_{\cP_\rmE}(\delta)}{\cU^1_{\cP_\rmE}(\delta)} = \frac{\big(\bbE_{\eta}\big[(\eta+\delta-\etam)^\ell\big]\big)^{1/\ell}}{\delta-\etam} = \Bigg(\bbE\Bigg[\bigg(\frac{\eta}{\delta-\etam}+1\bigg)^\ell\Bigg] \Bigg)^{1/\ell}\to1
\end{equation*}
as $\delta\to\infty$. Also, $\lim_{\delta\to\infty}\big(\cU^\ell_{\cP_\rmE}(\delta)/\cU^\ell_{\cP_{\mathrm{QR}}}(\delta)\big)=1$. Thus, 
\begin{equation*}
	\lim_{\delta\to\infty}\frac{\gamma^d_\ell(\delta)}{\gamma^1_\ell(\delta)}=\lim_{\delta\to\infty}\frac{\gamma^d_\ell(\delta)}{L\etam}=\lim_{\delta\to\infty}\frac{\cU^\ell_{\cP_{\mathrm{QR}}}(\delta)-\cU^\ell_{\cP_\rmE}(\delta)}{\cU^1_{\cP_{\mathrm{QR}}}(\delta)-\cU^1_{\cP_\rmE}(\delta)}=1
\end{equation*}
meaning that $\lim_{\delta\to\infty}\gamma^d_\ell(\delta)=L\etam$, which is the maximum extrapolability gain.

When $\ell=0$, note for all $\delta>0$,
\begin{equation*}
	\cU^\infty_{\cP_\rmE}(\delta)=L\sup_{\eta\in[\etam-\delta,\etam]}(\eta-\etam+\delta)= L\delta,
\end{equation*}
implying $\gamma^d_\infty(\delta)\equiv0$. Therefore, we conclude the proof.
\end{proof}

\section{Proofs of Section~\ref{sec:finite_sample}}\label{app:pf_finite_sample}
\subsection{Preliminaries} 
For two cdf's $F$ and $G$ over $\bbR$, recall that the Cram\'er distance between them is defined as
\begin{equation*}
	\cramer(F,G)=\int_{-\infty}^\infty(F(z)-G(z))^2dz.
\end{equation*}

The following lemma presents a bound of the Cram\'er distance between the population and empirical distributions.
\begin{lemma}\label{lem:bound_cramer_empirical}
	Let $Z$ be a random variable with a continuous cdf $F$ with a bounded support $\cZ\subseteq\bbR$ and $\{Z_1,\dots,Z_n\}$ be an i.i.d.\ sample whose empirical distribution function (edf) is denoted by $\hat{F}_n(z)=\frac{1}{n}\sum_{i=1}^n\ind_{\{Z_i\le z\}}$. Then with probability exceeding $1-\delta$, we have
	\begin{equation*}
		\cramer(F,\hat{F}_n) \le \frac{\nu(\cZ)\log(2/\delta)}{2n},
	\end{equation*} 
	where $\nu$ is the Lebesgue measure.
\end{lemma}
\begin{proof}[Proof of Lemma~\ref{lem:bound_cramer_empirical}]
	By the Dvoretzky-Kiefer-Wolfowitz inequality~\citep{massart1990tight}, for all $t>0$, we have 
	\begin{equation*}
		\bbP\left(\sup_{z\in\cZ}|\hat{F}_n(z)-F(z)| \ge t\right) \le 2e^{-2nt^2}.
	\end{equation*}

	Therefore, with probability exceeding $1-\delta$, we have 
	\begin{equation*}
		\sup_{z\in\cZ}|\hat{F}_n(z)-F(z)| \le \sqrt{\frac{\log(2/\delta)}{2n}}.
	\end{equation*}
	Then by definition,
	\begin{equation*}
		\cramer(F,\hat{F}_n) = \int_{\cZ}\left(F(z)-\hat{F}_n(z)\right)^2dz \le \nu(\cZ)\frac{\log(2/\delta)}{2n}.
	\end{equation*}
\end{proof}

The following lemma bounds the difference between quantiles in terms of the Cram\'er distance. 
\begin{lemma}\label{lem:bound_quantile_by_cramer}
	Let $F,G$ be any monotone cdf's with the support $\cZ\subseteq\bbR$. Assume $F$ has a density $f$ bounded away from 0, i.e.\ $b=\inf_{z\in\cZ}f(z)>0$. Let $Q_\alpha^F$ and $Q_\alpha^G$ be quantile functions of $F$ and $G$, respectively. Then we have 
	\begin{equation*}
		\sup_{\alpha\in[0,1]}\big|Q_\alpha^F-Q_\alpha^G\big| \le \left(\frac{3\cramer(F,G)}{b^2}\right)^{\frac{1}{3}}.
	\end{equation*}
\end{lemma}
\begin{proof}[Proof of Lemma~\ref{lem:bound_quantile_by_cramer}]
Let $\tau=\sup_{\alpha\in[0,1]}|Q_\alpha^F-Q_\alpha^G|$, which is achieved at $\alpha=\alpha_0$. Assume without loss of generality that $Q_{\alpha_0}^F= Q_{\alpha_0}^G+\tau$. As the derivative of $F$ is lower bounded by $b$, we have for $z\in[Q_{\alpha_0}^G,Q_{\alpha_0}^F]$ that $G(z)-F(z)\ge b(z-Q_{\alpha_0}^G)$. Thus, 
\begin{equation*}
	\cramer(F,G)=\int_\cZ(F(z)-G(z))^2 dz \ge \int_{Q_{\alpha_0}^G}^{Q_{\alpha_0}^F}(F(z)-G(z))^2 dz \ge \int_0^\tau(bu)^2du = \frac{b^2\tau^3}{3}.
\end{equation*}
This implies the uniform bound for the difference between two quantile functions.
\end{proof}

\subsection{Proof of Proposition~\ref{prop:finite_sample_reg}}\label{app:pf_finite_sample_reg}
\begin{proof}[Proof of Proposition~\ref{prop:finite_sample_reg}]
We first show the result for $L_2$ regression. Note that the true conditional mean is given by $\mu^\star(x)=\beta_0^\star+\beta_1^\star x + \beta_2^\star (x^2+\bbE[{\eta^\star}^2])$. Empirical $L_2$ regression yields the following
\begin{equation}\label{eq:finite_l2}
	\beta^\dag_0+\beta^\dag_1x+\beta^\dag_2x^2 = \beta_0^\star+\beta_1^\star x + \beta_2^\star (x^2+\bbE[{\eta^\star}^2]) + o_p(1)
\end{equation}
for $\beta^\dag\in\cB^\dag$ and $x=x_1,x_2$. From the two equations (up to a small error that converges to 0 in probability as $n\to\infty$), one cannot (approximately) identify three unknown parameters. To see this, let $c$ be any nonzero real number and take $\beta^\dag_0=\beta^\star_0+c$. Then \eqref{eq:finite_l2} becomes
\begin{equation*}
	c+\beta^\dag_1x+\beta^\dag_2x^2 = \beta_1^\star x + \beta_2^\star (x^2+\bbE[{\eta^\star}^2]) + o_p(1).
\end{equation*}
Then we have
\begin{align*}
	\beta_2^\dag - \beta^\star_2 &= \frac{c+\beta^\star_2\bbE[{\eta^\star}^2]}{x_1x_2} + o_p(1)\\
	\beta_1^\dag - \beta^\star_1 &= -(c+\beta_2^\star\bbE[{\eta^\star}^2])\left(\frac{1}{x_1}+\frac{1}{x_2}\right) + o_p(1).
\end{align*}

Thus, for any $x\notin\cX$, we have
\begin{align*}
	\bbE[(Y-(\beta^\dag_0+\beta^\dag_1x+\beta^\dag_2x^2))^2] &\ge (\beta_0^\star+\beta_1^\star x + \beta_2^\star (x^2+\bbE[{\eta^\star}^2])-(\beta^\dag_0+\beta^\dag_1x+\beta^\dag_2x^2))^2 \\
	&= \left[\frac{c(x-x_1)(x-x_2)}{x_1x_2} + \beta_2^\star\bbE[{\eta^\star}^2]x\frac{x_1+x_2}{x_1x_2} - \frac{\beta_2^\star\bbE[{\eta^\star}^2]}{x_1x_2}\right]^2\\
	&\ge \left[\frac{c(x-x_1)(x-x_2)}{x_1x_2}\right]^2
\end{align*}
which can be made arbitrarily large by choosing $c$, as $x\neq x_1$ and $x\neq x_2$.

Next, we show the results for $L_1$ regression, which can be extended to quantile regression in a straightforward way. Empirical $L_1$ regression yields
\begin{equation}\label{eq:finite_l1}
	\beta^\ddag_0+\beta^\ddag_1x+\beta^\ddag_2 x^2 = \beta^\star_0 +\beta^\star_1 x + \beta^\star_2 x^2 + o_p(1),
\end{equation}
for $\beta^\ddag\in\cB^\ddag$ and $x=x_1,x_2$. Again, from these two equations, one cannot (approximately) identify three unknown parameters. Similarly, let $c'$ be any nonzero real number and take $\beta^\ddag_0=\beta^\star_0+c'$. Then from \eqref{eq:finite_l1} we have
\begin{equation*}
	\begin{split}
		\beta_1^\ddag - \beta_1^\star &= -c'\left(\frac{1}{x_1}+\frac{1}{x_2}\right)\\
		\beta_2^\ddag - \beta_2^\star &= \frac{c'}{x_1x_2}
	\end{split}
\end{equation*}

Thus, for any $x\notin\cX$, we have
\begin{equation*}
	|\beta^\ddag_0+\beta^\ddag_1x+\beta^\ddag_2 x^2 - (\beta^\star_0 +\beta^\star_1 x + \beta^\star_2 x^2)| \ge \left|\frac{c(x-x_1)(x-x_2)}{x_1x_2}\right| 
\end{equation*}
which can be made arbitrarily large by choosing $c$, as $x\neq x_1$ and $x\neq x_2$. Thereby, we conclude the proof.
\end{proof}

\subsection{Proof of Theorem~\ref{thm:finite_sample_eng}}\label{app:pf_thm_finite_sample_eng}
A quadratic pre-ANM defined in \eqref{eq:quadratic_preanm} is indexed by parameter $\beta=(\beta_1,\beta_2,\beta_3)$ and a noise variable $\eta$. Let $F_{\beta,\eta|x}$ be the conditional cdf of $Y$ given $X=x$ induced by a model $(\beta,\eta)$. Let $F^\star_x=F_{\beta^\star,\eta^\star|x}$ be the true conditional cdf and $\hat{F}_{n|x}$ be the corresponding empirical distribution function. Note that in the posited setting, the empirical engression \eqref{eq:eng_emp} is equivalent to minimising the Cram\'er distance between the empirical distribution and the distribution induced by the model:
\begin{equation*}
	\min_{\beta,\eta}\sum_{x\in\cX}\cramer(F_{\beta,\eta|x},\hat{F}_{n|x}).
\end{equation*}
Let $Q_\alpha^\eta$ be the $\alpha$-quantile of a noise variable $\eta$.

\begin{proof}[Proof of Theorem~\ref{thm:finite_sample_eng}]
The proof proceeds in 5 steps. 
\begin{itemize}
	\item In step I, we bound the Cram\'er distance between the estimated conditional distribution $F_{\hat\beta,\hat\eta|x}$ and the true one $F^\star_{x}$; 
	\item in step II, for any model indexed by $(\beta,\eta)$, for a level $\alpha$, we bound $\|\beta - \beta^\star\|$ and $|Q_\alpha^\eta-Q_\alpha^{\eta^\star}|$ in terms of differences in their corresponding conditional quantiles of $Y|X$; 
	\item in step III, we combine the previous two steps and yield the finite-sample bound for parameter estimation; 
	\item in the last two steps, we derive the bounds for conditional quantile and mean estimation, respectively.
\end{itemize}

\medbreak
\noindent{\textit{Step I. }} 
With probability exceeding $1-\delta$, we have
\begin{equation}\label{eq:bound_cramer}
\begin{split}
	\sum_{x\in\cX}\cramer(F_{\hat\beta,\hat\eta|x},F^\star_{x}) &\le 2\sum_{x\in\cX}\left[\cramer(F_{\hat\beta,\hat\eta|x},\hat{F}_{n|x}) + \cramer(F^\star_{x},\hat{F}_{n|x}) \right]\\
	&\le 4\sum_{x\in\cX}\cramer(F^\star_{x},\hat{F}_{n|x})\\
	&\le \frac{2\nu(\cY)\log(2/\delta)}{n}
\end{split}
\end{equation}
with probability exceeding $1-\delta$, 
where $\cY$ is the training support of $Y$; the first inequality follows by Cauchy--Schwarz inequality; the second inequality is due to the fact that 
\begin{equation*}
	\cramer(F_{\hat\beta,\hat\eta|x},\hat{F}_{n|x}) = \min_{F} \cramer(F,\hat{F}_{n|x})),
\end{equation*}
where $F$ is taken among the cdf's induced by any model in \eqref{eq:quadratic_preanm}, by the definition of the empirical engression; the third inequality follows by applying Lemma~\ref{lem:bound_cramer_empirical}. 

\bigskip
\noindent{\textit{Step II. }} 
Consider any $\beta\in\cB$ and $\eta\sim P_\eta\in\cP_\eta$. Let $q_{\alpha}(x)$ be the $\alpha$-quantile of $\beta_0+\beta_1 (x+\eta)+\beta_2(x+\eta)^2$, $q^\star_\alpha(x)$ be the $\alpha$-quantile of $\beta^\star_0+\beta^\star_1 (x+\eta)+\beta^\star_2(x+\eta)^2$, and $Q_\alpha^\eta$ be the $\alpha$-quantile of $\eta$. By monotonicity, we have $q_{\alpha}(x)=\beta_0+\beta_1 x +\beta_2 x^2+(\beta_1+2\beta_2x)Q_\alpha^\eta+\beta_2{Q_\alpha^\eta}^2$. Let $\delta_\alpha(x)=q_{\alpha}(x) - q^\star_\alpha(x)$ be the difference between the two quantiles. The symmetry of the noise implies $Q_\alpha^\eta=-Q_{1-\alpha}^\eta$. For any $\alpha\in(0.5,1]$, it holds that
\begin{align}
	\beta_0+\beta_1 x +\beta_2 x^2+(\beta_1+2\beta_2x)Q_\alpha^\eta+\beta_2{Q_\alpha^\eta}^2 &= \beta^\star_0+\beta^\star_1 x +\beta^\star_2 x^2+(\beta^\star_1+2\beta^\star_2x)Q_\alpha^{\eta^\star}+\beta^\star_2{Q_\alpha^{\eta^\star}}^2 + \delta_\alpha(x) \label{eq:q_alpha}\\
	\beta_0+\beta_1 x +\beta_2 x^2-(\beta_1+2\beta_2x)Q_\alpha^\eta+\beta_2{Q_\alpha^\eta}^2 &= \beta^\star_0+\beta^\star_1 x +\beta^\star_2 x^2-(\beta^\star_1+2\beta^\star_2x)Q_\alpha^{\eta^\star}+\beta^\star_2{Q_\alpha^{\eta^\star}}^2 + \delta_{1-\alpha}(x) \label{eq:q_1_alpha}\\
	\beta_0+\beta_1 x +\beta_2 x^2 &= \beta^\star_0+\beta^\star_1 x +\beta^\star_2 x^2 + \delta_{0.5}(x)\label{eq:q_med}
\end{align}
\normalsize
for $x=x_1,x_2$. 

Taking the difference between \eqref{eq:q_alpha} and \eqref{eq:q_1_alpha} yields
\begin{equation}\label{eq:q_alpha_1_alpha}
	2(\beta_1+2\beta_2x)Q_\alpha^\eta = 2(\beta^\star_1+2\beta^\star_2x)Q_\alpha^{\eta^\star} + \delta_\alpha(x) - \delta_{1-\alpha}(x).
\end{equation}
Taking the difference of the above equation for $x=x_1$ and $x_2$ implies
\begin{equation*}
	\beta_2Q_\alpha^\eta = \beta^\star_2Q_\alpha^{\eta^\star} + \underbrace{\frac{\delta_\alpha(x_2) - \delta_\alpha(x_1) - \delta_{1-\alpha}(x_2) + \delta_{1-\alpha}(x_1)}{4(x_2-x_1)}}_{=:\tau_2(x_1,x_2,\alpha)}.
\end{equation*}
Plugging this into \eqref{eq:q_alpha_1_alpha} for $x=x_1$ leads to 
\begin{equation}\label{eq:beta_1_q}
	\beta_1Q_\alpha^\eta = \beta^\star_1Q_\alpha^{\eta^\star} \underbrace{- \frac{x_1}{2}\tau_2(x_1,x_2,\alpha) + \frac{\delta_\alpha(x_1) - \delta_{1-\alpha}(x_1)}{2}}_{=:\tau_1(x_1,x_2,\alpha)}.
\end{equation}
Taking the ratio of the above two equations, we have
\begin{equation}\label{eq:beta_ratio}
	\frac{\beta_2}{\beta_1} = \frac{\beta^\star_2}{\beta^\star_1} + \underbrace{\frac{\tau_2-\beta^\star_2\tau_1/\beta^\star_1}{\beta^\star_1 Q_\alpha^{\eta^\star}+\tau_1}}_{=:\tau_3(x_1,x_2,\alpha)}.
\end{equation}

On the other hand, taking the difference of \eqref{eq:q_med} for $x=x_1$ and $x_2$ yields
\begin{equation}\label{eq:q_med_x1_x2}
	\beta_1(x_1-x_2) + \beta_2(x_1^2-x_2^2) = \beta^\star_1(x_1-x_2) + \beta^\star_2(x_1^2-x_2^2) + \delta_{0.5}(x_1) - \delta_{0.5}(x_2).
\end{equation}
Plugging \eqref{eq:beta_ratio} into \eqref{eq:q_med_x1_x2} implies
\begin{equation}\label{eq:bound_beta1}
	\beta_1 - \beta^\star_1 = \frac{-\beta_1\tau_3(x_1,x_2,\alpha)(x_1+x_2)+(\delta_{0.5}(x_1) - \delta_{0.5}(x_2))/(x_1-x_2)}{1+\beta^\star_2(x_1+x_2)/\beta^\star_1}.
\end{equation}
From \eqref{eq:q_med_x1_x2} and \eqref{eq:q_med} for $x=x_1$, we know
\begin{align}
	\beta_2 - \beta^\star_2 &= \frac{\beta_1 - \beta^\star_1}{x_1+x_2} + \frac{\delta_{0.5}(x_1) - \delta_{0.5}(x_2)}{x_1^2-x_2^2} \label{eq:bound_beta2}\\
	\beta_0 - \beta^\star_0 &= x_1(\beta_1 - \beta^\star_1) + x_1^2(\beta_2 - \beta^\star_2) + \delta_{0.5}(x_1)\label{eq:bound_beta0}.
\end{align}
Combining \eqref{eq:bound_beta1}, \eqref{eq:bound_beta2} and \eqref{eq:bound_beta0} leads to a bound for the parameter difference
\begin{equation}\label{eq:bound_param_diff}
\begin{split}
	\|\beta-\beta^\star\| &\le |\beta_0 - \beta^\star_0| + |\beta_1 - \beta^\star_1| + |\beta_2 - \beta^\star_2| \\
	&\le \left(x_1+1+\frac{x_2^2+1}{x_1+x_2}\right)|\beta_1 - \beta^\star_1| + \frac{x_2^2+1}{x_2^2-x_1^2}|\delta_{0.5}(x_1)-\delta_{0.5}(x_2)| + |\delta_{0.5}(x_1)| \\
	&\le \left(x_1+1+\frac{x_2^2+1}{x_1+x_2}\right)\frac{\beta_1(x_1+x_2)|\tau_3(x_1,x_2,\alpha)|}{1+\beta^\star_2(x_1+x_2)/\beta^\star_1} + c_1\big(|\delta_{0.5}(x_1)|+|\delta_{0.5}(x_2)|\big)
\end{split}
\end{equation}
where $c_1<\infty$ is a constant depending on $x_1,x_2,\beta^\star_1$ and $\beta^\star_2$.

Furthermore, \eqref{eq:beta_1_q} implies
\begin{equation}\label{eq:quantile_diff}
	Q^\eta_\alpha - Q^{\eta^\star}_\alpha = \frac{(\beta_1^\star-\beta_1)Q^{\eta^\star}_\alpha+\tau_1}{\beta_1},
\end{equation}
which holds for all $\alpha\in[0,1]$ by symmetry and leads to a bound for the difference in the noise quantiles by combining with \eqref{eq:bound_beta1}.

\bigskip
\noindent{\textit{Step III. }} 
Let $\hat{q}_\alpha(x)$ be the $\alpha$-quantile of $\hat\beta_0+\hat\beta_1(x+\hat\eta)+\hat\beta_2(x+\hat\eta)^2$, which is the engression estimator for the conditional quantile. 
By Lemma~\ref{lem:bound_quantile_by_cramer}, we have for $x=x_1,x_2$
\begin{equation*}
	\sup_{\alpha\in[0,1]}|\hat{q}_\alpha(x)-q^\star_\alpha(x)| \le \left(\frac{3\cramer(F_{\hat\beta,\hat\eta|x},F^\star_x)}{b^2}\right)^{\frac{1}{3}}.
\end{equation*}
Applying the bound in \eqref{eq:bound_cramer} then yields
\begin{equation}\label{eq:bound_quantile_est_pf}
	\sup_{\alpha\in[0,1]}|\hat{q}_\alpha(x)-q^\star_\alpha(x)| \le \left(\frac{6\nu(\cY)\log(2/\delta)}{b^2n}\right)^{\frac{1}{3}},
\end{equation}
for $x=x_1,x_2$.

According to the results in Step II, by taking $\beta=\hat\beta$ and $\eta=\hat\eta$, we have from the above bounds that $\delta_\alpha(x)=\cO\big((\frac{\log n}{n})^{\frac{1}{3}}\big)$ uniformly for all $x=x_1,x_2$ and $\alpha\in[0,1]$. Then both $\tau_1$ and $\tau_2$ are of the same order, which implies that for all $\alpha$, with probability at least $1-\delta$, 
\begin{equation*}
	\tau_3(x_1,x_2,\alpha) \le \frac{c_2}{(x_2-x_1)Q_\alpha^{\eta^\star}}\left(\frac{\log(2/\delta)}{b^2n}\right)^{\frac{1}{3}},
\end{equation*}
where $c_2\propto\nu(\cY)^{\frac{1}{3}}$ is a constant depending on the size of the support $\cY$, $\beta^\star$, $\cX$. 
Taking the infimum over $\alpha$ yields
\begin{equation*}
	\inf_{\alpha\in[0,1]}\tau_3(x_1,x_2,\alpha) \le \frac{c_2}{(x_2-x_1)\etam}\left(\frac{\log(2/\delta)}{b^2n}\right)^{\frac{1}{3}}.
\end{equation*}

Therefore, we conclude from \eqref{eq:bound_param_diff} that
\begin{equation*}
	\|\hat\beta - \beta^\star\| \le \frac{C_1}{x_2-x_1}\left(\frac{\log(2/\delta)}{b^2n}\right)^{\frac{1}{3}}.
\end{equation*}

\bigskip
\noindent{\textit{Step IV. }} 
Based on \eqref{eq:quantile_diff} and the above bound for parameter estimation, we have for $\alpha\in[0,1]$
\begin{equation}\label{eq:quantile_bound_n}
	|Q^{\hat\eta}_\alpha - Q^{\eta^\star}_\alpha| \le c_3|Q^{\eta^\star}_\alpha|\left(\frac{\log(2/\delta)}{n}\right)^{\frac{1}{3}}
\end{equation}
for a constant $c_3<\infty$. Then for all $x\in\bbR$, we have
\begin{equation*}
	\hat{q}_\alpha(x) - q^\star_\alpha(x) = \hat\beta_0 - \beta^\star_0 + x(\hat\beta_1-\beta^\star_1) + x^2(\hat\beta_2-\beta^\star_2) + \hat\beta_1Q^{\hat\eta}_\alpha-\beta^\star_1Q^{\eta^\star}_\alpha + 2x(\hat\beta_2Q^{\hat\eta}_\alpha-\beta^\star_2Q^{\eta^\star}_\alpha) + \hat\beta_2{Q^{\hat\eta}_\alpha}^2-\beta^\star_2{Q^{\eta^\star}_\alpha}^2.
\end{equation*}
Note that
\begin{equation*}
	|\hat\beta_0 - \beta^\star_0 + x(\hat\beta_1-\beta^\star_1) + x^2(\hat\beta_2-\beta^\star_2)| \lesssim \max\{1,x,x^2\}\left(\frac{\log(2/\delta)}{n}\right)^{\frac{1}{3}},
\end{equation*}
\begin{equation*}
\begin{split}
	|\hat\beta_1Q^{\hat\eta}_\alpha-\beta^\star_1Q^{\eta^\star}_\alpha| &= |(\hat\beta_1-\beta^\star_1)Q^{\eta^\star}_\alpha + \hat\beta_1(Q^{\hat\eta}_\alpha-Q^{\eta^\star}_\alpha)| \lesssim |Q^{\eta^\star}_\alpha|\left(\frac{\log(2/\delta)}{n}\right)^{\frac{1}{3}},
\end{split}
\end{equation*}
\begin{equation*}
	|x(\hat\beta_2Q^{\hat\eta}_\alpha-\beta^\star_2Q^{\eta^\star}_\alpha)| \lesssim |xQ^{\eta^\star}_\alpha|\left(\frac{\log(2/\delta)}{n}\right)^{\frac{1}{3}},
\end{equation*}
and
\begin{equation*}
	|\hat\beta_2{Q^{\hat\eta}_\alpha}^2-\beta^\star_2{Q^{\eta^\star}_\alpha}^2| \lesssim |Q^{\eta^\star}_\alpha|^2\left(\frac{\log(2/\delta)}{n}\right)^{\frac{1}{3}}.
\end{equation*}
Thus, we conclude the bound
\begin{equation*}
	|\hat{q}_\alpha(x) - q^\star_\alpha(x)| \le C_3\max\{1,x,x^2\} |Q^{\eta^\star}_\alpha|\left(\frac{\log(2/\delta)}{n}\right)^{\frac{1}{3}},
\end{equation*}
for a constant $C_3<\infty$ that does not depend on $x$ or $\alpha$.

\bigskip
\noindent{\textit{Step V. }} 
As $\eta^\star\le\etam$, taking the supremum of \eqref{eq:quantile_bound_n} yields
\begin{equation*}
	\sup_{\alpha\in[0,1]}|Q^{\hat\eta}_\alpha - Q^{\eta^\star}_\alpha| \lesssim \left(\frac{\log(2/\delta)}{n}\right)^{\frac{1}{3}}.
\end{equation*}
As $\hat\eta$ and $\eta^\star$ have mean 0, we have
\begin{equation*}
	\hat\mu(x) - \mu^\star(x) = \underbrace{\hat\beta_0-\beta^\star_0 + x(\hat\beta_1-\beta^\star_1) + (x^2+\bbE[{\eta^\star}^2])(\hat\beta_2-\beta^\star_2)}_{\lesssim \max\{1,x,x^2\}\left(\frac{\log(2/\delta)}{n}\right)^{\frac{1}{3}}} + \hat\beta_2(\bbE[\hat\eta^2]-\bbE[{\eta^\star}^2]).
\end{equation*}
Note that
\begin{align*}
	|\bbE[\hat\eta^2]-\bbE[{\eta^\star}^2]| &= \left|\int_0^1(Q^{\hat\eta}_\alpha)^2 d\alpha - \int_0^1(Q^{\eta^\star}_\alpha)^2 d\alpha\right| \\
	&\le \int_0^1|(Q^{\hat\eta}_\alpha)^2 - (Q^{\eta^\star}_\alpha)^2| d\alpha\\
	&\le 2\etam \sup_{\alpha\in[0,1]}|Q^{\hat\eta}_\alpha - Q^{\eta^\star}_\alpha|\\
	&\lesssim \left(\frac{\log(2/\delta)}{n}\right)^{\frac{1}{3}}.
\end{align*}
Therefore, we conclude the bound
\begin{equation*}
	|\hat\mu(x) - \mu^\star(x)| \le C_2\max\{1,x,x^2\}\left(\frac{\log(2/\delta)}{n}\right)^{\frac{1}{3}},
\end{equation*}
where $C_2<\infty$ is a constant that does not depend on $x$.
\end{proof}

\subsection{Proof of Proposition~\ref{prop:finite_sample_eng_mis}}\label{app:pf_thm_finite_sample_eng_mis}
\begin{proof}[Proof of Proposition~\ref{prop:finite_sample_eng_mis}]
The proof proceeds similarly to the first three steps in the proof of Theorem~\ref{thm:finite_sample_eng}.

Recall that the true data generating model is $Y=\beta^\star_0+\beta^\star_1 x + \beta^\star_2 x^2 + \eta^\star$ and $X$ takes only two values during training. Let $F^\star_x$ be the conditional cdf of $Y$ given $X=x$ induced by the true model. Let $F_{\beta,\eta|x}$ be the conditional cdf of $Y$ given $X=x$ induced by a pre-ANM in the class \eqref{eq:quadratic_preanm}. Note that the pre-ANM in class \eqref{eq:quadratic_preanm} with $\tilde\beta_0=\beta_0^\star-\beta_2^\star x_1x_2$, $\tilde\beta_1=\beta_1^\star+\beta_2^\star(x_1+x_2)$, $\tilde\beta_2=0$, and $\tilde\eta=\eta^\star/\tilde\beta_1$ induces exactly the same conditional distribution of $Y|X=x$ as the true model, for $x=x_1,x_2$. That is, $F_{\tilde\beta,\tilde\eta|x}=F^\star_x$ for $x=x_1,x_2$. Thus, we have for the engression estimator $\hat\beta,\hat\eta$ that
\begin{equation*}
	\sum_{x\in\cX}\cramer(F_{\hat\beta,\hat\eta|x},\hat{F}_{n|x}) = \min_{\beta,\eta}\sum_{x\in\cX}\cramer(F_{\beta,\eta|x},\hat{F}_{n|x}) \le \sum_{x\in\cX}\cramer(F^\star_x,\hat{F}_{n|x}).
\end{equation*}
Therefore, similarly to \eqref{eq:bound_cramer} in the proof of Theorem~\ref{thm:finite_sample_eng}, we have with probability exceeding $1-\delta$ that
\begin{equation*}\label{eq:bound_cramer_mis}
	\sum_{x\in\cX}\cramer(F_{\hat\beta,\hat\eta|x},F^\star_{x}) \lesssim \frac{\log(2/\delta)}{n}.
\end{equation*}

Let $\hat{q}_\alpha(x)$ be the $\alpha$-quantile of $\hat\beta_0+\hat\beta_1(x+\hat\eta)+\hat\beta_2(x+\hat\eta)^2$, and $q^\star_\alpha(x)$ be the $\alpha$-quantile of $\beta^\star_0+\beta^\star_1 (x+\eta)+\beta^\star_2(x+\eta)^2$. Let $\hat\delta_\alpha(x)=\hat{q}_\alpha(x)-q^\star_\alpha(x)$ which depends on $n$.
Then by Lemma~\ref{lem:bound_quantile_by_cramer} and the above bound, we have for $x=x_1,x_2$ that
\begin{equation}\label{eq:bound_quantile_est_pf_mis}
	\sup_{\alpha\in[0,1]}|\hat\delta_\alpha(x)| \lesssim \left(\frac{\log(2/\delta)}{n}\right)^{\frac{1}{3}}.
\end{equation}

For any $\alpha\in(0.5,1]$, it holds that
\begin{align}
	\hat\beta_0+\hat\beta_1 x +\hat\beta_2 x^2+(\hat\beta_1+2\hat\beta_2x)Q_\alpha^{\hat\eta}+\hat\beta_2{Q_\alpha^{\hat\eta}}^2 &= \beta^\star_0+\beta^\star_1 x +\beta^\star_2 x^2+Q_\alpha^{\eta^\star} + \hat\delta_{\alpha}(x) \label{eq:q_alpha_mis}\\
	\hat\beta_0+\hat\beta_1 x +\hat\beta_2 x^2-(\hat\beta_1+2\hat\beta_2x)Q_\alpha^{\hat\eta}+\hat\beta_2{Q_\alpha^{\hat\eta}}^2 &= \beta^\star_0+\beta^\star_1 x +\beta^\star_2 x^2-Q_\alpha^{\eta^\star} + \hat\delta_{1-\alpha}(x) \label{eq:q_1_alpha_mis}\\
	\beta_0+\beta_1 x +\beta_2 x^2 &= \beta^\star_0+\beta^\star_1 x +\beta^\star_2 x^2 + \hat\delta_{0.5}(x)\label{eq:q_med_mis}
\end{align}
for $x=x_1,x_2$. 

Taking the difference between \eqref{eq:q_alpha_mis} and \eqref{eq:q_1_alpha_mis} yields
\begin{equation*}
	2(\hat\beta_1+2\hat\beta_2x)Q_\alpha^{\hat\eta} = 2Q_\alpha^{\eta^\star} + \hat\delta_\alpha(x) - \hat\delta_{1-\alpha}(x).
\end{equation*}
Taking the difference of the above equation for $x=x_1$ and $x_2$ implies
\begin{equation}\label{eq:beta_2_q_mis}
	\hat\beta_2Q_\alpha^\eta = \frac{\hat\delta_\alpha(x_2) - \hat\delta_\alpha(x_1) - \hat\delta_{1-\alpha}(x_2) +\hat \delta_{1-\alpha}(x_1)}{4(x_2-x_1)},
\end{equation}
thus $|\hat\beta_2|\lesssim\left(\frac{\log(2/\delta)}{n}\right)^{\frac{1}{3}}.$

On the other hand, taking the difference of \eqref{eq:q_med_mis} for $x=x_1$ and $x_2$ yields
\begin{equation}\label{eq:q_med_x1_x2_mis}
	\hat\beta_1(x_1-x_2) + \hat\beta_2(x_1^2-x_2^2) = \beta^\star_1(x_1-x_2) + \beta^\star_2(x_1^2-x_2^2) + \hat\delta_{0.5}(x_1) - \hat\delta_{0.5}(x_2).
\end{equation}
Plugging \eqref{eq:beta_2_q_mis} into \eqref{eq:q_med_x1_x2_mis} implies
\begin{equation*}
	|\beta_1 - \beta^\star_1 - \beta_2^\star(x_1+x_2)| \lesssim  \left(\frac{\log(2/\delta)}{n}\right)^{\frac{1}{3}}.
\end{equation*}
Then from \eqref{eq:q_med_mis} for $x=x_1$, we know
\begin{equation*}
	|\hat\beta_0 - \beta^\star_0 + \beta^\star_2 x_1 x_2| \lesssim  \left(\frac{\log(2/\delta)}{n}\right)^{\frac{1}{3}},
\end{equation*}
which concludes the proof.
\end{proof}

\subsection{Proof of Theorem~\ref{thm:consistency_general}}\label{app:thm_consistency_general}
\begin{proof}[Proof of Theorem~\ref{thm:consistency_general}]
Let $F_{g,h|x}$ be the conditional cdf of $Y$ given $X=x$ induced by a pre-ANM model indexed by $(f,g)$. Let $F^\star_x=F_{g^\star,h^\star|x}$ be the true cdf and $\hat{F}_{n_i|x_i}$ be the empirical distribution function of $Y$ given $X=x_i$, where $n_i$ be the number of $Y$ samples for each $X=x_i$, for $i=1,\dots,m$. Then the empirical engression is equivalent to
\begin{equation*}
	\min_{g,h}\sum_{i=1}^m \cramer(F_{g,h|x_i}, \hat{F}_{n_i|x_i}).
\end{equation*}

Then similarly to \eqref{eq:bound_cramer} in the proof of Theorem~\ref{thm:finite_sample_eng}, we have with probability exceeding $1-\delta$ that
\begin{equation*}
	\sum_{i=1}^m\cramer(F_{\hat{g},\hat{h}|x_i},F^\star_{x_i}) \lesssim \frac{\log(\tilde{n}/\delta)}{\tilde{n}},
\end{equation*}
where $\tilde{n}=\min\{n_1,\dots,n_m\}\ge cn/m\to\infty$ as $n\to\infty$ by Condition~\ref{ass:grid}. 

	Let $\hat{q}_\alpha(x)$ and $q^\star_\alpha(x)$ be the quantiles of $\hat{g}(x+\hat{h}(\varepsilon))$ and $g^\star(x+h^\star(\varepsilon))$, respectively. Let $\hat\delta_\alpha(x)=\hat{q}_\alpha(x)-q^\star_\alpha(x)$.  Then by Lemma~\ref{lem:bound_quantile_by_cramer} and the above bound, we have for $x\in\{x_1,\dots,x_m\}$ that
\begin{equation}\label{eq:bound_quantile_est_pf_general}
	\sup_{\alpha\in[0,1]}|\hat\delta_\alpha(x)| \lesssim \left(\frac{\log(\tilde{n}/\delta)}{\tilde{n}}\right)^{\frac{1}{3}}.
\end{equation}

Let $d_m=\max_{j=1,\dots,m}|x_j-x_{j-1}|=o(m)$ by Condition~\ref{ass:grid}. For any $x\in\cX$, there exists $i\in\{1,\dots,m\}$ such that $|x-x_i|\le d_m$. By the Lipschitz condition in Condition~\ref{ass:g_lip_mono} with the Lipschitz constant denoted by $L$, we know the quantile functions are also uniformly $L$-Lipschitz. Thus, we have for all $x\in\cX$
\begin{align*}
	\sup_{\alpha\in[0,1]}|\hat\delta_\alpha(x)| &\le \sup_{\alpha\in[0,1]}|\hat{q}_\alpha(x) - \hat{q}_\alpha(x_i)| + \sup_{\alpha\in[0,1]}|q^\star_\alpha(x) - q^\star_\alpha(x_i)| + \sup_{\alpha\in[0,1]}|\hat{q}_\alpha(x_i) - q^\star_\alpha(x_i)| \\
	&\lesssim L|d_m| + \left(\frac{\log(\tilde{n}/\delta)}{\tilde{n}}\right)^{\frac{1}{3}}.
\end{align*}
	
	Then for all $\alpha\in[0,1]$ and all $x\in\cX$, we have
	\begin{equation*}
		\hat{g}(x+\hat{h}(Q^\varepsilon_\alpha)) = g^\star(x+h^\star(Q^\varepsilon_\alpha)) + \hat\delta_\alpha(x).
	\end{equation*}
	As $\varepsilon\sim\mathrm{Unif}[0,1]$, the above indicates that for all $x\in\cX$ and any fixed number $\varepsilon\in[0,1]$, it holds that
	\begin{equation}\label{eq:equal_quantile}
		\hat{g}(x+\hat{h}(\varepsilon)) = g^\star(x+h^\star(\varepsilon)) + \hat\delta_\varepsilon(x).
	\end{equation}
	Let $y=g^\star(x+h^\star(\varepsilon))$. We have for all $\varepsilon\in[0,1]$ and $y\in\{g^\star(x+h^\star(\varepsilon)):x\in\cX\}$ that
	\begin{equation}\label{eq:equal_x}
		{g^{\star}}^{-1}(y) - h^\star(\varepsilon) = \hat{g}^{-1}(y - \hat\delta_\varepsilon(y)) - \hat{h}(\varepsilon),
	\end{equation}
	where we denote $\hat\delta_\varepsilon(y):=\hat\delta_\varepsilon({g^{\star}}^{-1}(y) - h^\star(\varepsilon))$ for ease of notation.
	
	Taking $\varepsilon=0$, we have for all $y\in g^\star(\cX)$ that
	\begin{equation*}
		{g^{\star}}^{-1}(y) = \hat{g}^{-1}(y-\hat\delta_{0.5}(y)).
	\end{equation*}
	
	For all $\varepsilon\in[0,1]$, according to Condition~\ref{ass:large_x_space}, there exists $x\in\cX$ such that $x+h^\star(\varepsilon)\in\cX$. Let $y=g^\star(x+h^\star(\varepsilon))\in g^\star(\cX)$. Then we have ${g^{\star}}^{-1}(y) = \hat{g}^{-1}(y-\hat\delta_{0.5}(y))$, plugging which into \eqref{eq:equal_x} leads to
	\begin{equation*}
		\hat{g}^{-1}(y-\hat\delta_{0.5}(y)) - h^\star(\varepsilon)=\hat{g}^{-1}(y - \hat\delta_\varepsilon(y)) - \hat{h}(\varepsilon)
	\end{equation*}
	By the uniform Lipschitz condition in Condition~\ref{ass:g_lip_mono}, the above equality implies for all $\varepsilon\in[0,1]$
	\begin{equation}\label{eq:bound_h}
		|\hat{h}(\varepsilon) - h^\star(\varepsilon)| \lesssim |\hat\delta_\varepsilon(y)| + |\hat\delta_{0.5}(y)|.
	\end{equation}
	
	On the other hand, note from \eqref{eq:equal_quantile} and Taylor expansion that
	\begin{equation*}
		\hat{g}(x+h^\star(\varepsilon)+\hat{h}(\varepsilon) - h^\star(\varepsilon)) = g^\star(x+h^\star(\varepsilon))+\hat\delta_\varepsilon(x) = \hat{g}(x+h^\star(\varepsilon))+\hat{g}'(x')(\hat{h}(\varepsilon) - h^\star(\varepsilon)),
	\end{equation*}
	where $x'=x+h^\star(\varepsilon)+t(\hat{h}(\varepsilon) - h^\star(\varepsilon))$ for some $t\in[0,1]$.
	This implies for all $x\in\cX$ and $\varepsilon\in[0,1]$
	\begin{equation}\label{eq:bound_g}
		|\hat{g}(x+h^\star(\varepsilon)) - g^\star(x+h^\star(\varepsilon))| \lesssim |\hat{h}(\varepsilon) - h^\star(\varepsilon)|+|\hat\delta_\varepsilon(x)|.
	\end{equation}
	
	By letting $n\to\infty$ in \eqref{eq:bound_h} and \eqref{eq:bound_g} and noting the uniform bound in \eqref{eq:bound_quantile_est_pf_general}, we arrive at the desired consistency results.
\end{proof}

\section{Extrapolability of quadratic pre-post-additive noise models}\label{app:pre_post_anm}
Consider the following quadratic model that allows for both pre- and post-additive noises, called a pre-post-additive noise model
\begin{equation*}
	Y = g_\theta(X+\eta)+\beta X + \xi,
\end{equation*}
where $\theta=(\theta_0,\theta_1,\theta_2)\in\bbR^3, \beta\in\bbR$ we take a quadratic function $g_\theta(x)=\theta_0+\theta_1x+\theta_2x^2$, and $\eta$ and $\xi$ have symmetric distributions and unbounded supports.

The following proposition shows the distributional extrapolability of the quadratic pre-post-ANM model class.
\begin{proposition}\label{prop:extra_pre_post_anm}
	The class of quadratic pre-post-ANM models $$\{g_\theta(x+\eta)+\beta x + \xi:\theta\in\bbR^3, \beta\in\bbR, \text{$\eta$ and $\xi$ have symmetric distributions and unbounded supports}\}$$ is distributionally extrapable. 
\end{proposition}
\begin{proof}[Proof of Proposition~\ref{prop:extra_pre_post_anm}]
By Definition~\ref{def:extra_dist}, it suffices to show for any $\theta,\beta,\eta,\xi,\theta',\beta',\eta',\xi'$ such that 
	\begin{equation}\label{eq:quad_pre_post_anm_dist_equal}
		g_\theta(x+\eta)+\beta x+\xi\overset{d}=g_{\theta'}(x+\eta')+\beta' x+\xi'
	\end{equation}
	for $x\in\cX$, it also holds that $g_\theta(x+\eta)+\beta x+\xi\overset{d}=g_{\theta'}(x+\eta')+\beta' x+\xi'$ for all $x\in\bbR$. 

In particular for the quadratic model class, it further suffices to show for model components that satisfy \eqref{eq:quad_pre_post_anm_dist_equal}, it holds that
\begin{equation*}
	\theta=\theta',\ \beta=\beta',\ \eta\overset{d}=\eta',\ \xi\overset{d}=\xi'.
\end{equation*}

Note that
\begin{equation*}
	g_\theta(x+\eta)+\beta x+\xi = \theta_0+(\theta_1+\beta)x + \theta_2 x^2 + \theta_1\eta + 2\theta_2 x\eta + \theta_2\eta^2+\xi.
\end{equation*}

Taking the mean of both sides of \eqref{eq:quad_pre_post_anm_dist_equal} yields
\begin{equation*}
	\theta_0+(\theta_1+\beta)x + \theta_2 x^2 +\theta_2\bbE[\eta^2] = \theta'_0+(\theta'_1+\beta')x + \theta'_2 x^2 +\theta'_2\bbE[\eta^2]
\end{equation*}
for all $x\in\cX$, thus indicating $\theta_1+\beta = \theta_1'+\beta'$ and $\theta_2 = \theta_2'$. 

Taking the variance of both sides of \eqref{eq:quad_pre_post_anm_dist_equal} yields
\begin{equation*}
	(\theta_1+2\theta_2x)^2\var(\eta)+\theta_2^2\var(\eta^2)+\var(\xi) = (\theta'_1+2\theta'_2x)^2\var(\eta')+{\theta_2'}^2\var({\eta'}^2)+\var(\xi')
\end{equation*}
for all $x\in\cX$, thus indicating $\theta_1=\theta_1'$ and $\var(\eta)=\var(\eta')$, which according to the above relations, further implies $\beta=\beta'$ and $\theta_0=\theta_0'$. 

Now, consider the characteristic function of both sides of \eqref{eq:quad_pre_post_anm_dist_equal}; we have for all $t\in\bbR$
\begin{equation*}
	\bbE[e^{it(\theta_1+2\theta_2x)\eta}]\bbE[e^{it\theta_2\eta^2}]\bbE[e^{it\xi}] = \bbE[e^{it(\theta_1+2\theta_2x)\eta'}]\bbE[e^{it\theta_2{\eta'}^2}]\bbE[e^{it\xi'}]
\end{equation*}
holds for all $x\in\cX$. 
Then we know
\begin{equation*}
	\bbE[e^{it2\theta_2(x-x')\eta}] = \bbE[e^{it2\theta_2(x-x')\eta'}]
\end{equation*}
for all $x,x'\in\cX$. Thus, we have $\eta\overset{d}=\eta'$ and also $\xi\overset{d}=\xi'$, which concludes the proof.
\end{proof}

\section{Engression with other loss functions}\label{app:losses}
As noted in Section~\ref{sec:engression}, the choices of the loss function in engression can be rather flexible and the theoretical results in Section~\ref{sec:theory_extra} always hold. In this section, we discuss some alternative choices, including the KL divergence and the kernel score~\citep{dawid2007geometry} or the MMD distance~\citep{gretton2012kernel,sejdinovic2013equivalence} with several common kernels, and investigate their empirical performance in simulated settings.

\subsection{Kernel/MMD loss}
A bivariate function $k:\bbR\times\bbR\to\bbR$ is said to be a positive definite kernel if it is symmetric in its arguments and $\sum_{i=1}^n\sum_{j=1}^na_ia_jk(z_i,z_j)\ge0$ for all $n\in\mathbb{N}_+$, $\sum_{i=1}^na_i=0$, and all $z_i\in\bbR$.  
Given a distribution $P$, an observation $z$, and a kernel $k$, the kernel score is defined as
\begin{equation*}
	S_k(P,z) = \bbE_{Z\sim P}[k(Z,z)]-\frac{1}{2}\bbE_{Z,Z'}[k(Z,Z')],
\end{equation*}
where $Z,Z'$ are two i.i.d.\ draws from $P$ in the second term. If $k$ is a characteristic kernel, then $S_k$ is a strictly proper scoring rule \citep{steinwart2021strictly}. The associated distance function for two distributions $P$ and $P'$ is the squared MMD distance~\citep{gretton2012kernel}
\begin{equation*}
	\mathrm{MMD}^2_k(P,P')=\bbE[k(Z,\tilde{Z})]-2\bbE[k(Z,Z')]+\bbE[k(Z',\tilde{Z}')],
\end{equation*}
where $Z$ and $\tilde{Z}$ are two independent draws from $P$, and $Z'$ and $\tilde{Z}'$ are independent draws from $P'$. 

The energy score is a special case of the kernel score with $k(z,z')=-\|z-z'\|$. Other commonly used characteristic kernels include the Gaussian kernel $k(z,z')=\exp(-\|z-z'\|^2/2\sigma)$, the Laplace kernel $k(z,z')=\exp(-\|z-z'\|/\sigma)$, and the inverse multiquadric kernel $k(z,z')=1/\sqrt{\|z-z'\|^2+c}$. In our experiments, we use the default hyperparameter values of $\sigma=1$ for Gaussian and Laplace kernels and $c=1$ for the inverse multiquadric kernel.

Now, define the loss function based on the (negative) kernel score as 
\begin{align*}
	\cL_k(P_\gen(y|x);\ptr(y|x)) =\ &\bbE_{Y\sim \ptr(y|x)}[-S_k(P_\gen(y|x),Y)]\\
	 =\ &\bbE_{Y\sim \ptr(y|x)}[-\bbE_{\varepsilon}[k(\gen(x,\varepsilon),Y)]+\frac{1}{2}\bbE_{\varepsilon,\varepsilon'}[k(\gen(x,\varepsilon),\gen(x,\varepsilon'))]],
\end{align*}
where $\varepsilon,\varepsilon'$ are two i.i.d.\ draws from $\mathrm{Unif}[0,1]$. 
We know if $k$ is a characteristic kernel, the above kernel loss satisfies the characterisation property and hence can be used as the loss function for engression. 
The population engression with the kernel loss is 
\begin{equation*}
	\min_{\gen\in\preanm}\cL_k(P_\gen(y|x);\ptr(y|x)).
\end{equation*}

\subsection{KL loss}
For two probability densities $p,p'$, the KL divergence between $p,p'$ is defined as
\begin{equation*}
	\kl(p,p')=\int p(z)\log(p(z)/p'(z))dz.
\end{equation*}
We have $\kl(p,p')\geq0$ where the equality holds if and only if $p=p'$.

Let $p_{\mathrm{tr}}(y|x)$ and $p_\gen(y|x)$ be the conditional densities of the training and model distributions, respectively. Define the loss function based on the KL divergence as
\begin{align*}
	\cL_{\kl}(P_\gen(y|x);\ptr(y|x)) =\ &\kl(p_{\mathrm{tr}}(y|x),p_\gen(y|x))\\
	=\ &-\bbE_{p_{\mathrm{tr}}(x,y)}[\log(p_\gen(Y|X))]+\bbE_{p_{\mathrm{tr}}(x,y)}[\log(p_{\mathrm{tr}}(y|x))]
\end{align*}
where the second term is a constant, so minimising the KL loss is equivalent to the maximum likelihood estimation. The population version of engression with the KL loss is defined as
\begin{equation}\label{eq:kl_loss}
	\min_{\gen\in\preanm}\cL_\kl(p_\gen(y|x);p_{\mathrm{tr}}(y|x)).
\end{equation} 

However, when $\cG$ and $\cH$ are parameterised by neural networks, the KL loss does not have a closed form to be easily optimised. To this end, we adapt a generative adversarial network (GAN) algorithm from \citet{shen2022asymptotic} for solving \eqref{eq:kl_loss}. Let $\gen_\theta(x,\varepsilon)=g_{\theta_1}(Wx+h_{\theta_2}(\varepsilon))+\beta x$, also known as the generator in the generative models literature, where both $g$ and $h$ are parameterised by neural networks and $\theta=(\theta_1,\theta_2,\beta,W)$ collects all the parameters. We summarise the procedure of the GAN-based engression in Algorithm~\ref{alg:gan}, where in line 4, we solve an empirical nonlinear logistic regression among a class $\cD$ of discriminators; in practice, we use early stopping and only take a few gradient descent steps, as in \citet{goodfellow2014}. Under suitable conditions, the algorithm yields a consistent estimator for solving \eqref{eq:kl_loss} \citep{shen2022asymptotic}.

\begin{algorithm}
\DontPrintSemicolon
\KwInput{Sample $\{(x_i,y_i):i=1,\dots,n\}$, initial parameter $\theta_0$, learning rate $\eta$, number of iterations $T$}
\For{$t=0,1,2,\dots,T$}{
Sample $\varepsilon_i\sim\mathrm{Unif}[0,1]$, $i=1,\dots,n$\\
Generate $\hat{y}_i=\gen_{\theta_{t}}(x_i,\varepsilon_i)$ for $i=1,\dots,n$\\
${D}_t=\argmin_{D\in\cD}\big[\frac{1}{n}\sum_{i=1}^{n}\log(1+e^{-D(y_i)})+\frac{1}{n}\sum_{i=1}^{n}\log(1+e^{D(\hat{y}_i)})\big]$\\
$\theta_{t+1}=\theta_{t}+\frac{\eta}{n}\sum_{i=1}^{n}\exp(-{D}_t(\hat{y}_i))\nabla_\theta \gen_{\theta}(x_i,\epsilon_i)^\top|_{\theta=\theta_t} \nabla{{D}}_t(\hat{y}_i)$ 
}
\KwReturn{$\theta_T$}
\caption{GAN-based engression}
\label{alg:gan}
\end{algorithm}

\subsection{Empirical results}
In Figures~\ref{fig:simu_visual_app_softplus}-\ref{fig:simu_visual_app_log}, we show the results of engression with various loss functions for the simulation settings in Table~\ref{tab:simu_set}, Section~\ref{sec:simu}. Compared with engression with the energy loss adopted in the main text, in particular the results in Figure~\ref{fig:simu_visual}, we see that engression with other loss functions in general exhibit comparable performance in contrast to the significantly higher extrapolation uncertainty of baseline approaches. Indeed, we also observe some differences in terms of extrapolation performance. For example, GAN-based engression with the KL loss performs better than the other loss functions on \texttt{softplus} and \texttt{square}, while the Laplace kernel loss appears to be the most superior on \texttt{log}. These empirical observations suggest that the theoretical investigation of various loss functions and their impact on engression are worthy subjects for future research.

\begin{figure}
\centering
\begin{tabular}{@{}c@{}c@{}c@{}c@{}c@{}}
	&\small{Gaussian kernel} & \small{Laplace kernel} & \small{Inverse multiquadric kernel} & \small{KL GAN}\\
	\rotatebox[origin=c]{90}{\small{Conditional median}}&
	\includegraphics[align=c,width=0.24\textwidth]{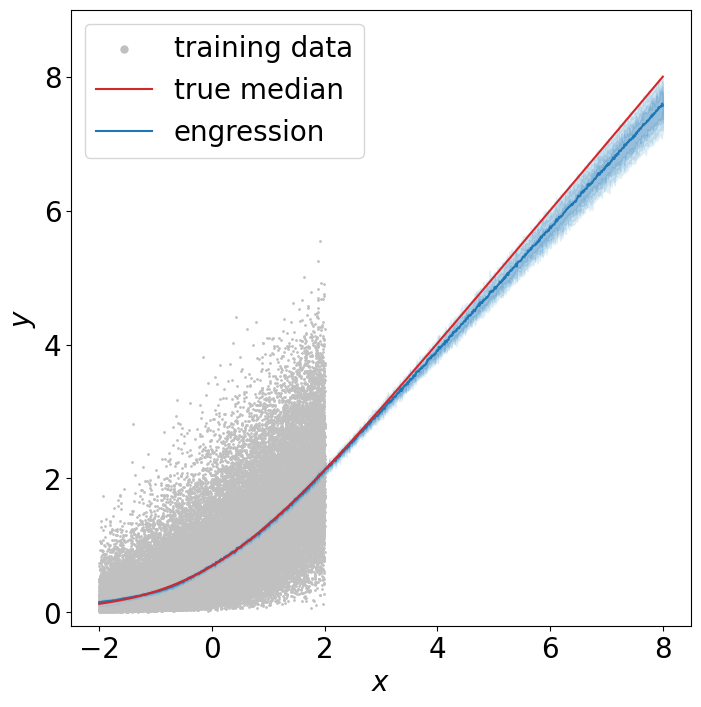} &
	\includegraphics[align=c,width=0.24\textwidth]{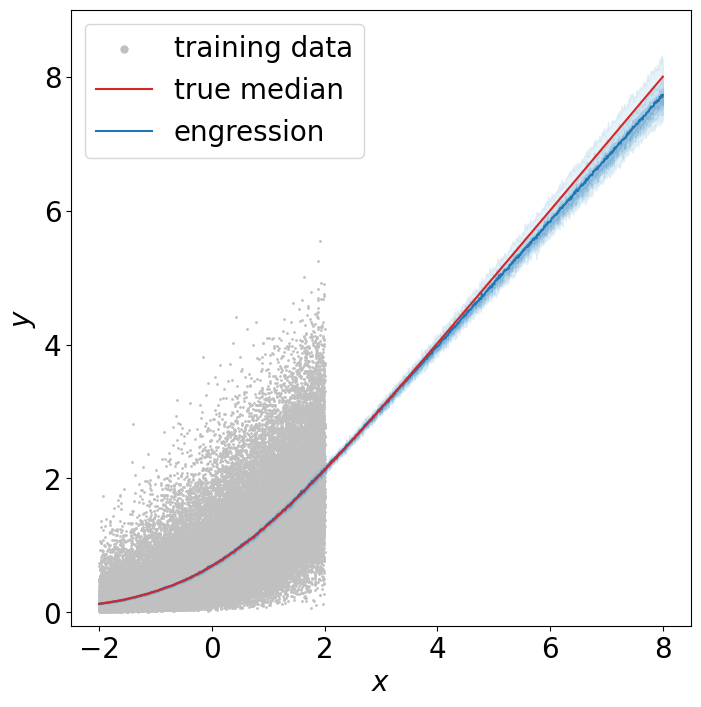} &
	\includegraphics[align=c,width=0.24\textwidth]{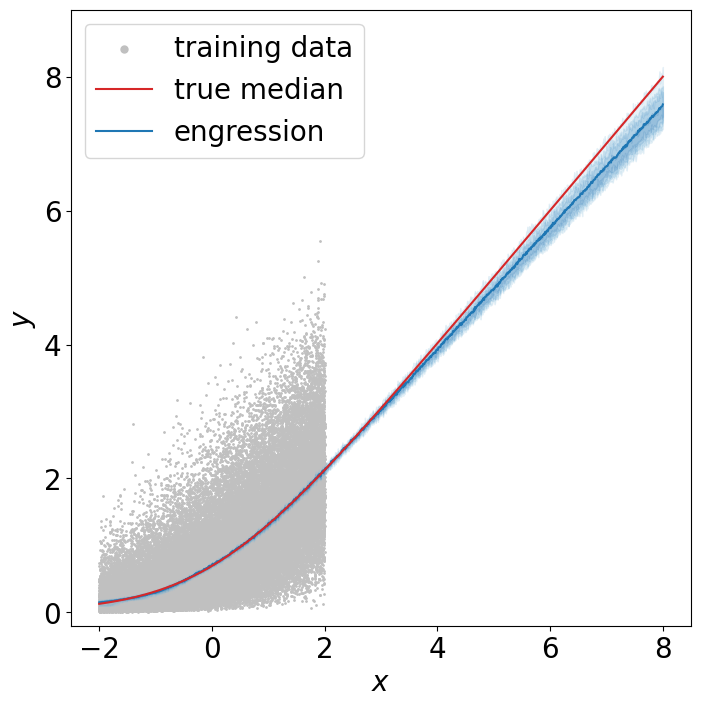} &
	\includegraphics[align=c,width=0.24\textwidth]{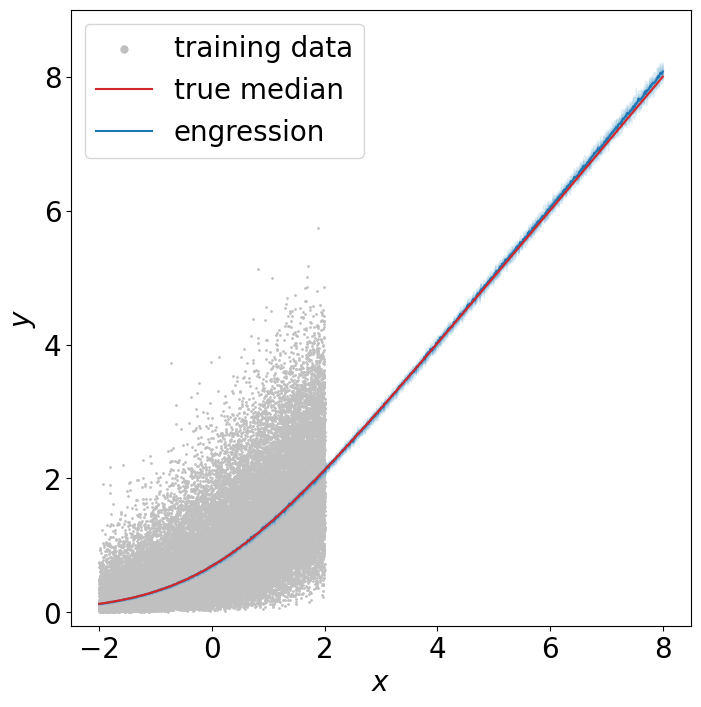}\\
	\rotatebox[origin=c]{90}{\small{Conditional mean}}&
	\includegraphics[align=c,width=0.24\textwidth]{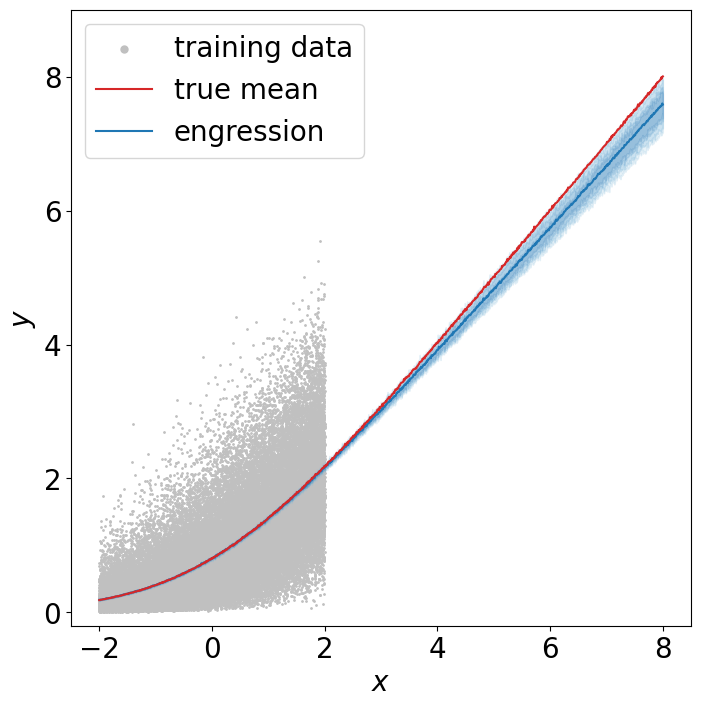} &
	\includegraphics[align=c,width=0.24\textwidth]{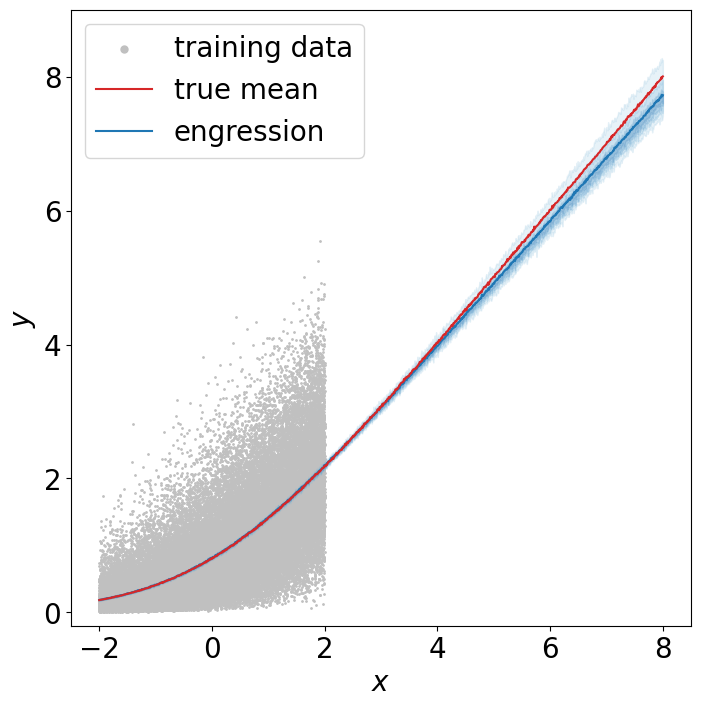} &
	\includegraphics[align=c,width=0.24\textwidth]{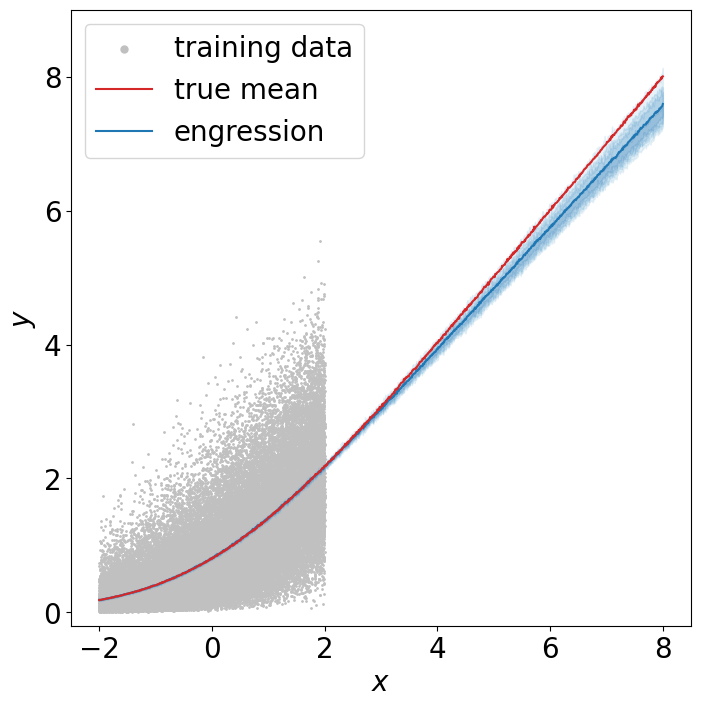} &
	\includegraphics[align=c,width=0.24\textwidth]{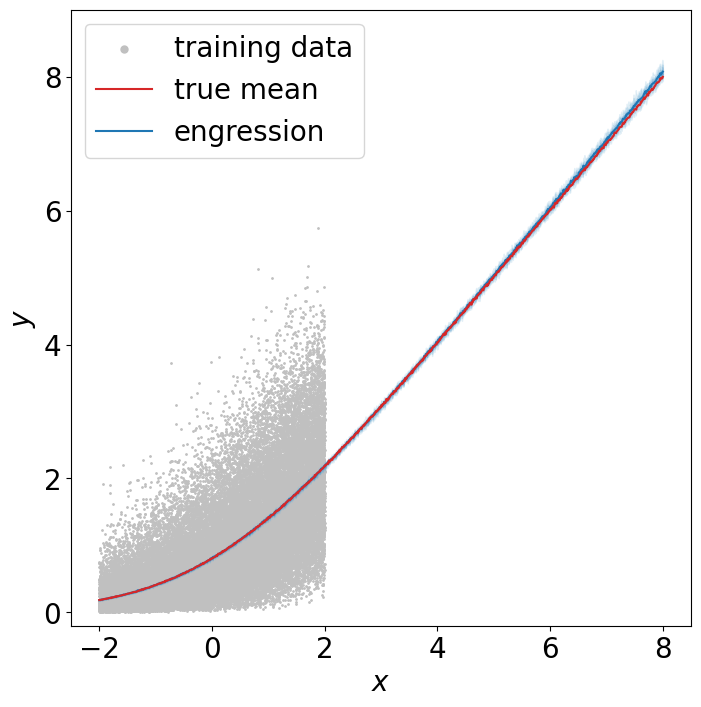}
\end{tabular}
\caption{Performance of engression with other loss functions on \texttt{softplus}. }
\label{fig:simu_visual_app_softplus}
\end{figure}

\begin{figure}
\centering
\begin{tabular}{@{}c@{}c@{}c@{}c@{}c@{}}
	&\small{Gaussian kernel} & \small{Laplace kernel} & \small{Inverse multiquadric kernel} & \small{KL GAN}\\
	\rotatebox[origin=c]{90}{\small{Conditional median}}&
	\includegraphics[align=c,width=0.24\textwidth]{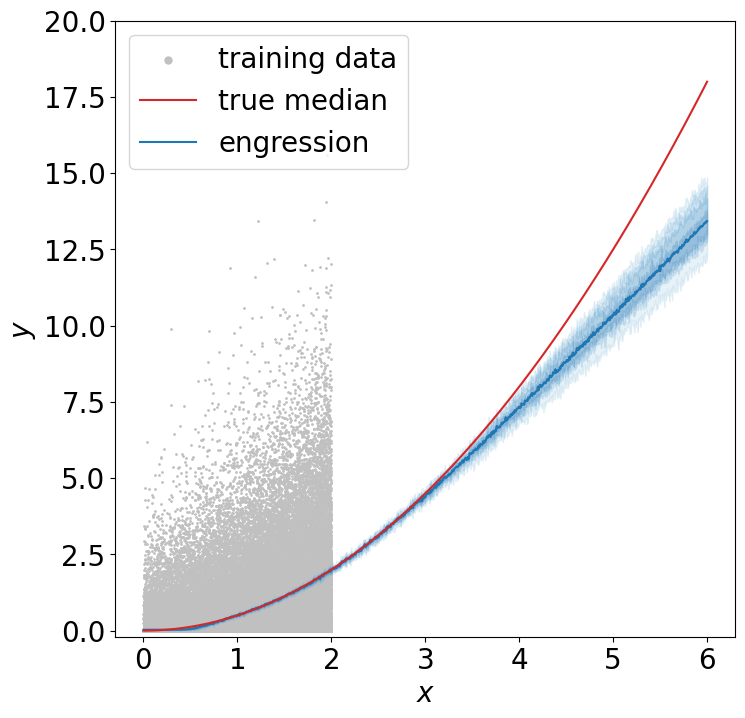} &
	\includegraphics[align=c,width=0.24\textwidth]{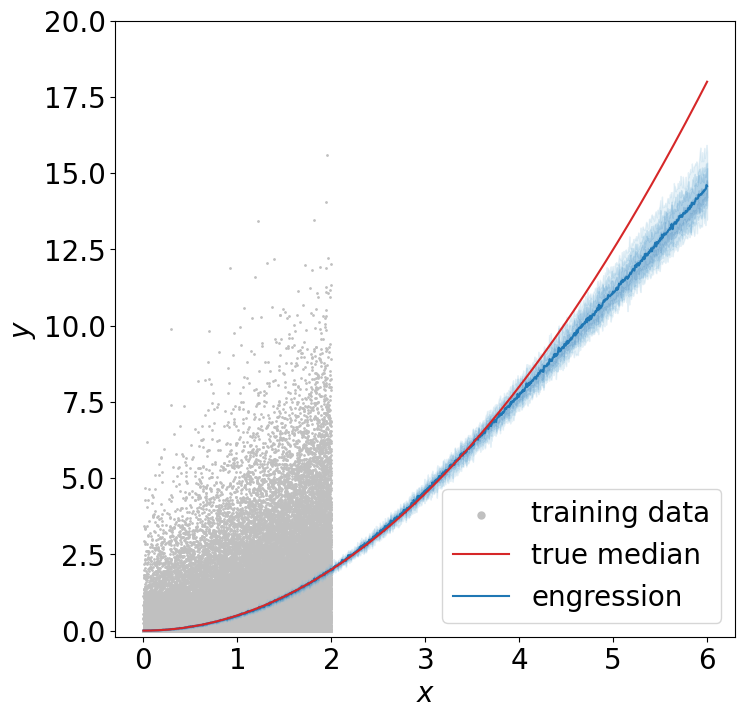} &
	\includegraphics[align=c,width=0.24\textwidth]{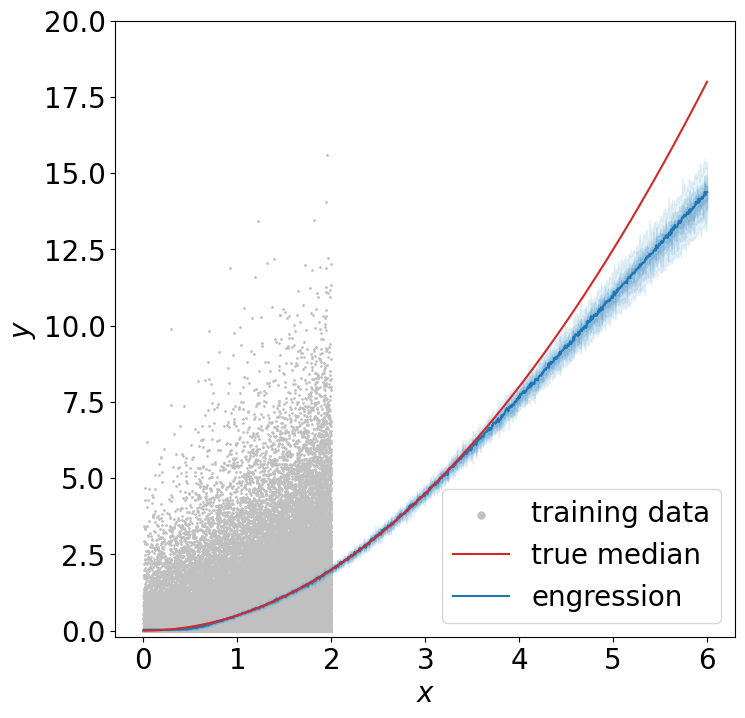} &
	\includegraphics[align=c,width=0.24\textwidth]{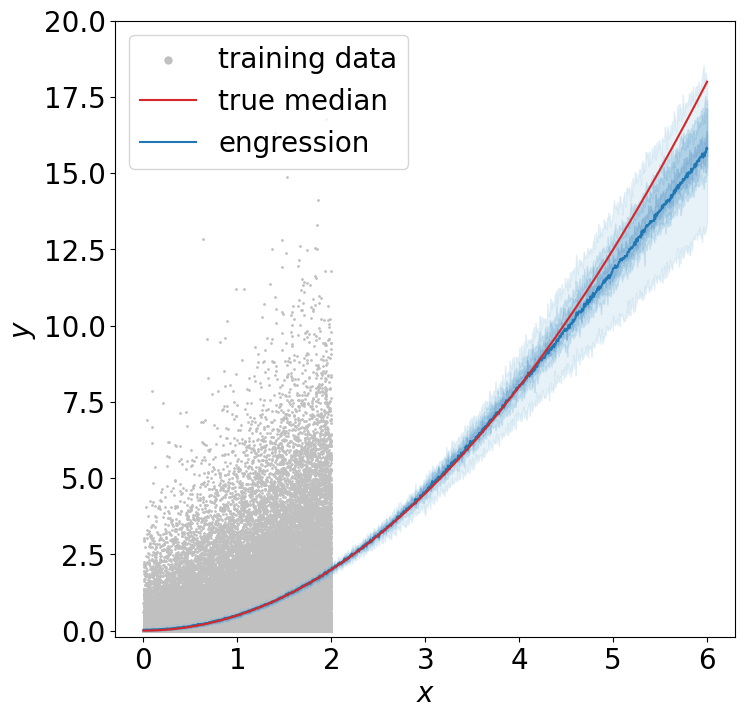}\\
	\rotatebox[origin=c]{90}{\small{Conditional mean}}&
	\includegraphics[align=c,width=0.24\textwidth]{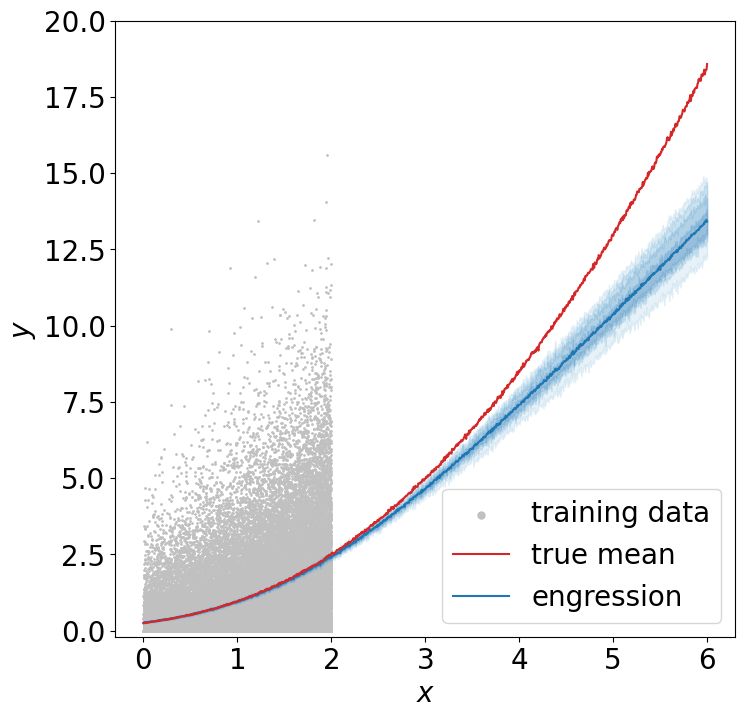} &
	\includegraphics[align=c,width=0.24\textwidth]{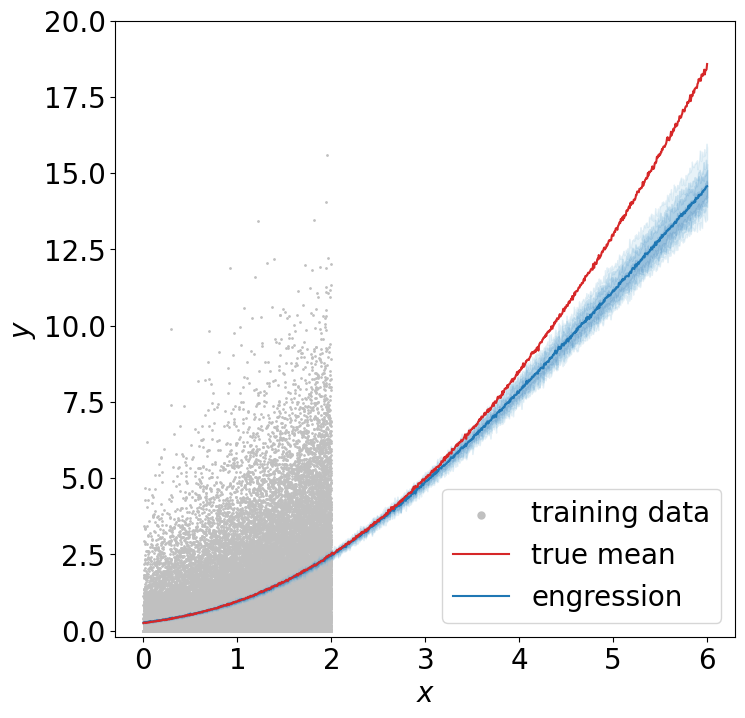} &
	\includegraphics[align=c,width=0.24\textwidth]{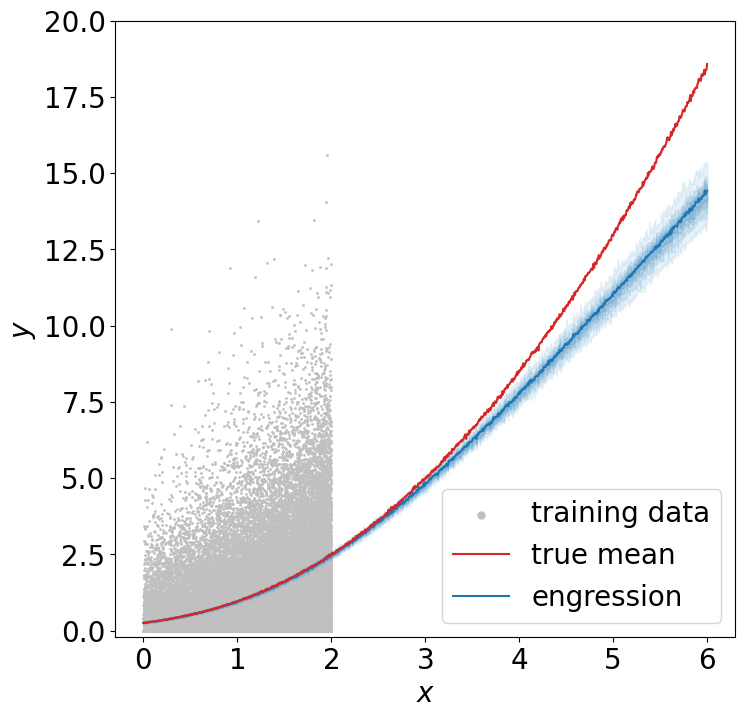} &
	\includegraphics[align=c,width=0.24\textwidth]{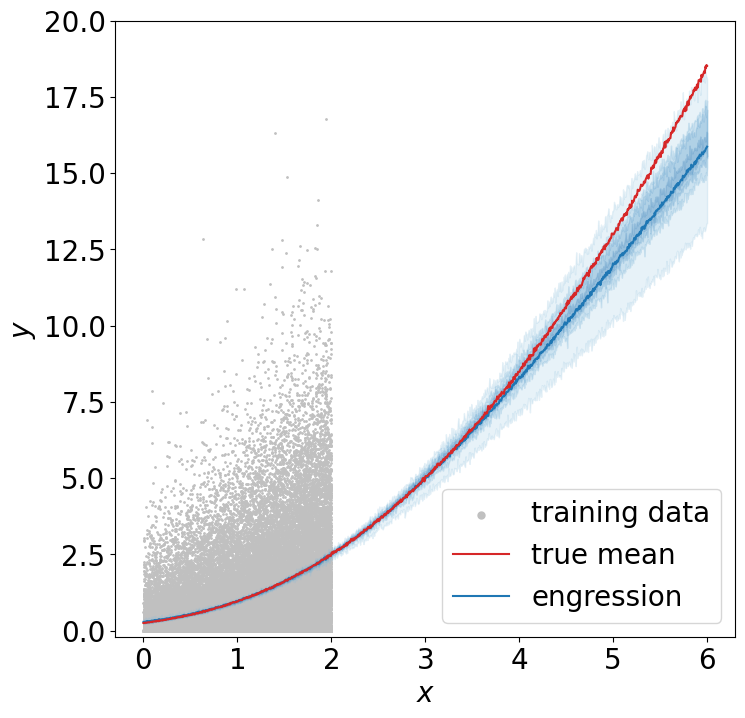}
\end{tabular}
\caption{Performance of engression with other loss functions on \texttt{square}. }
\label{fig:simu_visual_app_square}
\end{figure}

\begin{figure}
\centering
\begin{tabular}{@{}c@{}c@{}c@{}c@{}c@{}}
	&\small{Gaussian} & \small{Laplace} & \small{Inverse multiquadric} & \small{KL GAN}\\
	\rotatebox[origin=c]{90}{\small{Conditional median}}&
	\includegraphics[align=c,width=0.24\textwidth]{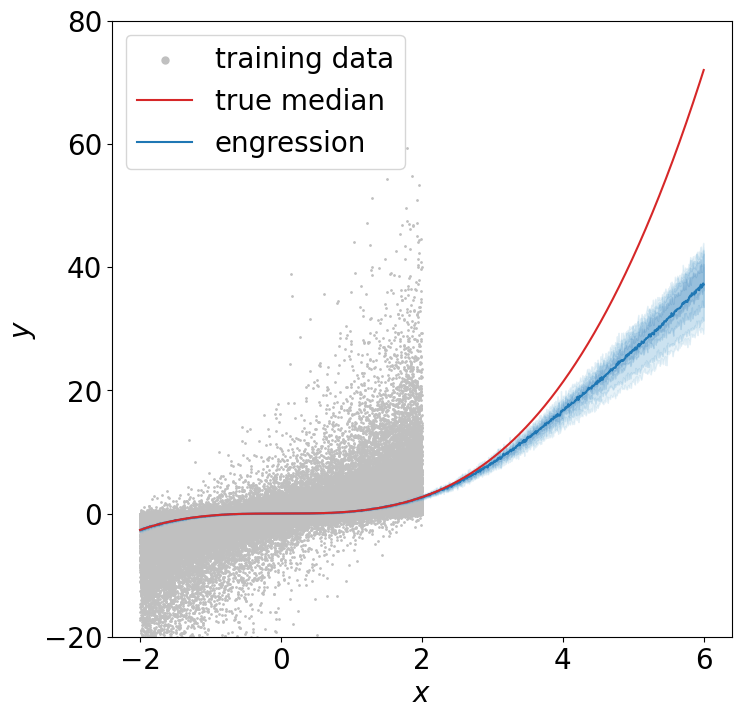} &
	\includegraphics[align=c,width=0.24\textwidth]{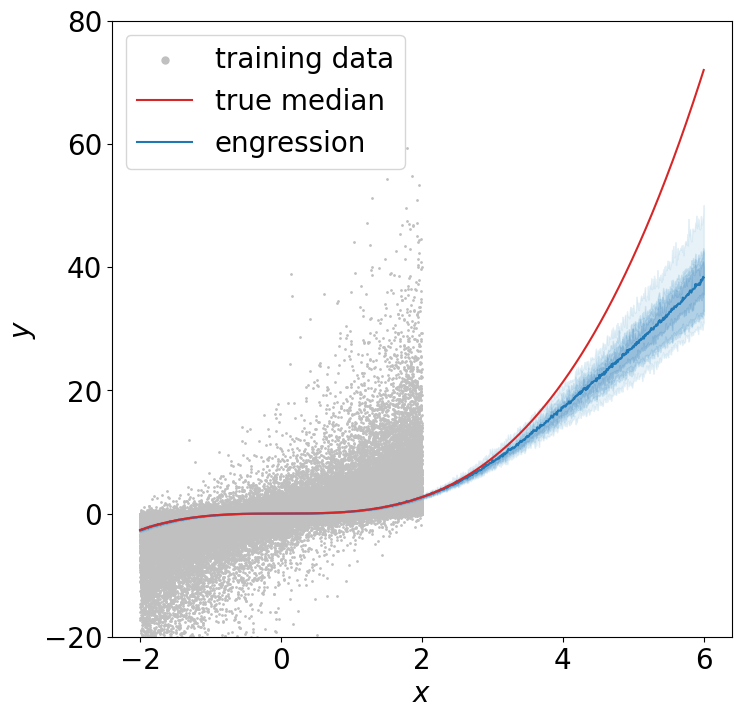} &
	\includegraphics[align=c,width=0.24\textwidth]{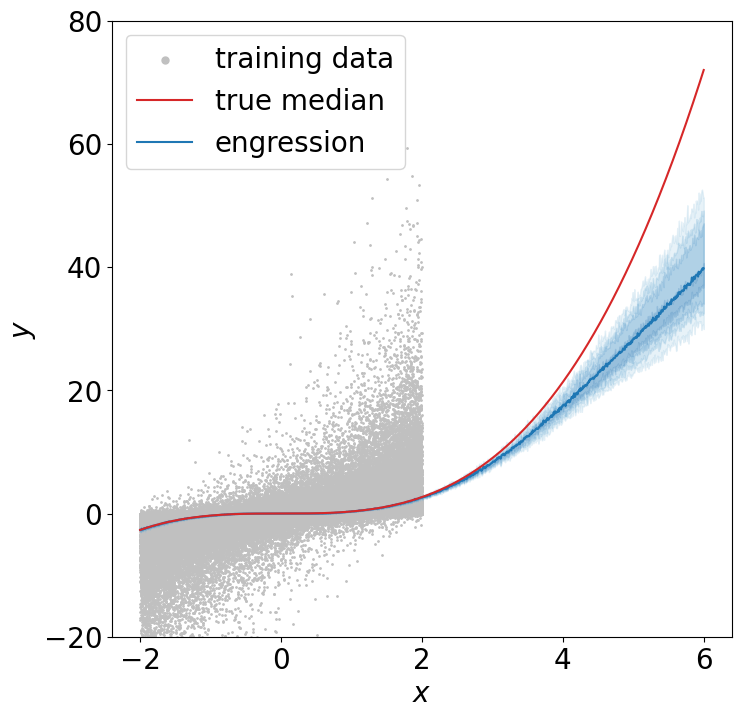} &
	\includegraphics[align=c,width=0.24\textwidth]{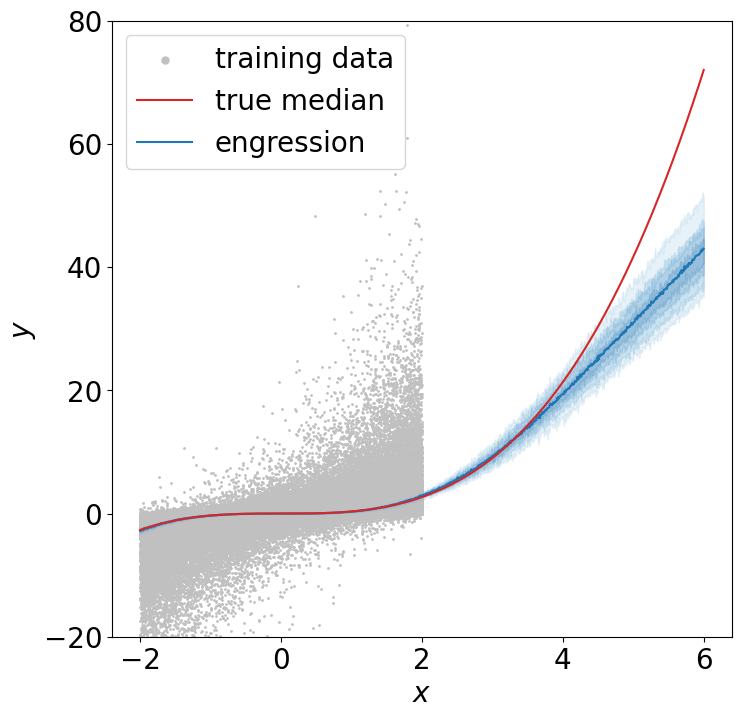}\\
	\rotatebox[origin=c]{90}{\small{Conditional mean}}&
	\includegraphics[align=c,width=0.24\textwidth]{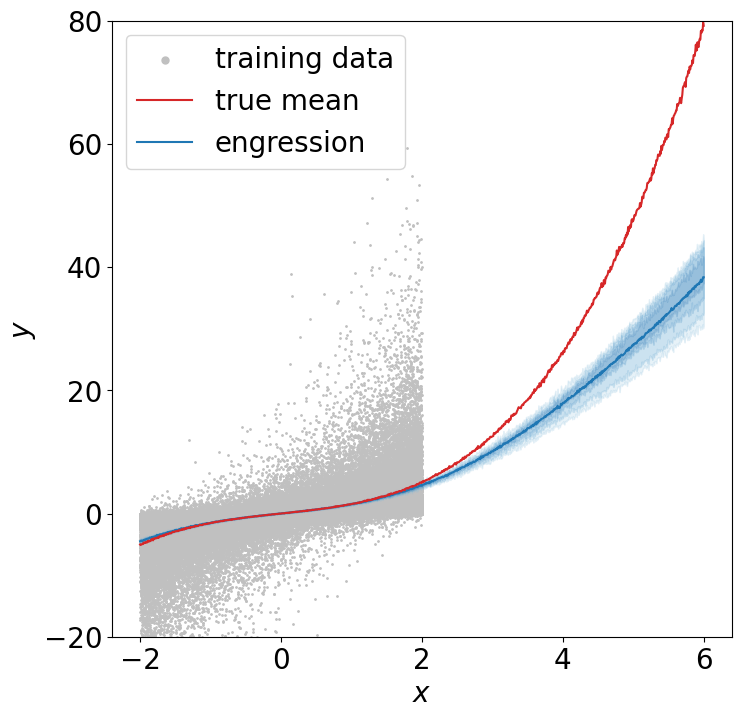} &
	\includegraphics[align=c,width=0.24\textwidth]{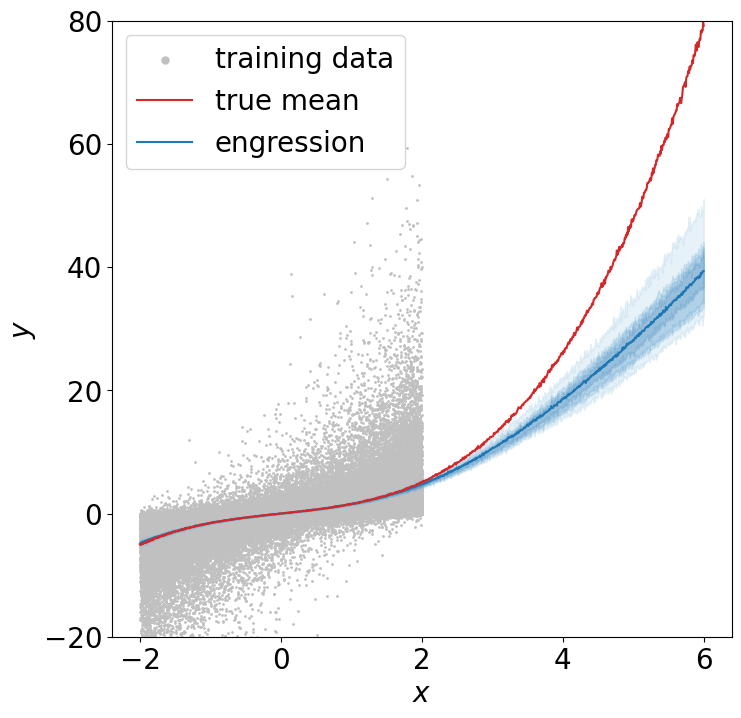} &
	\includegraphics[align=c,width=0.24\textwidth]{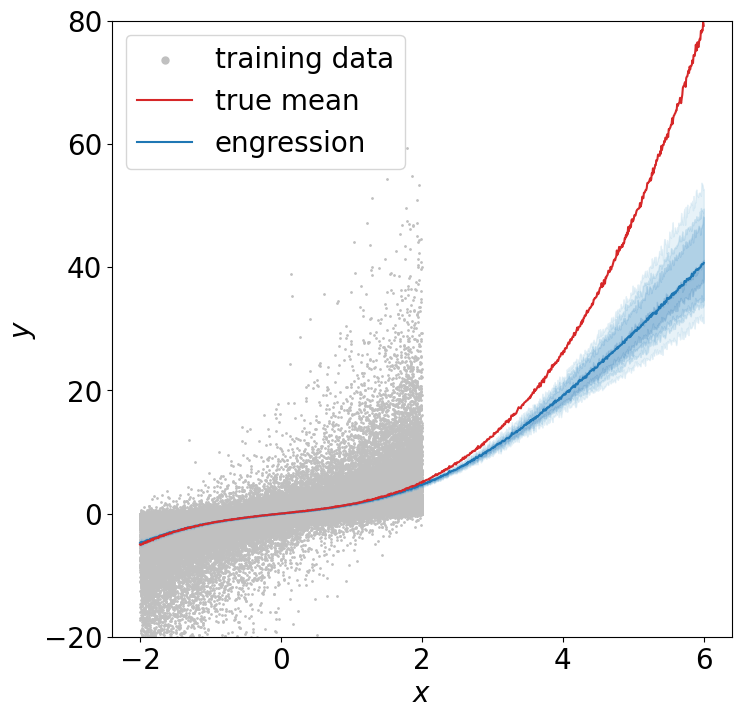} &
	\includegraphics[align=c,width=0.24\textwidth]{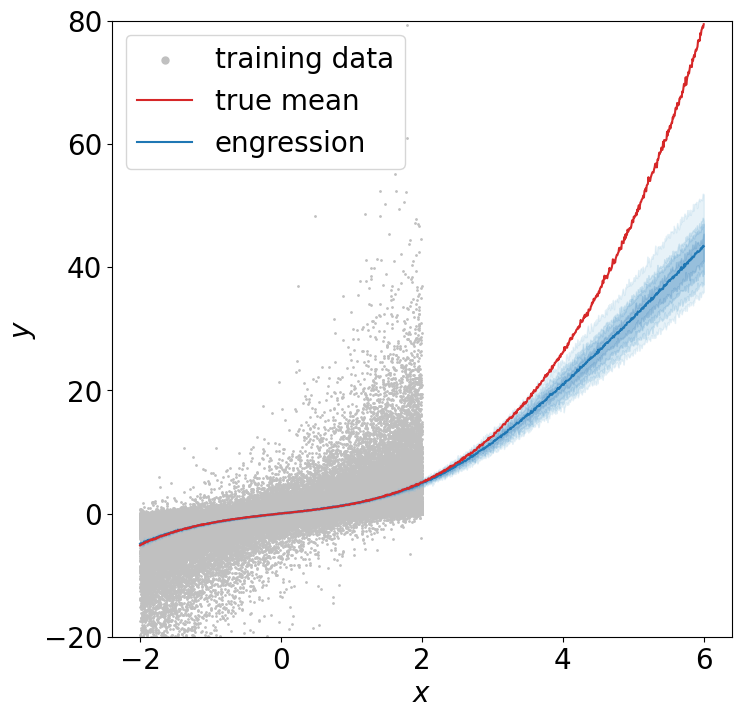}
\end{tabular}
\caption{Performance of engression with other loss functions on \texttt{cubic}. }
\label{fig:simu_visual_app_cubic}
\end{figure}

\begin{figure}
\centering
\begin{tabular}{@{}c@{}c@{}c@{}c@{}c@{}}
	&\small{Gaussian} & \small{Laplace} & \small{Inverse multiquadric} & \small{KL GAN}\\
	\rotatebox[origin=c]{90}{\small{Conditional median}}&
	\includegraphics[align=c,width=0.24\textwidth]{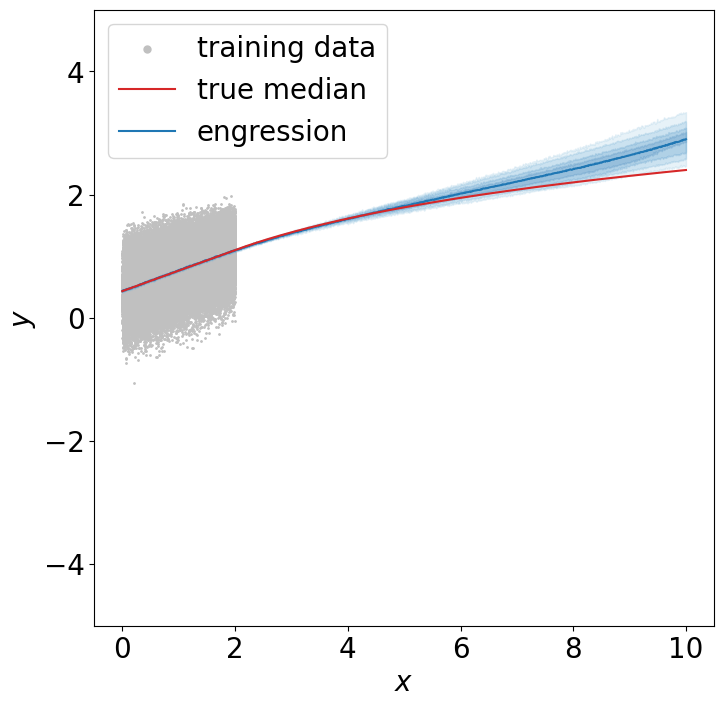} &
	\includegraphics[align=c,width=0.24\textwidth]{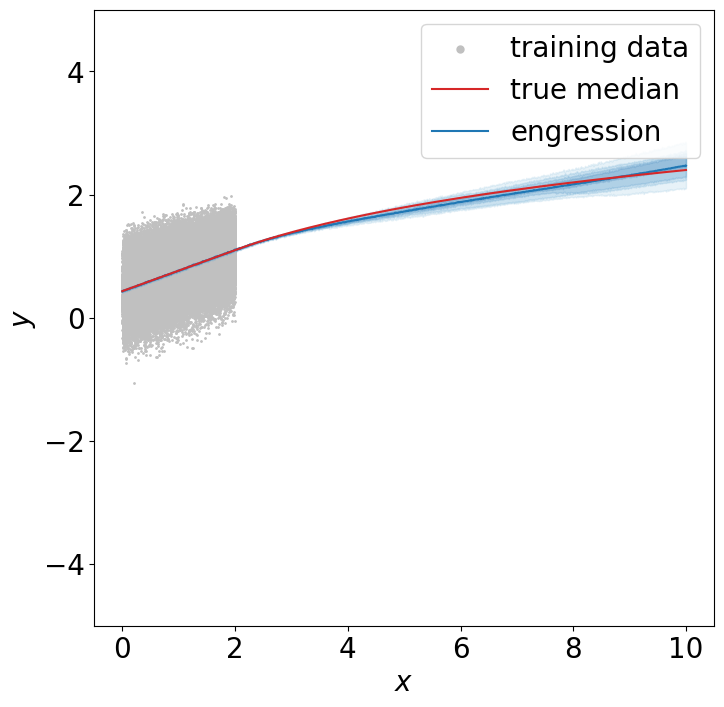} &
	\includegraphics[align=c,width=0.24\textwidth]{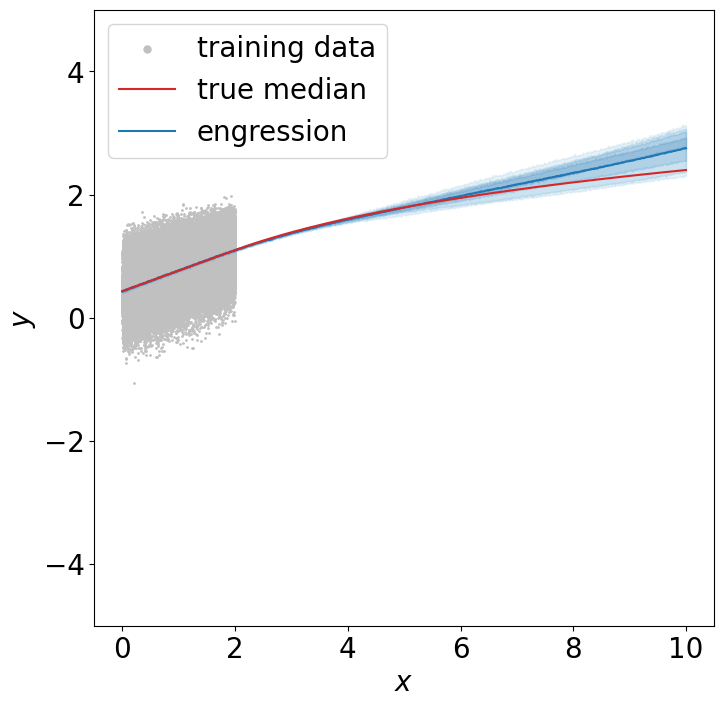} &
	\includegraphics[align=c,width=0.24\textwidth]{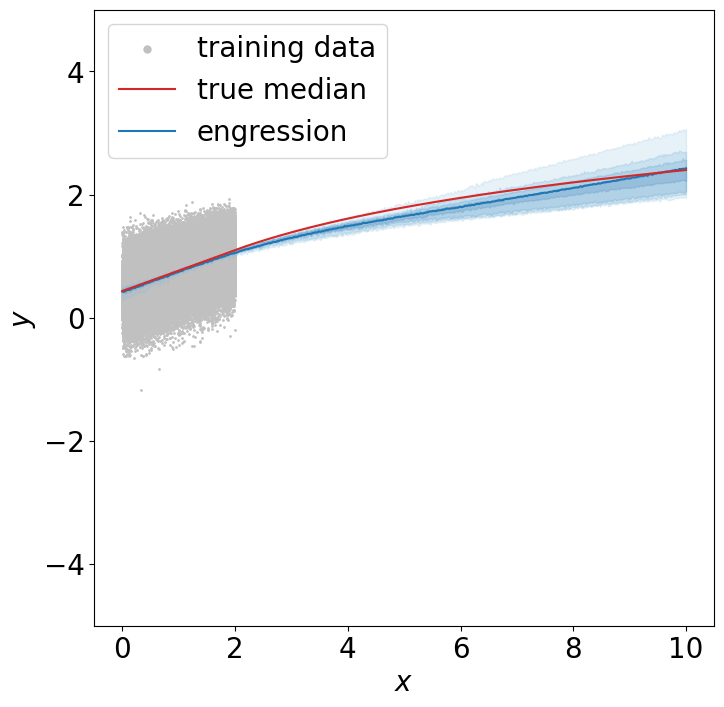}\\
	\rotatebox[origin=c]{90}{\small{Conditional mean}}&
	\includegraphics[align=c,width=0.24\textwidth]{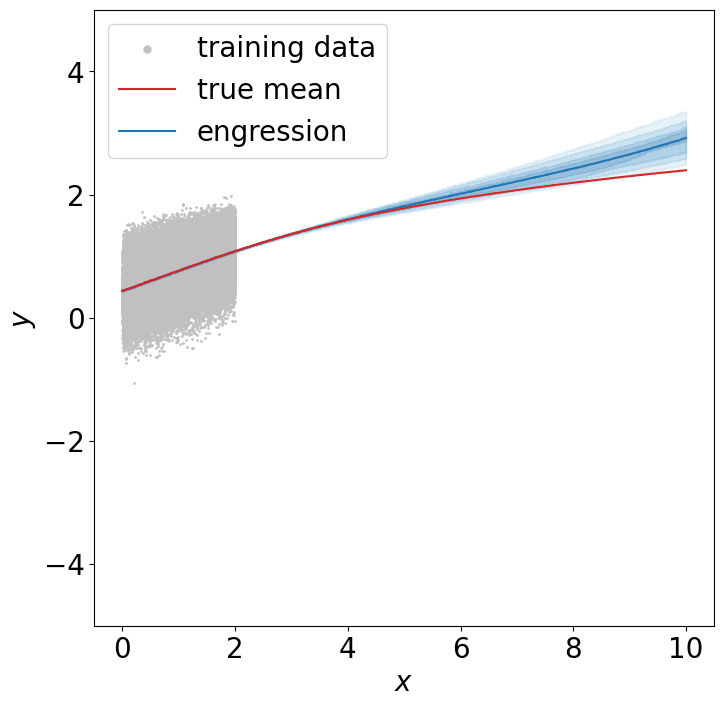} &
	\includegraphics[align=c,width=0.24\textwidth]{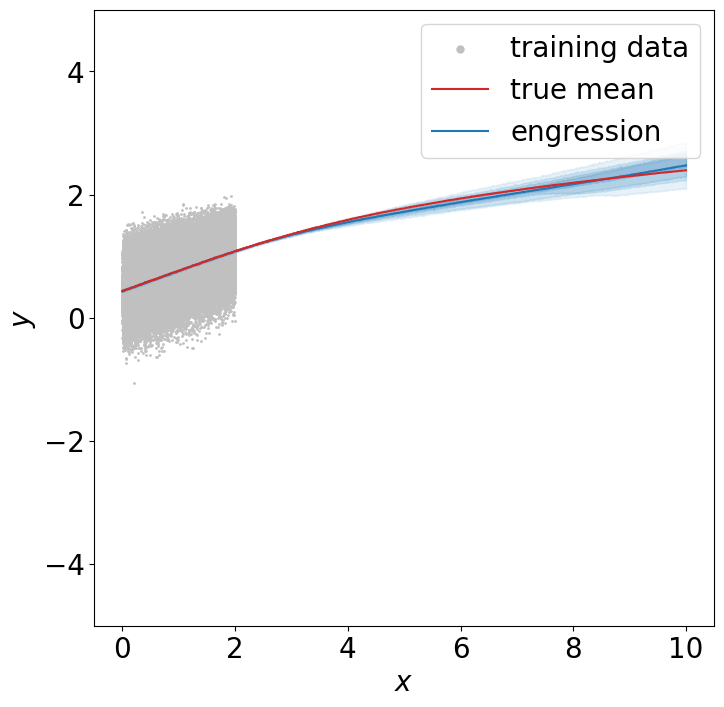} &
	\includegraphics[align=c,width=0.24\textwidth]{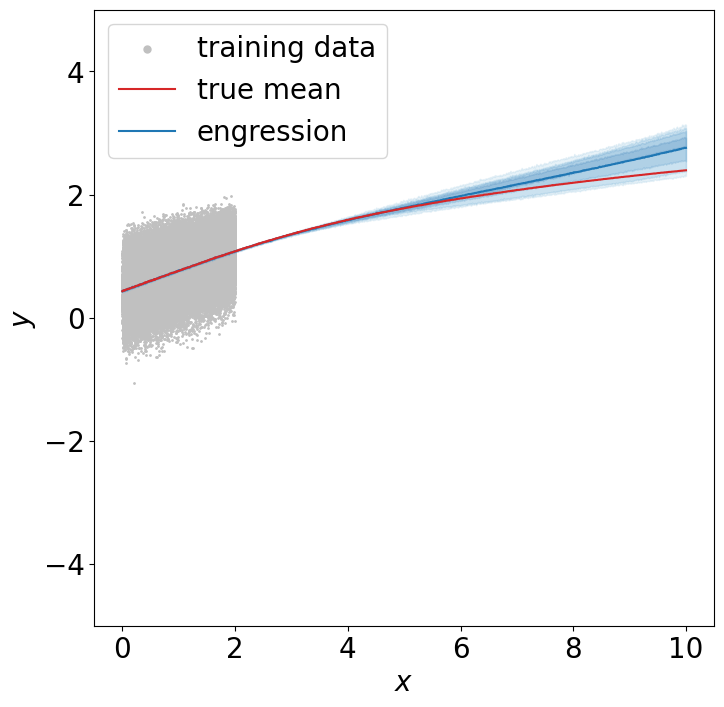} &
	\includegraphics[align=c,width=0.24\textwidth]{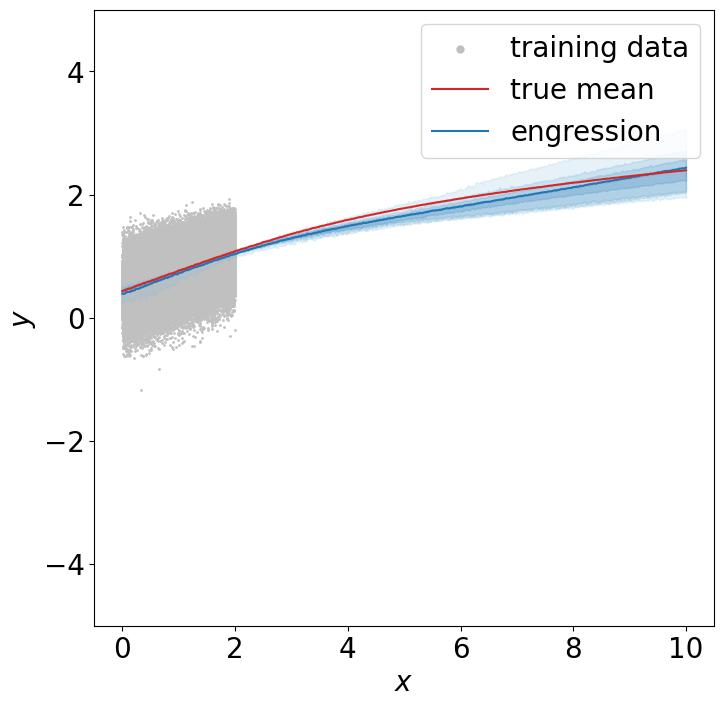}
\end{tabular}
\caption{Performance of engression with other loss functions on \texttt{log}. }
\label{fig:simu_visual_app_log}
\end{figure}

\section{Experimental details}\label{app:exp_detail}
We first describe the setup that we adopt generally across all experiments. For estimating the energy loss, at each gradient descent step, we use two generated samples for each observation of the covariate. We generate new samples of $\varepsilon$ for different steps, leading to a large generated sample size $m$ in total in the empirical loss in \eqref{eq:eng_emp}. Regarding the neural network architecture, we always adopt multilayer perceptrons (MLPs). In Figure~\ref{fig:nn_archi}, we illustrate the structure for one layer that we typically adopt for all input and intermediate layers (if the NN contains more than two layers); the output layer is always a single dense (fully-connected) layer. For $L_1$ and $L_2$ regression, we use deterministic networks, where each layer follows the same architecture in Figure~\ref{fig:nn_archi} without concatenating any noise. For optimisation, we adopt the Adam optimiser~\citep{kingma2014adam} with default values for the beta parameters on the full batch of data in the experiments reported in this paper, while mini-batch gradient descent is implemented as well in our software. Throughout all experiments, we always keep the same NN architecture and optimisation hyperparameter for engression and regression. For engression, we typically set the noise dimension to $d_n=100$ unless stated otherwise. 


\begin{figure}
\centering
\includegraphics[page=1, clip, trim=20cm 12.5cm 21cm 12.5cm, width=0.6\textwidth]{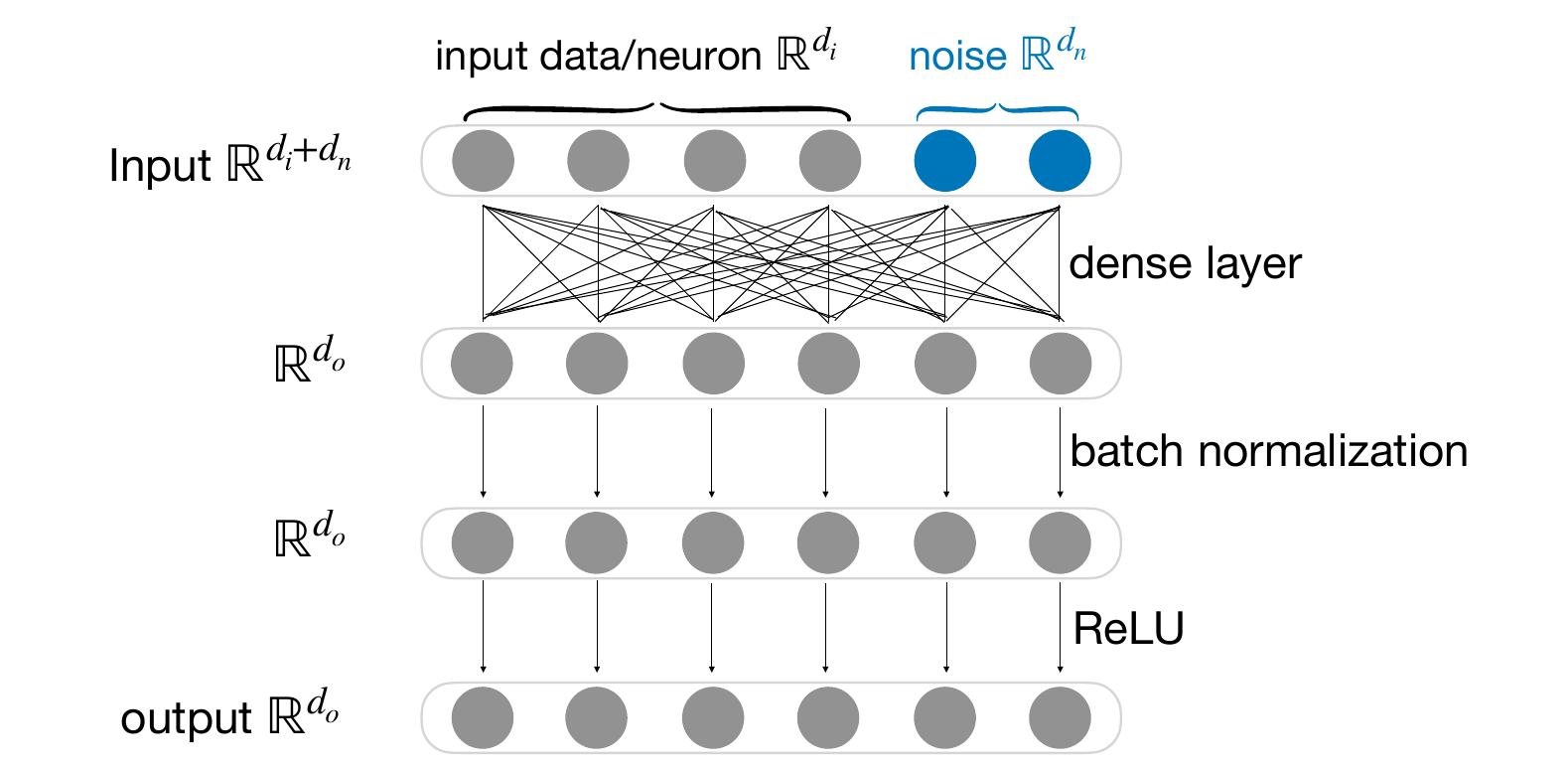}
\caption{Architecture for one layer with input dimension $d_i$, noise dimension $d_n$ and output dimension $d_o$.}\label{fig:nn_archi}
\end{figure}

\subsection{Simulations}
For \texttt{softplus} and \texttt{square}, we simulate 50k training data and use NNs with one layer and 100 hidden neurons (i.e.\ one input layer as shown in Figure~\ref{fig:nn_archi} followed by a dense output layer). For \texttt{cubic} and \texttt{log}, we simulate 100k training data and use NNs with 3 layers and 100 hidden neurons each.  
We adopt the Adam optimiser with a learning rate of $5\times10^{-2}$ and train all the models for 10k iterations.

For kernel losses, we use exactly the same training setups as with the energy loss. For GAN, for numerical stability, we use a learning rate $1\times10^{-3}$ for the model (generator) parameter $\theta$; at each gradient step of the pre-ANM, the discriminator is updated for 10 steps using the Adam optimiser with a learning rate of $5\times10^{-3}$. The discriminator is an MLP with exactly the same architecture as the generator except no noise. All models are trained on the full batch for 10k iterations.

\subsection{Real data}

\subsubsection{Data sets and preprocessing}

We adopt 6 real, public data sets from various domains:
\begin{itemize}
	\item \textbf{GCM}\quad The GCM data set consists of 13,590 observations for three variables: daily global mean temperature (in Kelvin) obtained from the CanESM2 global circulation climate (GCM) model \citep{kushner2018canadian} that covers the years 1950-2100 and regridded to a regular 2.5 degree-grid, the global annual mean radiation anomaly (in $W/m^2$) that summarises the Earth's energy imbalance (e.g.\ from changing CO$_2$ concentrations), and the time stamp (integer). We consider predicting the temperature with either the radiation or time and the reverse direction. 
	\item \textbf{Abrupt4x}\quad The Abrupt4x data set contains temperature data from the CESM2 global circulation climate model \citep{danabasoglu2020community} that spans 999 years. At the beginning, the CO$_2$ concentration in the atmosphere is multiplied by a factor of 4 compared to a pre-industrial level and then the climate slowly adjusts to a new and warmer equilibrium afterwards. We take the global mean temperature either monthly (using all years) or daily (using the first 209 years due to the huge sample size). We use the time stamp to predict either the monthly or daily temperature and predict the time using the temperatures.
	\item \textbf{NHANES}\quad National Health and Nutrition Examination Survey (NHANES) is a nationally representative study conducted by the US Centers for Disease Control and Prevention (CDC) to assess the health and nutritional status of the United States population. In total, there were 20,470 study participants who participated in the NHANES 2003-2006 study. The original data set is obtained from the R package rnhanesdata~\citep{leroux2019organizing} and processed following a similar pipeline as described in \citet{cui2021additive}. For this experiment, we exclude participants who did not wear a physical activity monitor or had high-quality accelerometry data for less than 3 days and with missing BMI and age information. The final data set contains 10,914 participants. We focus on 4 variables: the total activity count (a unit measuring physical activity intensity), sedentary time (in minutes), body mass index, and age (in months). We do pairwise prediction for all pairs of variables in both directions.
	\item \textbf{Single-cell}\quad The single-cell RNA sequencing data set was published by \citet{replogle2022mapping} who performed genome-scale Perturb-seq targeting all expressed genes with CRISPR perturbations across human cells. We utilise the data set on the RPE1 cells. After removing missing data and excluding genes with low expression, we select 6 pairs of genes with the highest correlation coefficient, with a sample size of around 11,400 for each pair. We do pairwise predictions in both directions.
	\item \textbf{Birth}\quad This data set is obtained from the CDC Vital Statistics Data Online Portal (\url{https: //www.cdc.gov/nchs/data_access/vitalstatsonline.htm}) and contains the information about around 3.8 million births in 2018, from which we remove all missing entries and randomly sample 30,000 individuals for our experiment. We focus on predicting the birth weight (in grams) from the pregnancy duration (in months, discrete). 
	\item \textbf{Air}\quad The air quality data set \citep{misc_air_quality_360} is obtained from the UCI machine learning repository. It contains hourly averaged concentrations for five pollutants: CO, non-methane hydrocarbons concentration (NMHC), Total Nitrogen Oxides (NOx), Nitrogen Dioxide (NO2), O3, and three meteorological covariates: temperature in degrees Celsius, relative humidity (in percentage) and absolute humidity (AH) for  8991 hours after removing missing values. We do pairwise predictions for all five pollutants, as well as NO2 -- temperature, NO2 -- AH, and AH -- temperature in both directions.
\end{itemize}

In the large-scale experiments for univariate prediction, as described in Section~\ref{sec:exp_uni}, for each prediction task (i.e.\ a specific pair of the response and predictor), we partition the training and test data at the 0.3, 0.4, 0.5, 0.6, and 0.7 quantiles of the predictor and designate the smaller or larger portion as the training set. In Figure~\ref{fig:data_partition}, we demonstrate one example of data setup on the air quality data. 
After training/test partitioning, we normalize the full data set based on the sample mean and standard deviation of the training set. 

\begin{figure}
\centering
\begin{tabular}{@{}c@{}c@{}c@{}c@{}c@{}c@{}}
	\includegraphics[width=0.165\textwidth]{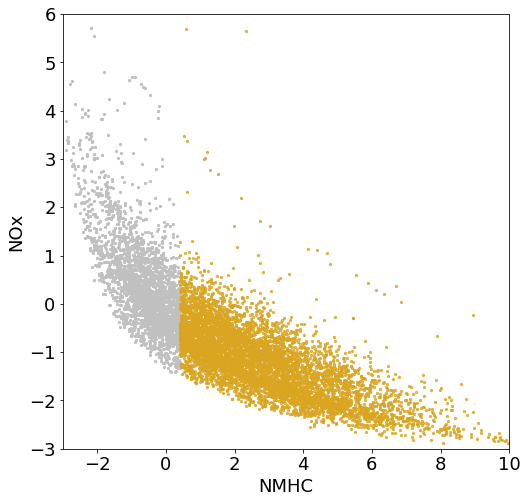} & 
	\includegraphics[width=0.165\textwidth]{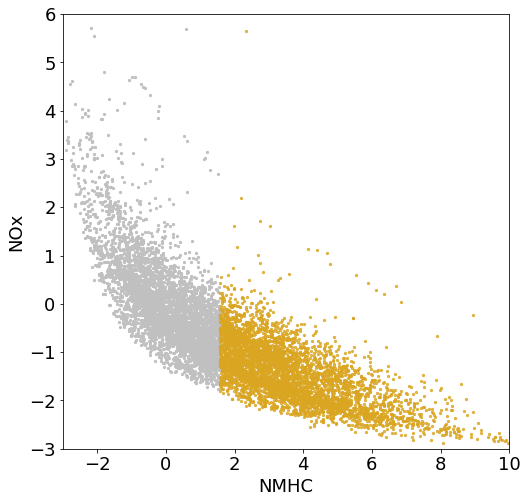} & 
	\includegraphics[width=0.165\textwidth]{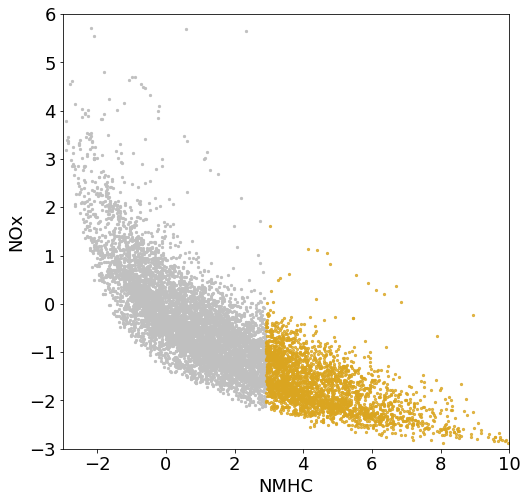} &
	\includegraphics[width=0.165\textwidth]{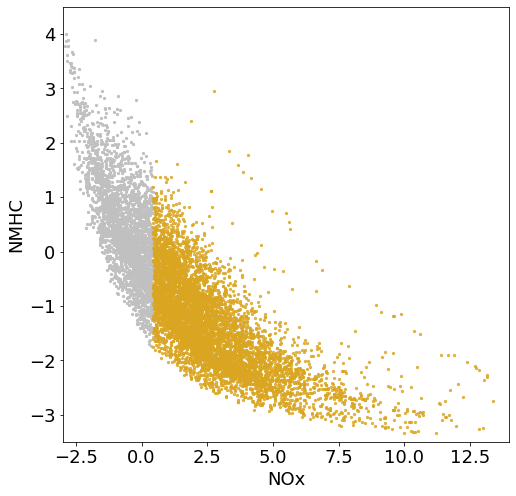} & 
	\includegraphics[width=0.165\textwidth]{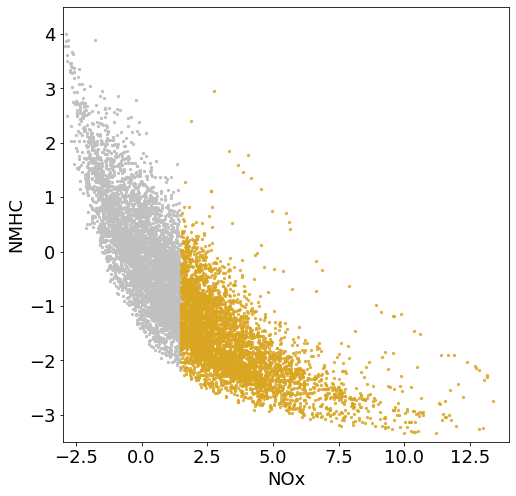} & 
	\includegraphics[width=0.165\textwidth]{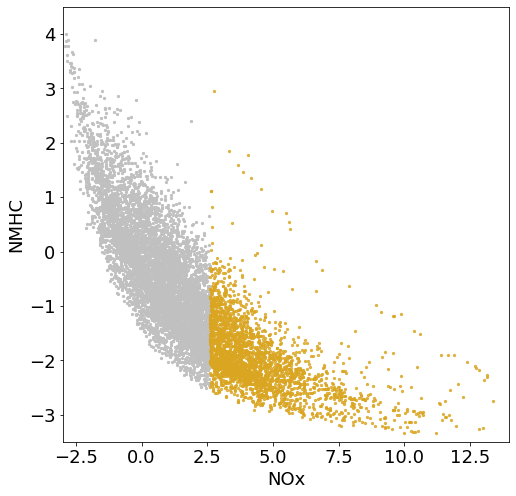} \\
	\small{(a) 0.3} & \small{(b) 0.5} & \small{(c) 0.7} & \small{(d) flipped 0.3} & \small{(e) flipped 0.5} & \small{(f) flipped 0.7}
\end{tabular}
\caption{An example of data setup on the air quality data. Partitioning training and test data at different quantiles of the predictor results in different levels of ``similarity" between training and test distribution, implying varying difficulty for extrapolation. Comparing (d)-(e) to (a)-(c), flipping the roles of the predictor and response leads to potentially different underlying models.}\label{fig:data_partition}
\end{figure}

\subsubsection{Hyperparameters}

For large-scale experiments on univariate prediction in Section~\ref{sec:exp_uni}, we consider various combinations of learning rates, numbers of training iterations, layers, and neurons per layer, as listed in Table~\ref{tab:hyperparam}.

For multivariate prediction in Section~\ref{sec:exp_multi}, we use exactly the same NN architectures and optimisation hyperparameters for both engression and regression and on both the air quality and NHANES data sets with any variables as the response. We use MLPs with 3 layers and 500 neurons per layer; for engression the noise dimension is set to 500 as well. We again use the Adam optimiser with a learning rate of 0.001 and train all models for 3000 iterations. For evaluation, the predictions for the conditional mean by engression are obtained by sampling 128 $\varepsilon$'s for each $x$, which is used for both visualisation and computing the squared prediction error. 

For prediction intervals in Section~\ref{sec:exp_pi}, we again keep all the hyperparameters the same for engression and regression. For GCM data, we predict the global mean temperature from the time whose relationship is relatively simple, so we use an MLP with 2 layers and 100 hidden neurons trained for 500 iterations. For the air quality and NHANES data sets where the dependence is more complex, we use MLPs with 3 layers and 100 neurons per layer trained for 1000 iterations. The noise dimension for engression is 100 as typical. The Adam optimiser is adopted with a learning rate of 0.01.

\begin{table}
    \centering
    \caption{Set of hyperparameters. We experiment with all possible (in total 18) combinations.}
    \label{tab:hyperparam}
    \vskip 0.1in
    \begin{tabular}{c|c}
    \toprule
    (num\_layer, hidden\_dim)     & (2, 100), (3, 10), (3, 100) \\\midrule
    learning rate     & $1\times10^{-3}$, $1\times10^{-2}$\\\midrule
    number of training steps & 500, 1000, 3000\\\bottomrule
    \end{tabular}
\end{table}

\subsection{Engression and pre-ANMs}\label{app:eng_preanm}
We explain how the practical implementation of the engression model aligns with the pre-ANM class. 
Based on definition \eqref{eq:preanm}, a pre-ANM class is given by 
\begin{equation}\label{eq:preanm_class}
	\{\gen(x,\varepsilon):=g(W x+h(\varepsilon))+\beta^\top x:W\in\bbR^{k\times d},\beta\in\bbR^d,g\in\cG,h\in\cH\},
\end{equation}
where $\cG$ and $\cH$ are two function classes, noise $\varepsilon$ follows a pre-defined distribution, and each model depends on a tuple $(W,\beta,g,h)$ which is unknown and to be learned. 
In fact, when parametrising $\gen(x,\varepsilon)$ by neural networks whose first operator is a linear transformation, we can write the model more explicitly as
\begin{equation*}
	\gen(x,\varepsilon)=g(Wx+V\varepsilon),
\end{equation*}
where $Wx+V\varepsilon$ is the first linear layer and $g$ is the subsequent layers. This form already aligns with pre-ANMs by adding scalable pre-specified noises to covariates. To further allow a fully learnable noise distribution, we adopt the skip connection technique \citep{he2016deep} and the model can thus be written as
\begin{equation*}
	\gen(x,\varepsilon)=g_2(g_1(W_1x+V\varepsilon) + W_2x) + \beta^\top x,
\end{equation*}
where two skip connections are adopted after $g_1$ and $g_2$, respectively. The model is then essentially the same as the pre-ANM class \eqref{eq:preanm_class}.

\end{document}